\documentclass[11pt]{article}
\usepackage{fullpage}
\usepackage{amsmath,amssymb}\usepackage{amsfonts,amsthm,mathrsfs,graphicx}
\usepackage{mathpazo}
\usepackage{algorithm2e}
\usepackage{mathtools}

\def\showauthornotes{0}

\usepackage[l2tabu, orthodox]{nag}

\usepackage{xspace,enumerate}

\usepackage[dvipsnames]{xcolor}
\usepackage{tikz-cd}

\usepackage[T1]{fontenc}
\usepackage[full]{textcomp}

\usepackage[american]{babel}

\usepackage{mathtools}

\usepackage{stmaryrd}

\usepackage{amsthm}

\newtheorem{theorem}{Theorem}[section]
\newtheorem*{theorem*}{Theorem}

\newtheorem{proposition}[theorem]{Proposition}
\newtheorem*{proposition*}{Proposition}
\newtheorem{lemma}[theorem]{Lemma}
\newtheorem*{lemma*}{Lemma}
\newtheorem{corollary}[theorem]{Corollary}
\newtheorem{corollary*}{Corollary}

\newtheorem*{conjecture*}{Conjecture}

\newtheorem*{fact*}{Fact}

\newtheorem*{hypothesis*}{Hypothesis}

\theoremstyle{definition}
\newtheorem{definition}[theorem]{Definition}

\newtheorem{claim}[theorem]{Claim}
\newtheorem*{claim*}{Claim}

\theoremstyle{remark}
\newtheorem{remark}[theorem]{Remark}
\newtheorem*{remark*}{Remark}

\newtheorem*{observation*}{Observation}

\newtheorem*{rep@theorem}{\rep@title}
\newcommand{\newreptheorem}[2]{%
\newenvironment{rep#1}[1]{%
 \def\rep@title{#2 ##1}%
 \begin{rep@theorem}}%
 {\end{rep@theorem}}
}

\usepackage[varg]{pxfonts} 

\usepackage{tgpagella}
\usepackage{lmodern}

\def\showcolorlinks{1}
\ifnum\showcolorlinks=1
\usepackage[
pagebackref,
colorlinks=true,
urlcolor=blue,
linkcolor=blue,
citecolor=OliveGreen,
]{hyperref}
\fi

\ifnum\showcolorlinks=0
\usepackage[
pagebackref,
letterpaper=true,
colorlinks=false,
pdfborder={0 0 0}
]{hyperref}
\fi

\usepackage{prettyref}

\newcommand{\savehyperref}[2]{\texorpdfstring{\hyperref[#1]{#2}}{#2}}

\newrefformat{eq}{\savehyperref{#1}{\textup{(\ref*{#1})}}}
\newrefformat{lem}{\savehyperref{#1}{Lemma~\ref*{#1}}}
\newrefformat{def}{\savehyperref{#1}{Definition~\ref*{#1}}}
\newrefformat{thm}{\savehyperref{#1}{Theorem~\ref*{#1}}}
\newrefformat{cor}{\savehyperref{#1}{Corollary~\ref*{#1}}}
\newrefformat{cha}{\savehyperref{#1}{Chapter~\ref*{#1}}}
\newrefformat{sec}{\savehyperref{#1}{Section~\ref*{#1}}}
\newrefformat{app}{\savehyperref{#1}{Appendix~\ref*{#1}}}
\newrefformat{tab}{\savehyperref{#1}{Table~\ref*{#1}}}
\newrefformat{fig}{\savehyperref{#1}{Figure~\ref*{#1}}}
\newrefformat{hyp}{\savehyperref{#1}{Hypothesis~\ref*{#1}}}
\newrefformat{rem}{\savehyperref{#1}{Remark~\ref*{#1}}}
\newrefformat{item}{\savehyperref{#1}{Item~\ref*{#1}}}
\newrefformat{step}{\savehyperref{#1}{step~\ref*{#1}}}
\newrefformat{conj}{\savehyperref{#1}{Conjecture~\ref*{#1}}}
\newrefformat{fact}{\savehyperref{#1}{Fact~\ref*{#1}}}
\newrefformat{prop}{\savehyperref{#1}{Proposition~\ref*{#1}}}
\newrefformat{prob}{\savehyperref{#1}{Problem~\ref*{#1}}}
\newrefformat{claim}{\savehyperref{#1}{Claim~\ref*{#1}}}
\newrefformat{relax}{\savehyperref{#1}{Relaxation~\ref*{#1}}}
\newrefformat{red}{\savehyperref{#1}{Reduction~\ref*{#1}}}
\newrefformat{part}{\savehyperref{#1}{Part~\ref*{#1}}}

\newcommand{\Sref}[1]{\hyperref[#1]{\S\ref*{#1}}}

\usepackage{nicefrac}

\def\usemicrotype{0}
\ifnum\usemicrotype=1
\usepackage{microtype}
\fi

\ifnum\showauthornotes=1
\newcommand{\Authornote}[2]{{\sffamily\small\color{red}{[#1: #2]}}}
\newcommand{\Authornotecolored}[3]{{\small\color{#1}{[#2: #3]}}}
\newcommand{\Authorcomment}[2]{{\sffamily\small\color{gray}{[#1: #2]}}}
\newcommand{\Authorstartcomment}[1]{\sffamily\small\color{gray}[#1: }

\newcommand{\Authorfnote}[2]{\footnote{\color{red}{#1: #2}}}
\newcommand{\Authorfixme}[1]{\Authornote{#1}{\textbf{??}}}
\newcommand{\Authormarginmark}[1]{\marginpar{\textcolor{red}{\fbox{\Large #1:!}}}}
\else
\newcommand{\Authornote}[2]{}
\newcommand{\Authornotecolored}[3]{}
\newcommand{\Authorcomment}[2]{}
\newcommand{\Authorstartcomment}[1]{}

\newcommand{\Authorfnote}[2]{}
\newcommand{\Authorfixme}[1]{}
\newcommand{\Authormarginmark}[1]{}
\fi

\def\showfixme{0}
\ifnum\showfixme=1

\fi

\usepackage{boxedminipage}

\newcommand\sett[2]{\left\{ #1 \left| \; \vphantom{#1 #2} \right. #2  \right\}}
\newcommand{\set}[1]{\{#1\}}

\newcommand{\norm}[1]{\lVert#1\rVert}

\def\dim{{\textrm{dim}}}

\newcommand{\Esymb}{\mathbb{E}}
\newcommand{\Psymb}{\mathbb{P}}

\DeclareMathOperator*{\E}{\Esymb}

\DeclareMathOperator*{\ProbOp}{\Psymb}

\renewcommand{\Pr}{\ProbOp}

 \usepackage{dsfont}
\usepackage{mathrsfs}

\newcommand{\textparen}[1]{\text{(#1)}}

\ifx\because\undefined
\newcommand{\because}[1]{\textparen{because #1}}
\else
\renewcommand{\because}[1]{\textparen{because #1}}
\fi

\newcommand\precdot{\mathrel{\ooalign{$\prec$\cr
  \hidewidth\hbox{$\cdot\mkern0.5mu$}}}}
\newcommand\succdot{\mathrel{\ooalign{$\succ$\cr
  $\cdot$}}}

\newcommand{\symdiff}{\Delta}

\newcommand\bdot\bullet

\DeclareMathOperator{\poly}{poly}

\DeclareMathOperator{\supp}{supp}
\DeclareMathOperator{\dist}{dist}

\newcommand{\dt}{\diffmacro{t}}

\newcommand{\Z}{\mathbb Z}
\newcommand{\N}{\mathbb N}
\newcommand{\R}{\mathbb R}
\newcommand{\C}{\mathcal C}

\newcommand{\cC}{\mathcal C}

\newcommand{\cF}{\mathcal F}

\renewcommand{\leq}{\leqslant}
\renewcommand{\le}{\leqslant}
\renewcommand{\geq}{\geqslant}
\renewcommand{\ge}{\geqslant}

\let\epsilon=\varepsilon

\newcommand{\MYstore}[2]{%
  \global\expandafter \def \csname MYMEMORY #1 \endcsname{#2}%
}

\newcommand{\MYload}[1]{%
  \csname MYMEMORY #1 \endcsname%
}

\newcommand{\MYnewlabel}[1]{%
  \newcommand\MYcurrentlabel{#1}%
  \MYoldlabel{#1}%
}

\newcommand{\MYdummylabel}[1]{}

\newcommand{\torestate}[1]{%
  \let\MYoldlabel\label%
  \let\label\MYnewlabel%
  #1%
  \MYstore{\MYcurrentlabel}{#1}%
  \let\label\MYoldlabel%
}

\newcommand{\restatetheorem}[1]{%
  \let\MYoldlabel\label
  \let\label\MYdummylabel
  \begin{theorem*}[Restatement of \prettyref{#1}]
    \MYload{#1}
  \end{theorem*}
  \let\label\MYoldlabel
}

\newcommand{\restatelemma}[1]{%
  \let\MYoldlabel\label
  \let\label\MYdummylabel
  \begin{lemma*}[Restatement of \prettyref{#1}]
    \MYload{#1}
  \end{lemma*}
  \let\label\MYoldlabel
}

\newcommand{\restateprop}[1]{%
  \let\MYoldlabel\label
  \let\label\MYdummylabel
  \begin{proposition*}[Restatement of \prettyref{#1}]
    \MYload{#1}
  \end{proposition*}
  \let\label\MYoldlabel
}

\newcommand{\restatefact}[1]{%
  \let\MYoldlabel\label
  \let\label\MYdummylabel
  \begin{fact*}[Restatement of \prettyref{#1}]
    \MYload{#1}
  \end{fact*}
  \let\label\MYoldlabel
}

\newcommand{\restate}[1]{%
  \let\MYoldlabel\label
  \let\label\MYdummylabel
  \MYload{#1}
  \let\label\MYoldlabel
}

\newcommand{\e}{\epsilon}
\newcommand{\eps}{\epsilon}

\let\origparagraph\paragraph
\renewcommand{\paragraph}[1]{\origparagraph{#1.}}

\allowdisplaybreaks

\newcommand\F{\mathbb{F}}
\def\cosyst{\mathrm{cosyst}}
\def\coloc{\mathrm{coloc}}
\def\syst{\mathrm{syst}}
\def\cyc{\mathrm{cyc}}
\def\cocyc{\mathrm{cocyc}}

\def\F{\mathbb{F}}

\def\rest{\textrm{res}}
\def\corest{\textrm{co-res}}
\def\type{\texttt{type}}
\newcommand\remove[1]{}

\def\Id{I}
\def\mA{\mathcal{A}}

\def\im{\mathrm{im}\,}

\def\Op{\mathrm{Op}}

\def\Nb{\mathrm{Nb}}
\def\code{\mathcal{C}}
\def\mC{\mathcal{C}}

\def\sheaf{\mathcal{F}}

\def\dist{\mathsf{d}}

\def\e{\mathsf{e}}
\def\dt{\delta_\ell}
\def\y{\mathbf{y}}
\def\r{r}

\newcommand{\up}{\mathfrak{u}}
\newcommand{\down}{\mathfrak{d}}
\newcommand{\ol}[1]{\overline{#1}}

\newcommand{\conc}{\!\!\parallel\!\!}
\newcommand{\Cay}{\mathrm{Cay}}
\newcommand{\mF}{\mathcal{F}}

\renewcommand{\dim}{\mathrm{dim}}
\newcommand{\hp}{h^\perp}

\title{Expansion of
high-dimensional cubical complexes\\[2mm] \Large with application to quantum
locally testable codes}

\author{Irit Dinur\thanks{Department of Applied Math and Computer Science, The Weizmann Institute of Science. Email: \texttt{irit.dinur@weizmann.ac.il}.} \and Ting-Chun Lin\thanks{Department of Physics, University of California San Diego, CA, and Hon Hai Research Institute, Taipei. Email: \texttt{til022@ucsd.edu}.} \and Thomas Vidick\thanks{Department of Applied Math and Computer Science, The Weizmann Institute of Science.  Email: \texttt{thomas.vidick@weizmann.ac.il}.}}

\begin{document}

\maketitle
\begin{abstract}
We introduce a high-dimensional cubical complex, for any dimension $t\in\mathbb{N}$, and apply it to the design of quantum locally testable codes. Our complex is a natural generalization of the constructions by Panteleev and Kalachev and by Dinur et.\ al of a square complex (case $t=2$), which have been applied to the design of classical locally testable codes (LTC) and quantum low-density parity check codes (qLDPC)  respectively.

We turn the geometric (cubical) complex into a chain complex by relying on constant-sized local codes $h_1,\ldots,h_t$ as gadgets. A recent result of Panteleev and Kalachev on existence of tuples of codes that are product expanding enables us to prove lower bounds on the cycle and co-cycle expansion of our chain complex.

For $t=4$ our construction gives a new family of ``almost-good'' quantum LTCs --- with constant relative rate, inverse-polylogarithmic relative distance and soundness, and constant-size parity checks. Both the distance of the quantum code and its local testability are proven directly from the cycle and co-cycle expansion of our chain complex.
\end{abstract}

\newpage
\tableofcontents

\section{Introduction}

Expander graphs are bounded-degree graphs with strong connectivity, or more precisely \emph{expansion}, properties. Explicit constructions of expander graphs are non-trivial, but, by now, abound; their use is ubiquitous in algorithms design, complexity theory, combinatorics, and other areas. High-dimensional expanders generalize the expansion requirements of expander graphs to higher-dimensional structures, composed of vertices and edges as well as higher-dimensional faces. Constructions of high-dimensional expanders (HDX) are difficult and comparatively few techniques are known. In recent years HDX have found impactful applications to randomized algorithms, complexity theory, and the design of error-correcting codes, among others. In addition to their intrinsic interest as combinatorial/geometric objects, the increasing range of applications further motivates the study of HDX.

In this paper we introduce a natural generalization of a family of two-dimensional (graphs are one-dimensional) expanders introduced recently and independently in~\cite{DELLM,PK2} and used to simultaneously construct the first good classical locally testable codes (LTC) \cite{DELLM, PK2} and the first qood quantum low-density parity check codes (qLDPC) \cite{PK2}. We extend this construction to arbitrary dimensions and prove lower bounds on its \emph{cycle} and \emph{co-cycle} expansion, which we define below. The bounds that we obtain depend on underlying parameters of the complex, namely the (spectral) expansion of an associated family of Cayley graphs and a certain robustness-like property of an associated family of constant-sized \emph{local codes}. The main application of our construction, which has been our primary motivation, is towards the construction of quantum locally testable codes (qLTC). By instantiating our complex using a suitable abelian lift of an expander, as described in~\cite{JeronimoMO0T22}, and leveraging recent results on product expansion of tuples of random codes over large fields~\cite{PK-future}, we obtain the following result.

\begin{theorem}[Informal, see Corollary \ref{cor:main-code-explicit}]\label{thm:qltc}
There exists an explicit family of $[n,k,d]$ quantum LDPC codes where $n\to\infty$ and $k=\Omega(n),d=\Omega(n/(\log n)^3)$, such that the codes are CSS codes with parity checks of weight $O(1)$ and are qLTC with soundness $\rho=\Omega(1/(\log n)^3)$.
\end{theorem} 

We refer to the main text for the definition of a quantum CSS code and the soundness of a qLTC. In the theorem, by ``explicit'' we mean that the parity check matrix for the code of dimension $n$ in our family can be computed in time polynomial in $n$.
We note that the soundness parameter in the theorem can be improved to a constant by allowing a larger check weight $O((\log n)^3)$.\footnote{We omit the proof of this fact, which follows from standard arguments; briefly, one uses  (explicit) parity samplers to combine parity checks to form higher-weight checks.
This process does not change $k$ and $d$.}
Consequently, one can obtain an explicit CSS code with $k=\Omega(n)$, $d = \Omega(n/(\log n)^3)$, parity checks of weight $O((\log n)^3)$, and soundness $\rho = \Omega(1)$.
Additionally, the distance parameter can also be improved to the optimal through distance amplification~\cite{wills2023tradeoff}. In particular, we obtain a CSS code with $k,d=\Omega(n)$, parity checks of weight $O(\text{poly} \log n)$, and soundness $\rho = \Omega(1)$.

In contrast, previous best constructions of qLTC had $k = 1$, $d = \Theta(\sqrt{n})$, $\rho = \Omega(1/(\log n)^2)$ for \cite{hastings2017quantum} and $\rho = \Omega(1/\log n)$ for \cite{leverrier2022towards}. Additional tradeoffs are possible (see Section~\ref{sec:intro:qltc} for further discussion); however, prior to our work no qLTC codes with even polynomial (let alone linear) dimension were known while keeping the locality and soundness at most polylogarithmic and at least inverse-polylogarithmic respectively.

\subsection{Construction}
\label{sec:intro-cons}

For any integer $t\geq 1$ we construct a $t$-dimensional cubical chain complex that generalizes the Sipser-Spielman construction of expander codes~\cite{SipserSp96} (case $t=1$) and the lifted product codes from~\cite{PK1} (case $t=2$) to higher dimensions $t\geq 3$.

The construction is based on two ingredients. Firstly, there is the cubical complex $X$, which is a purely geometric structure (technically, a graded incidence poset with certain good expansion properties). Secondly, there is a system of local coefficients (a sheaf $\sheaf$, in the terminology of~\cite{first2022good}) that is constructed from a family of classical codes of constant dimension and good ``robustness'' properties. The chain complex is a sequence of coboundary maps from $i$-cochains to $i+1$-cochains,
\[ C^0(X,\sheaf)\to C^1(X,\sheaf)\to\cdots\to C^t(X,\sheaf).\]
We expand on these ingredients below.

\paragraph{Cubical complex}
The (higher-dimensional) cubical complex $X$ can be constructed from any set $G$ of size $N=|G|$, and finite subsets of permutations $A_1,\ldots,A_t$ of $G$ of size $n=|A_i|$.\footnote{In general it is possible to make the integer $n$ also depend on $i\in\{1,\ldots,t\}$; for simplicity we assume that there is no such dependence.} The sets $A_i$ should be closed under inverse and such that permutations taken from different sets commute.

For example, in the case of $t=2$ we can take $G$ to be any finite group and let $A_1$ act by multiplication on the left, and $A_2$ by multiplication on the right. This choice, with constant-size sets $A_1$ and $A_2$ of expanding generators, underlies the square complex from~\cite{DELLM,PK2}. For larger $t\geq 3$, it is not clear how to generalize the fact that multiplication on the right commutes, as an action, with multiplication on the left, $(ag)b = a(gb)$. A natural approach is to take $G$ an abelian group, such as $G=\Z_2^{\log N}$, and let $A_i$ be subsets of generators of $G$ acting on it by multiplication. This choice has the disadvantage that each $A_i$ must have size at least logarithmic in $|G|$ (otherwise the graphs won't expand enough).
A better cubical complex with {\em constant} degree can be obtained through an abelian lift \cite{JeronimoMO0T22}. We describe this approach, which gives us the best parameters overall, in Section~\ref{sec:particular-instance-of-complex}.
Other, more structured choices may be possible.

Irrespective of the specific choice of $G$ and $\{A_i\}$, the resulting complex $X(G;\{A_i\}_{i=1}^t$) is $2^t$-partite, with vertex set $G\times\{0,1\}^t$ and $k$-dimensional faces associated with the $2^{|S|}$ vertices $\{g\cdot \prod_{j\in T} a_j:\, T\subseteq S\}$, where $S\subseteq\{1,\ldots,t\}$ has size $|S|=k$ and $a_j\in A_j$. Edges coming from permutations in $A_i$ are said to have direction $i$, and more generally, a face has type $S$ if $S$ is the set of directions of edges it contains. See Section~\ref{sec:geometry} for a complete description, and Figure~\ref{fig:geometry} for an illustration.

The only requirement on the permutation sets $A_i$, other than pairwise commutation, is that the associated graphs $\Cay(G,A_i)$, with vertex set $G$ and edge set $\{g,\pi(g)\}$ for $\pi\in A_i$, have good spectral expansion\footnote{Traditionally, $\Cay(G,A_i)$ is a connected graph. In our construction, what we call ``$\Cay(G,A_i)$'' may be disconnected. The specific notion we require is introduced in Definition~\ref{def:x-expand} --- informally, $\Cay(G,A_i)$ should be a union of not too many disjoint copies of spectral expanders.}. As we will show, this suffices to imply expansion of a natural set of high-dimensional random walks on the cubical complex that are used in the proof of (co)-cycle expansion.

\paragraph{Local coefficients}The cubical complex $X$ is turned into a chain complex $C(G;\{A_i\}_{i=1}^t$) by endowing it with a system of local coefficients. In the terminology of \cite{first2022good}, we put a sheaf structure over the geometrical complex $X$. To each face $f$ of the complex, which can be a vertex, an edge, a square, a cube, etc., we associate a small coefficient space that is in some sense dual to it.
For example, suppose $t=3$ and our complex $X$ is a cubical complex. Fix a $1$-dimensional face, namely an edge. Its space of coefficients is a space of $2$-dimensional matrices. If the edge is labeled by direction $1$, then the matrix dimensions are labeled using the other directions of the complex, namely $2,3$. These spaces are described in detail in Section~\ref{sec:loc-coeffs}. We introduce a system of linear maps $\mF_{f\to f'}$ from the coefficient space of a $k$-face $f$ to that of a ($k+1$)-face $f'$. These maps will be constructed from a collection of parity check matrices $h_1,\ldots,h_t$, where each $h_i\in \F_q^{m_i\times n}$ for some $m_i\leq n$ and some finite field $\F_q$. These matrices will be used as follows. Suppose for example that, as before, $t=3$, $f$ is an edge and $f'$ is a face that contains the edge $f$, as well as another (pair of) edges in direction $2$ (and a fourth edge, parallel to $f$). Then the map $\mF_{f\to f'}$ is obtained by taking the coefficient $x(f)\in \F_q^{m_2\times m_3}$ associated with $f$, expanding it in the second direction by applying $h_2^T$ to obtain $h_2^T x(f) \in \F_q^{n\times m_3}$, and then restricting to the row indexed by the face $f'$ (which is uniquely identified by its edge in direction $2$, labeled by an element of $A_2$ which has size $n$) and returning the resulting row vector, an element of the coefficient space $\F_q^{m_3}$ associated with the face $f'$. These maps are described in much more detail in Section~\ref{sec:co-maps} (where they are denoted $\corest_{f,f'}$).

\subsection{Two-way robustness and product expansion}
\label{sec:intro:robust}

Our proofs of cycle and co-cycle expansion of the complex follow from the expansion of the graphs $Cay(G,A_i)$ and from a property of the local parity check matrices $(h_1,\ldots,h_t)$ which we call {\em two-way robustness}. A collection of linear transformations $h_i:\F_q^n\to\F_q^{k_i}$, $1\leq i \leq t$, is said to be {\em two-way robust} if the $t$-fold homological product of $h_1^T,\ldots,h_t^T$ is {\em coboundary expanding} in all dimensions $i<t$, and the same holds also for $(h_1^\perp)^T,\ldots,(h_t^\perp)^T$, where for each $i$, $h_i^\perp:\F_q^n\to\F_q^{n-k_i}$ is (some) matrix that satisfies $h_i^\perp h_i^T = 0$, i.e. the code with parity check matrix $h_i^\perp$ is dual to the code with parity check matrix $h_i$.
We refer to Section~\ref{sec:robustness} for the precise definition. The reason for considering the $t$-fold homological product is that this is precisely the local view of the complex from a vertex. Namely, the view when restricting to cells that contain a fixed vertex. So our requirements for co-cycle expansion boils down to a local condition (two-way robustness) on the restriction of the code to cells around a vertex and a global condition (spectral expansion of each $\Cay(G,A_i)$). This is in perfect analogy to previous works on co-cycle expansion of simplicial complexes \cite{KKL14, EvraK16} and of cubical complexes \cite{DELLM, PK2}. In both types of complexes, like here, global co-cycle expansion followed from some form of spectral expansion together with coboundary expansion at the links.

In case $t=2$, the notion of two-way robustness coincides with product expansion as in~\cite{PK2} (and with agreement testability, see \cite{DELLM}). However, for $t\geq 3$ these notions diverge, as studied in \cite{kalachev2023prexp-vs-agreement}.
In case $t=2$, Panteleev and Kalachev proved that a pair of random maps $h_1,h_2$ satisfy, with high probability, a version of two-way robustness. The robustness parameters were later improved in \cite{kalachev2022two,DHLV}, but again, only for a pair of random maps, namely $t=2$.
For three or more maps it was conjectured in \cite{kalachev2022two} that random maps are product expanding. Recently, Kalachev and Panteleev~\cite{PK24} showed that any (constant) number of random maps are product expanding, provided that the maps are over a large enough field --- of size $2^{\poly(n)}$, see Theorem~\ref{thm:PK-future} for a precise statement. Because for us $n$ can be taken a constant, this large field size is not an obstacle.

Our proof of global cycle and co-cycle expansion, however, requires {\em two-way robustness}, which is a property that is a priori stronger than product expansion. While product expansion addresses coboundary expansion in level $i=t-1$, two-way robustness requires coboundary expansion for all levels $i<t$. In Section~\ref{sec:product-robust} we prove a reduction from the latter to the former, thereby completing all requirements on the sheaf structure.
An interesting component in this reduction is a proof that expansion with respect to coefficients in a vector space $V$ implies the same with respect to $V^M$, with no dependence on the expansion parameter on $M$.

 \subsection{Main result on expansion}
To state our main result we need to define the notions of distance and expansion that we consider. We will use the term systolic (resp. co-systolic) distance to denote the lowest weight of a $k$-cycle (co-cycle) that is not a $k$-boundary (co-boundary).
Namely, if $C(X)$ is a chain complex with boundary map $\partial$ and coboundary map $\delta$, then
\[ \mu_\syst(k) \,=\, \min\big\{ |x| \,:\, x\in \ker\partial_k - \im\partial_{k+1} \big\},\qquad \mu_\cosyst(k) \,=\, \min\big\{ |x| \,:\, x\in \ker\delta_k - \im\delta_{k-1} \big\}\;.\]
Here we employ the usual notation $\partial_k:C_k(X)\to C_{k-1}(X)$ and $\delta_k:C^k(X)\to C^{k+1}(X)$, with $C_k(X)$ and $C^k(X)$ the (isomorphic) spaces of chains and co-chains respectively.

We establish lower bounds on the \emph{systolic} and \emph{co-systolic} distance of the complex, as well as lower bounds on the cycle and co-cycle expansion, which we define next. The lower bounds are attained via bounding a locally-minimal version of distance (see Definition \ref{def:co-loc-min-dist}), using the strategy initiated in \cite{KKL14,EvraK16}, and used also in \cite{PK2,DELLM} as well as several other works.
The \emph{cycle expansion} of the complex is the smallest size of the boundary of a $k$-chain, relative to the distance of the $k$-chain to the set of $k$-cycles. More formally,
\[  \eps_\cyc(k) \,=\, \min \Big\{ \frac{|\partial_k(x)|}{  dist(x, \ker \partial_k) }\,:\, x\in C_k(X) - \ker \partial_k \Big\}\;,\]
and similarly, co-cycle expansion is defined by
\[  \eps_\cocyc(k) \,=\, \min \Big\{ \frac{|\delta_k(x)|}{  dist(x, \ker \delta_k) }\,:\, x\in C_k(X) - \ker \delta_k \Big\}\;.\]
In the definitions above, $dist(\cdot,\cdot)$ and $|\cdot|$ refers to the Hamming weight, or, in the case of a sheaf complex, the \emph{block-Hamming weight}. 
For precise definitions of all these notions, we refer to Section~\ref{sec:prelim-chains}. We remark that co-cycle expansion was first studied in \cite{KKL14} in the simplicial setting (earlier works \cite{LinialM06, Gromov2010} introduced the notion of coboundary expansion, which coincides with co-cycle expansion for exact chains). The term co-systolic expansion was introduced in \cite{EvraK16}, and combined both co-systolic distance and co-cycle expansion. In retrospect it makes sense to separate into two definitions as above.

Our main result is a general result on (co-)systolic expansion of dimension-$t$ chain complexes of the form $C(G;\{A_i\})$ described in Section~\ref{sec:intro-cons} above. We bound the expansion parameters of the complex as a function of the spectral expansion of the graphs $\Cay(G,A_i)$ and the two-way robustness of the local codes $\{ h_i\}$.
We note that for $t=1$ and $n$ a constant, this result recovers the special case that Sipser-Spielman codes based on good enough expander graphs and good local codes have a linear distance. For $t=2$ and $n$ a constant, we recover the expansion properties of square complexes that underlie the constructions of locally testable and quantum low-density parity check codes from~\cite{DELLM,PK2}. We refer to Theorem~\ref{thm:main} for a statement of this general result.

Here we instantiate the general result with a specific choice of $(G,\{A_i\})$ that are obtained from an abelian lift as in~\cite{JeronimoMO0T22} (this construction is described in Section~\ref{sec:ex2}), and local codes with two-way robustness that can be obtained by exhaustive search (since in this construction the degree $n$ is constant; existence is guaranteed as explained in Section~\ref{sec:intro:robust} above). This gives an explicit construction of a (co-)systolic expander of dimension $t$, as stated in the following corollary.

\begin{theorem}[Main construction, restated as Corollary~\ref{cor:main-cor}]\label{thm:main-applied}
Let $t\geq 2$ be an integer. There is an explicit construction of a chain complex $C(G;\{A_i\})$ and local codes $\{h_i\}$ such that the $k$-chain space $C_k$ has dimension $\Theta(N)$  for every $0\leq k \leq t$ and furthermore
\begin{enumerate}
\item The co-chain complex $C^*$ has  co-systolic distance $\mu_\cosyst(k) = \Theta(|X(k)|/(\log N)^{t-1})$ for every $0\leq k \leq t-1$ and co-cycle expansion $\eps_\cocyc(k)=\Theta(1/(\log N)^{t-1})$  for every $0\leq k \leq t-2$\ ;
\item The chain complex $C_*$ has systolic distance  $\mu_\syst(k) = \Theta(|X(k)|/(\log N)^{t-1})$ for every $1\leq k \leq t$ and cycle expansion  $\eps_\cyc(k) = \Theta(1/(\log N)^{t-1})$ for every $2\leq k \leq t$.
\end{enumerate}
\end{theorem}

In the theorem, the constants implicit in the $\Theta(\cdot)$ notation depend exponentially on $t$, so our bounds get exponentially worse as $t$ increases. 
The main application we have in mind, to quantum locally testable codes, requires only $t=4$, and we did not attempt to optimize the dependence on~$t$.

\subsection{Towards good quantum locally testable codes}\label{sec:intro:qltc}

The existence of good quantum error-correcting codes is one of the pillars that underlie the promise of quantum computing to deliver impactful applications in the long term. Good codes with \emph{low-density parity checks} are particularly sought after for their potential application to fault-tolerance~\cite{gottesman2014fault} and quantum complexity theory~\cite{anshu2023nlts}. In the past few years a flurry of works obtained better understanding and better parameters \cite{evra2022decodable, hastings2021fiber, BE, PK1} culminating in a construction of good quantum LDPC codes \cite{PK2}. This was followed by a couple of related variants \cite{LinH22, leverrier2022quantum,DHLV}. These constructions rely on a well-known connection between the design of quantum codes and chain complexes.  In particular, all existing constructions of good qLDPC rely on virtually the same length-$3$ chain complex, the ``square complex'' from~\cite{DELLM} (which is also a special case of ``balanced product codes'' from~\cite{BE} and of ``lifted product codes'' from \cite{PK1}).


Quantum \emph{locally testable} codes (qLTC) were introduced in~\cite{aharonov2015quantum} as a natural quantization of the notion of local testability for classical codes. Although the formal definition, first given in~\cite{eldar2017local}, is somewhat technical, quantum LTCs have the same intuitive basis as their classical counterparts --- informally, whenever a word (quantum state) is at a certain relative distance $\delta$ from the codespace, this word must violate a fraction $\rho\,\delta$, for some constant $\rho>0$ referred to as the \emph{soundness} of the LTC, of the parity checks. 
The first constructions of qLTCs with soundness that scales better than inverse polynomial were given in~\cite{hastings2017quantum,leverrier2022towards}. These achieve soundness that scales as $1/\log^2(n)$ and $1/\log(n)$ respectively; however, the distance of the code is small ($\Theta(\sqrt{n})$), the codes have constant dimension, and the weight of the parity checks is logarithmic.

Prior to our work quantum LTCs with good locality and soundness and with either reasonable dimension (say, above  poly-logarithmic) or super-quadratic $(n^{\frac{1}{2}+\eps})$ distance were not known to exist. Such codes are expected to have applications in quantum complexity theory. Until recently the most prominent was the famous ``NLTS conjecture''~\cite{eldar2017local}; which was known to follow from qLTC with good distance and soundness (without regard to dimension). However, the NLTS theorem was recently proven without the use of qLTC~\cite{anshu2023nlts}, relying only on the recent qLDPC constructions from~\cite{PK2,leverrier2022quantum}.

While there is no formal connection with the long-standing quantum PCP conjecture~\cite{aharonov2013guest}, it is natural to hope that progress on the former may eventually lead to progress on the latter. One reason is that LTCs are strongly tied to PCPs in the classical world, which, perhaps more philosophically, is due to the inherent local-to-global nature of both questions, where a global property (being a codeword in the case of LTC and being a minimum energy state in the case of qPCP) is to be tested by random local tests. The local to global aspect manifests in high dimensional expanders, shedding some insight as to the connection between these objects.

Recall our main result on quantum locally testable codes, which is stated as Theorem~\ref{thm:qltc} above. This result is obtained from Theorem~\ref{thm:main-applied}, for the case $t=4$, using standard arguments in the construction of quantum (CSS) codes. This result falls short of constructing good quantum LTCs since the parameters differ from ideal parameters by an inverse poly-logarithmic factor. Still, these parameters mark a significant improvement compared to previously known qLTCs \cite{hastings2017quantum,leverrier2022towards}, which achieve $k=1$, $d = \Theta(\sqrt{n})$, and $\rho = \Omega(1/\log n)$. (Additional tradeoffs are possible between the parameters, for example increasing the dimension $k$ at the cost of a decrease of the relative distance; we refer to~\cite{wills2023general,wills2023tradeoff} and in particular Table 3 in~\cite{wills2023tradeoff} for details.

The codes from Theorem~\ref{thm:qltc} are explicit CSS codes, and in particular have an efficient encoder. Moreover, although we do not show this explicitly, we expect these codes to have linear-time decoders and noisy syndrome decoders (often referred to as single-shot decoders), along the same lines as~\cite{DHLV,gu2023single}. Although there is no formal connection known, we observe that qLTC may be of additional relevance to the problem of noisy syndrome decoding. This is because a \emph{simpler} problem than decoding from a noisy syndrome is deciding, given an approximate bound on the size of the syndrome, if the state is close to the codespace or not---which is precisely the problem solved by qLTCs. This observation suggests that, beyond the basic LDPC requirement, the LTC condition may have applications to e.g.\ fault-tolerance that have not yet been fully explored. 
\subsection{Techniques}

Our construction of a higher-dimensional cubical chain complex, when instantiated at $t=2$, is not identical to the chain complex of \cite{PK2}, but rather can be seen as a ``$90^\circ$ rotation'' of it as in \cite{DHLV}. The vector spaces participating in the chain complex are ordered by geometric dimension, a feature that facilitates generalization to $t>2$.
The price is that our construction is not self-dual, in the sense that the geometric properties of the complex are different from those of the co-complex. Therefore, the argument for showing co-cycle expansion is different from the argument for showing cycle expansion. Nevertheless, the latter is based on a reduction to the former for a \emph{different} co-complex that is, in some sense, dual to the original one.

\paragraph{Cocycle expansion} Our proof for co-cycle expansion follows the general template of earlier works on cocycle expansion in HDX \cite{kaufman2016isoperimetric,EvraK16}, and the analogous statements for qLDPC codes and for LTCs \cite{PK2,DELLM}.

Let us outline this argument for showing cocycle expansion at the highest level of the complex, i.e.\ for co-chains in dimension $k=t-1$, which is the most interesting case.
Consider an element $x\in C^{t-1}(X)$ such that $\delta(x)=0$, and assume it is locally minimal, that is, assume its weight is minimal with respect to adding coboundaries of local elements. We let $\mathcal{A}$ denote the set of \emph{active faces}, which is the subset of geometric faces $f\in X(t-1)$ such that $x(f)\neq 0$. The core of the proof consists in showing that the condition $\delta(x)=0$ implies that the set $\mathcal{A}$ does not expand according to a natural random walk $W$ on $X(t-1)$. However, the random walk is such that it \emph{is} expanding (this is where the assumption that $\Cay(G,A_i)$ expands is used); therefore either $\mA=\emptyset$ or $|\mA|$ is large, which is the desired conclusion. To show that the set $\mathcal{A}$ does not expand, we show that most neighbors of an $f\in \mathcal{A}$ according to $W$ lie in $\mathcal{A}$. This step uses the definition of $\delta$ to argue that the condition $\delta(x)=0$ implies that many neighbors of $f$ under $W$ must be in $\mA$. The robustness condition is used crucially at this step.

Indeed, one of the contributions of this work is in coming up with the appropriate definition for the random walk $W$, together with the definition for two-way robustness so that this argument goes through.

We note that the walk $W$ is defined as a mixture of $(t-1)$ different walks, each of which goes ``down'' $j\in[t-1]$ steps in the complex to a $(t-1-j)$-dimensional sub-face of $f$, takes one parallel step, and then ``up'' again $j$ steps to a ``neighbor'' of $f$ satisfying certain conditions. As such, this walk is similar to iterations of the ``down-up'' random walk that is familiar in the combinatorial analysis of high-dimensional expanders. It is often used to circumvent certain non-expanding link-like obstacles that would appear in a more naive choice of random walk $W$. In this sense the extension of the arguments from \cite{DHLV, PK2, DELLM} from $t=2$ to higher dimension $t>2$ is analogous to the move from \cite{kaufman2016isoperimetric} to \cite{EvraK16}. We refer to Section~\ref{sec:codistance} for details on this part of the argument.

\paragraph{Cycle expansion}
Our argument for cycle expansion is completely different from the argument for co-cycle expansion described above. In fact, the argument is an indirect reduction to the case of co-cycle expansion, generalizing a reduction that appeared in \cite{DHLV}, for the case $t=2$. Moving to higher dimensions $t>2$ is more involved, consisting of a nice arrow-chasing argument on a double chain complex which we describe next.

The idea is to define a new ``dual'' chain, whose coefficients are local views from geometric objects of the main chain. In the dual chain, the coefficient at a vertex tells us everything about the $k$-faces touching that vertex. We define two kinds of boundary maps for this dual chain, and show that they commute nicely with each other, giving rise to a double chain complex (shown in Figure \ref{fig:xz-2}). The argument proceeds as follows. Given a low weight $k$-chain $x$ such that $\partial x=0$, we wish to find a preimage $z$ such that $\partial z = x$. We first look at the dual chain $\set{x_v}$ obtained by local views of $x$. We take a preimage $z_v$ separately for each $x_v$. This is possible because the local maps are exact and $\partial x=0$ implies $\partial x_v=0$ for all $v$. The preimages $z_v$ are also local views, but they don't necessarily agree with each other, and this is where things become interesting.
We move to look at the differences between local views
of adjacent vertices, assigning these differences as local views to the {\em edges}. These new local views have zero boundary because of the commutativity of the arrow diagram. As a result we are in a similar situation as before, except that we moved one dimension up, from vertices to edges. We proceed inductively (again using the exactness of the (local) boundary map to find local preimages, take their differences, and so on).
This inductive arrow-chasing ends in one of two cases. Either we reach a stage where all local views agree, or we push all the way to the end of the chain without reaching agreement. In the former case, we have obtained a global ``corrections'' chain which can then be propagated all the way back down, leading to a global solution $z$ such that $\partial z=x$.  The interesting case however is when one reaches the end of the chain, without agreement. In this case the key observation is that local views, which lie on the higher-dimensional faces of the complex (i.e.\ $t$-cubes) are tensor codewords (because their boundary is $0$) and can thus be decoded. The ``decoded'' object can be interpreted as a co-cycle at the $(t-1)$-st level of a complex with dual local codes $h_1^\perp,\ldots, h_t^\perp$. Applying the co-distance argument, we find a pre-image co-cycle. When re-encoded, this pre-image co-cycle provides corrections which can once again be propagated down the chain to define $z$.

\subsection{Discussion}

In this work we construct quantum LTCs with parameters that are ``almost'' good, up to polylogarithmic factors.
One approach for removing the poly-logarithmic factors is to venture away from cubical complexes and study chain complexes based on simplicial high-dimensional expanders, as done in \cite{evra2022decodable}. The advantage is that constant degree expanding complexes are known in any dimension $t$. In contrast, our complexes are also constant-degree, but unfortunately their expansion guarantees hold only for sets of density up to $1/\poly\log n$. In \cite{evra2022decodable}, the distance of the quantum code follows from the systolic and co-systolic distance of the complex. Unfortunately, the systolic distance is inherently sublinear (in their construction it is polylogarithmic). This obstacle theoretically should also appear for cubical complexes, but it is circumvented through the addition of local codes.
Why not add local codes to high-dimensional expanders? This seems like a reasonable direction; the main difficulty is that, unlike in the cubical case, the local structure around a vertex does not have a natural product structure. Two-way robustness seems quite tricky when the structure is not product, and is currently not known.

Generally speaking the connection between HDX and quantum codes has been extremely fruitful, leading to insights such as the distance balancing method from~\cite{evra2022decodable}. This connection crucially requires one to consider both the complex and the co-complex, to reflect the symmetry between $X$ and $Z$ parity checks in the definition of a quantum CSS code. Most constructions of HDX, including ours, do not exhibit a perfect symmetry between complex and co-complex. For our complexes, we are nevertheless able to reduce cycle expansion to co-cycle expansion of a \emph{related} complex. This proof technique could be of independent interest, and lead to new constructions of HDX exhibiting expansion in both directions.

Quantum LDPC codes are intrinsically linked with the study of high-dimensional surfaces, and indeed their recent construction leads to new topological objects~\cite{freedman2021building} (and vice-versa~\cite{hastings2016quantum,portnoy2023local}). It will be interesting to determine if our higher-dimensional construction has similar consequences. Conversely, one may hope that a construction of a high-dimensional cubical complex with constant degree could be obtained by leveraging connections with group theory, a central source of constructions of simplicial complexes and HDX~\cite{lubotzky2018high}. Cubical complexes with constant degree and strong expansion properties are known to exist \cite{JL}, but these don't seem to support a system of local coefficients that would lead to a useful chain complex.

Of course one of the main focus points of the area, which we leave entirely open, is the possible application of qLTC constructions to quantum complexity theory. By providing an explicit family of candidate good qLTC we uncover a concrete object that may form the basis for later explorations of the connections with complexity.

\paragraph{Organization} We start with some preliminaries regarding notation, chain complexes, quantum codes and local testability in Section~\ref{sec:prelim}. In Section~\ref{sec:results} we give an overview of our results. In Section~\ref{sec:complex} we describe our high-dimensional cubical complex. In Section~\ref{sec:loc-chain} we introduce an auxiliary ``local'' chain complex, that will be used in the analysis. In this section we prove that a random tuple of codes, for large enough field size, is two-way robust (relying on the recent result on product expansion of such codes). In Section~\ref{sec:expansion} we describe and analyze some random walks on the cubical complex, that will also be used in the analysis. Finally, in Section~\ref{sec:codistance} and Section~\ref{sec:distance} we prove our main results, lower bounds on the locally co-minimal distance and distance respectively at each level of our chain complex.

\paragraph{Acknowledgments} We thank Louis Golowich for pointing us to \cite{JeronimoMO0T22} which leads to a construction with constant (rather than logarithmic) degree. We thank Hayata Yamasaki and Quynh T. Nguyen for pointing out an error in the previous rate analysis. ID is supported by ERC grant 772839, and ISF grant 2073/21. TV is supported by a research grant from the Center for New Scientists at the Weizmann Institute of Science, AFOSR Grant No. FA9550-22-1-0391, and ERC Consolidator Grant VerNisQDevS (101086733). TCL is supported in part by funds provided by the U.S. Department of Energy (D.O.E.) under the cooperative research agreement DE-SC0009919 and by the Simons Collaboration on Ultra-Quantum Matter, which is a grant from the Simons Foundation (652264 JM).

\section{Preliminaries}
\label{sec:prelim}

\subsection{Chain Complexes}
\label{sec:prelim-chains}

Chain complexes are algebraic constructs that have proven helpful to connect the study of quantum codes with notions of high-dimensional expansion. Although there are more general definitions, for us a chain complex $C$ will always be specified by a sequence of finite-dimensional vector spaces $\{C_i\}_{i\in\Z}$ over a finite field $\F_q$, termed \emph{chain spaces}, together with linear maps $\partial_i:C_i\to C_{i-1}$, termed \emph{boundary maps}, that satisfy the condition $\partial_{i-1} \partial_{i}=0$ for all $i$. (It will always be the case that $C_i=\{0\}$ for all $|i|$ large enough.) (We will focus on the case where $\F_q$ has characteristic $2$, which allows us to view a linear map over $\F_q$ as a linear map over $\F_2$.) Since $C_i$ is a finite vector space over $\F_q$, it always takes the form $C_i = \F_q^{D_i}$ for some integer $D_i$. Elements of $C_i$ are called \emph{$i$-chains}. Taking the  standard inner product on $\F_q^{D_i}$, we can define the dual space $C^i:=(C_i)^* \simeq \F_q^{D_i} = C_i$. Elements of $C^i$ are called \emph{$i$-cochains}, and the map $\partial_i^* : C^{i-1}\to C^i$ is denoted $\delta_{i-1}$ and called the \emph{co-boundary} map.

Given a chain complex $C$, we can associate to it subsets of its $i$-chains called $i$-cycles (elements of $\ker\partial_i$) and $i$-boundaries (elements of $\im \partial_{i-1}$). We can also define further algebraic objects such as homology groups and various notions of distance and expansion. We avoid surveying all such quantities here, referring the interested reader to e.g.~\cite{PK2}. Instead, we focus on the definitions that are essential for this paper.

\paragraph{Sheaf complexes}
We consider chain complexes that are obtained by attaching \emph{local coefficient spaces} $\mF$ (or, in the terminology of~\cite{first2022good}, a \emph{sheaf}) to a \emph{graded poset} $X$. While ultimately these are ``usual'' chain complexes as above, their definition from two distinct objects provides a convenient way to separate the ``global (geometric) structure'', provided by $X$, and the ``local (algebraic) structure'', provided by $\mF$.

We give the definitions.\footnote{We give definitions that are sufficiently general to capture our constructions, but are more restrictive than the most general setting considered in the literature, such as in~\cite{PK2,first2022good}. In particular, we emphasize that we only consider complexes defined over $\F_q$ with characteristic 2.} Recall that a \emph{poset} is a set $X$ equipped with a partial order $\preceq$. A \emph{graded poset} is a poset $X$ together with a map $\rho:X\to\mathbb{Z}$ called a \emph{rank function} such that for any $f,g\in X$, if $f\preceq g$ then $\rho(f)\leq \rho(g)$, and furthermore if $f\prec g$ and there is no $f'\in X$ such that $f\prec f'$ and $f'\prec g$ (a condition which we write as $f\precdot g$), then $\rho(g)=\rho(f)+1$.
Given a graded poset $X$ we define $X(i)=\{f\in X: r(f)=i\}$. Finally, we say that a graded poset $X$ is an \emph{incidence poset} if for every $f\preceq f''$ such that $r(f'')=r(f)+2$, there is an \emph{even} number of $f'\in X$ such that $f\precdot f'\precdot f''$.

Given a graded poset $X$ with rank function $r$, a \emph{system of local coefficients}, or \emph{sheaf} $\mF$ for $X$ is given by two collections of objects. Firstly, for each $f\in X$ we specify a \emph{local coefficient space} $V_f$. Although in general $V_f$ may have a general group structure, here we only consider the case where each $V_f$ is a finite dimensional $\F_q$-vector space $V_f\simeq \F_q^{m(f)}$ for some ``local dimension'' parameter $m(f)$. Secondly, for every $f\succeq f'$ there is a homomorphism (of vector spaces) $\mF_{f\to f'}:V_f\to V_{f'}$ such that whenever $f\succeq f'\succeq f''$, we have that $\mF_{f'\to f''}\circ \mF_{f\to f'}=\mF_{f\to f''}$.

Given a graded incidence poset $X$ and a local system of coefficients $\mathcal{F}$ on $X$ we define a \emph{chain complex} $C_*(X,\mathcal{F})$ as follows. For each $i\in \mathbb{Z}$, define $C_i(X,\mF)=\oplus_{f\in X(i)} V_f$ and $C_*(X) = \oplus_i C_i(X,\mF)$. Next, define linear maps $\partial_i: C_i(X,\mF)\to C_{i-1}(X,\mF)$ by letting, for $f\in X(i)$ and $u\in V_f$, $\partial_i(u) = \sum_{f'\precdot f} \mF_{f\to f'}(u)$. It is easily verified that these maps satisfy the chain complex condition $\partial_{i-1}\partial_i=0$ as a result of the compatibility condition for the sheaves, and the incidence condition on the poset. We write $C^*(X,\mF)=\oplus_i C^i(X,\mF)$ for the co-complex. Sometimes, when the sheaf $\mF$ is clear from context we write $C_i(X)$ and $C^i(X)$ for $C_i(X,\mF)$ and $C^i(X,\mF)$ respectively.

\paragraph{co-cycle and cycle expansion}

The definition of co-cycle expansion is based on definitions from the domain of simplicial complexes. Coboundary expansion was introduced in \cite{LinialM06, Gromov2010}, and the extension to cycle and co-cycle expansion is direct. First, we introduce the notation
\begin{equation}\label{eq:block-ham}
|x|\,=\,\sum_{f\in X}\, 1_{x(f)\neq 0}\;,
\end{equation}
where for $x\in C^*(X,\mF)$ and $f\in X$ we write $x(f)\in V_f$ for the coefficient of $x$ associated with the face $f$. Eq.~\eqref{eq:block-ham} is an analogue of the Hamming weight  which counts the number of faces on which $x$ has a nonzero coefficient. $|x|$ is called the \emph{block-Hamming weight}, or \emph{weight} for short, of $x$. Note that this quantity is in general \emph{smaller} than the Hamming weight, which would sum the Hamming weights of each vector $x(f)$, $f\in X$.
\begin{definition}[Co-cycle expansion]\label{def:cosyst-exp}
Let $C^*(X,\sheaf)$ be a co-chain sheaf complex of dimension $t$. For integer $0\leq i < t$, let $\eps_\cocyc(i)$ be the co-cycle expansion
\[  \eps_\cocyc(i) \,=\, \min \Big\{ \frac{|\delta(x)|}{  \min_{y\in\ker\delta_i}|x-y|}\,:\, x\in C^i(X) - \ker \delta_i \Big\}\;.\]

\end{definition}

An analogous definition of \emph{cycle expansion} $\eps_\cyc$ is given in the obvious way, with the co-boundary map $\delta$ replaced by the boundary map $\partial$.

\paragraph{Co-systolic and systolic distance}
We define the the co-systolic distance of $C^i(X)$ to be the smallest weight of a co-cycle that is not a coboundary,
\[ \mu_\cosyst(i) \,=\, \min\big\{ |x| \,:\, x\in \ker\delta_i - \im \delta_{i-1}\big\}\;,\]
The systolic distance $\mu_\syst$ is defined similarly by replacing the coboundary map by the boundary map $\partial$.
\paragraph{Local minimality}
An important definition, that will be used throughout, is that of \emph{co-local minimality} of a co-chain $u\in C^i(X)$, and of the \emph{locally co-minimal distance} of a space $C^i(X,\mF)$.

The following definition is essentially from~\cite{kaufman2016isoperimetric}.

\begin{definition}\label{def:co-loc-min-dist}
For $i>0$, an element $x\in C^i(X,\mF)$ is \emph{locally co-minimal} if $|x+\delta(y)|\geq |x|$ for all $y\in C^{i-1}(X,\mF)$ supported on $X_{\ge v}(i-1)= \{f \in X(i-1): f \succeq v\}$ for some $v \in X(0)$.
\end{definition}

This notion is referred to as \emph{local} co-minimality because $X_{\ge v}(i-1)= \{f \in X(i-1): f \succeq v\}$ has a finite size and is considered `local' compared to the `global' $X(i-1)$.
Notice that the weight $x$ cannot be lowered through any `local' change $\delta(y)$.

Next we define a notion of locally co-minimal distance in the natural way:
\[ \dist_\coloc(i) \,=\, \min\big\{ |x| \,:\, x\in \ker\delta_i - \{0\},\, \text{$x$ is locally co-minimal}\big\}\;.\]
A lower bound on the locally co-minimal distance immediately implies bounds on the co-systolic distance and co-cycle expansion parameters of the complex.

\begin{lemma}\label{lem:cosys-exp}
  Let $C^*(X,\sheaf)$ be a co-chain sheaf complex of dimension $t$. Then for any $0\leq i \leq t-1$,
  \[\mu_\cosyst(i) \,\geq\, \dist_\coloc(i)  \;,\]
  and for any $0\leq i \leq t-2$
  \[ \eps_\cocyc(i) \,\geq\,  \min\Big\{ \frac{1}{\max_{v\in X(0)} |X_{\geq v}(i)|}\,,\; \frac{\dist_\coloc(i+1)}{|X(i)|}\Big\}  \;.\]
\end{lemma}
\begin{proof}
  By definition, $\mu_\cosyst(i)$ is the smallest norm of an element $x\in C^i(X)$ such that $\delta(x)=0$ and $x\notin \im\delta$. We claim that such an $x$ must be locally minimal. This is because if it is not locally minimal, then there is $y\in C^{i-1}(X)$ such that if $x'=x+\delta(y)$ then $|x'|<x$. But $x'$ still satisfies $\delta(x')=0$ and $x'\notin \im \delta$, which contradicts the minimality of $x$. Therefore $\mu_\cosyst(i)\geq \dist_\coloc(i)$.

  We now consider the co-cycle expansion. Let $x\in C^i(X)$ such that $x\neq 0$. If $|\delta(x)| \geq \dist_\coloc(i+1)$ then because $|x|\leq |X(i)|$ there is nothing to show. So assume $|\delta(x)|< \dist_\coloc(i+1)$. If $\delta(x)=0$ then again there is nothing to show. If $\delta(x)\neq 0$ then it cannot be locally co-minimal. So let $v\in X(0)$ and $y\in X_{\geq  v}(i)$ be such that $|\delta(x)+\delta(y)|<|\delta(x)|$. Let $x^{(1)}:=x+y$ and note that
  $|x-x^{(1)}|\leq |X_{\geq v}(i)|$.
  Repeating this process at most $|\delta(x)|$ times eventually gives some $x^{(\ell)}$ such that $|\delta(x^{(\ell)})|\leq|\delta(x)|$ and $\delta(x^{(\ell)})$ is locally co-minimal.
  Therefore $\delta(x^{(\ell)}) = 0$ and $x^{(\ell)}\in\ker\delta$. Finally, $|x-x^{(\ell)}|\leq \max_{v\in X(0)} |X_{\geq v}(i)| |\delta(x) |$, as desired.
\end{proof}

\subsection{Quantum LDPC Codes}

Since our construction falls within the framework of CSS codes~\cite{calderbank1996good, steane1996error}, we restrict our attention to such codes. A \emph{quantum CSS code} is uniquely specified by a length-$3$ chain complex
\begin{equation}\label{eq:length3-css}
 X\colon\;\F_2^{X(0)} \xrightarrow{H_X^T} \F_2^{X(1)} \xrightarrow{H_Z} \F_2^{X(2)}\;,
\end{equation}
where $X(1)$ is identified with the qubits, $|X(1)|=n$, and each element of $X(0)$ (resp.\ $X(2)$) is identified with an $X$-parity check (resp. $Z$-parity check). We represent the linear maps $H_X,H_Z$ through their matrices $H_X \in \F_2^{X(0)\times X(1)}$ and $H_Z\in \F_2^{X(2)\times X(1)}$, and note that the transpose in~\eqref{eq:length3-css} is taken according to the canonical basis of each space, the same basis that we fixed to identify $(\F_2^m)^*\simeq\F_2^m$.

Let $\sigma_X = \begin{pmatrix} 0 & 1 \\ 1 & 0 \end{pmatrix}$ and $\sigma_Z = \begin{pmatrix} 1 & 0 \\ 0 & -1\end{pmatrix}$ denote the usual Pauli matrices. The code associated with the chain complex $X$, which we denote $\mC = CSS(H_X,H_Z)$, is a subspace of the space $(\C^2)^{\otimes n}$ of $n$ qubits, with stabilizer generators $\sigma_X(a)=\sigma_X^{a_1}\otimes \cdots \otimes \sigma_X^{a_n}$ for $a$ ranging over the $n_X=|X(0)|$ rows of $H_X$ and $\sigma_Z(b)=\sigma_Z^{b_1}\otimes \cdots \otimes \sigma_Z^{b_n}$ for $b$ ranging over the $n_Z=|X(2)|$ rows of $H_Z$. The code has dimension
\begin{equation}\label{eq:css-dim}
k = \dim\ker H_Z - \dim\ \im H_X^T
\end{equation}
 and distance $d = \min(d_X, d_Z)$ where
\begin{equation}\label{eq:css-dist}
  d_X = \min_{c^1 \in \ker H_Z - \im H_X^T} |c^1|_H\;,\quad
  d_Z = \min_{c_1 \in \ker H_X - \im H_Z^T} |c_1|_H\;,
\end{equation}
where here $|c|_H$ denotes the Hamming weight. We summarize these parameters by writing that $\mC =CSS(H_X,H_Z)$ is an $[n,k,d]$ quantum (CSS) code. Finally, the code is called \emph{low-density parity check (LDPC)} if each row vector of the matrices $H_X, H_Z$ have constant Hamming weight.\footnote{Of course this notion only makes sense when considering a family of codes with asymptotically growing size $n$; in which case we mean that the row weights of $H_X$ and $H_Z$ should remain bounded by a universal constant as $n\to\infty$.}

\subsection{Local testability}

We start with the definition of local testability for a classical code $\code = \ker H$, where $H\in \F_2^{m\times n}$ is a parity check matrix. Recall that such a code is termed an $[n,k,d]$ code, where $k=\dim\ker H$ and $d = \min\{|c|:c\in \ker H -\{0\}\}$.  We give a (somewhat restricted) definition of local testability that depends on the choice of $H$, which implicitly specifies a natural ``tester'' for the code. Intuitively, a code is locally testable if, the further a word $x$ is from the code, the higher the probability of the ``natural tester'' --- which selects a row of $H$ uniformly at random and checks the corresponding parity of $x$ --- to reject.

\begin{definition}\label{def:cltc}
The (classical) code $\code=\ker H$ is called locally testable \emph{with soundness $\rho$} if for all $x\in \F_2^n$ it holds that
\begin{equation}\label{eq:lt-def}
 \frac{1}{m} |Hx|_H \,\geq\, \rho \frac{d_H(x,\code)}{n}\;,
\end{equation}
where we use $d_H$ to denote the distance function associated with the Hamming weight (to distinguish it from $d$, which is the distance function associated with the block Hamming weight $|\cdot|$).
\end{definition}

This definition can be restated in terms of expansion of the natural length-$2$ chain complex associated with $H$,
\[ X\colon \;\F_2^{X(0)} \xrightarrow{\delta=H} \F_2^{X(1)} \;,\]
where $X(0)=\{1,\ldots,n\}$ and $X(1)=\{1,\ldots,m\}$. Then the definition says that for every element $x\in \F_2^{X(0)}$, the (relative) weight $\frac{1}{m}|\delta(x)|_H$ is at least a constant times the (relative) distance of $x$ from $\ker\delta$. In the study of chain complexes, the parameter $\rho$ is referred to as the \emph{co-cycle expansion} of the complex, and denoted $h^1(X,\Z_2)$. (We will not use this notation.)

Now consider a quantum code $\code$. A general definition of local testability, and of the associated soundness parameter $\rho$, exists and is given in~\cite{eldar2017local}, which formalizes the same intuition, between distance to the code space and probability of rejection by a natural tester, as for the classical setting. This definition is somewhat cumbersome to work with and we will not need it, so we omit it. Instead, we collect the following two facts from the literature.

Firstly, it is known that local testability of a quantum CSS code is directly linked to local testability of the two associated check matrices $H_X,H_Z$. Formally, we have the following.

\begin{lemma}[Adapted from Claim 3 in~\cite{aharonov2015quantum}]
If the quantum CSS code $\code=CSS(H_X,H_Z)$ is locally testable with soundness $\rho$ (in short, $\rho$-qLTC) then the classical codes with parity-check matrices $H_X$ and $H_Z$ are each locally testable with soundness at least $\rho/2$. Conversely, if the classical codes with parity-check matrices $H_X$ and $H_Z$ are locally testable with soundness $\rho$ then the quantum CSS code $\code=CSS(H_X,H_Z)$ is locally testable with soundness at least $\rho$.
\end{lemma}

Because the lemma is tight up to a factor $2$, and we will not be concerned with constant factors, we can take the local testability of $H_X$ and $H_Z$ as our definition of local testability for the quantum code $CSS(H_X,H_Z)$.

\begin{definition}[Definition of local testability for quantum CSS codes]\label{def:qltc}
The (quantum) CSS code $\code=CSS(H_X,H_Z)$ is called locally testable \emph{with soundness $\rho$} if both $\code_X=\ker H_X$ and $\code_Z = \ker H_Z$ are (classical) codes with soundness at least $\rho$.
\end{definition}

We emphasize that this definition is local to our paper, and is equivalent to the general accepted definition of a quantum LTC up to a multiplicative factor $2$ in the soundness.

To translate between the notion of (co-)cycle expansion that we formulate our main analytic results in, and quantum local testability properties of the quantum codes that can be built from the underlying complex, we use the following lemma.





\begin{lemma}\label{lem:soundness-param}
Let $C(X,\mF)$ be a sheaf complex, and $i\in\Z$. Let $\code = CSS(\partial_i,\delta_i)$ be the quantum CSS code associated with the linear maps $\partial_i$ and $\delta_i$ over $\F_q$.\footnote{Here, we slightly abuse notation and use $C_i\simeq C^i$ to view both maps $\partial_i$, $\delta_i$ acting on the same space.} (When constructing codes, $\partial_i$ and $\delta_i$ are treated as $\F_2$-linear maps.) Let $D_i=\dim\, C_i(X,\sheaf)$ and let $M_i=\max_{f\in X(i)}\dim\ V_f$. Suppose further that $C(X,\mF)$ has systolic and co-systolic distance in dimension $i$ at least $\mu_\syst$ and $\mu_\cosyst$ respectively, and cycle and co-cycle expansion in dimension $i$ at least $\eps_\cyc$ and $\eps_\cocyc$ respectively.

Then $\code$ is an $[n,k,d]$ quantum code, where
\[ k\,=\,\dim_{\F_2}\, C_i(X)-\dim_{\F_2}\,\im\delta_{i}^T-\dim_{\F_2}\,\im\partial_{i}^T\]
 and
\[ d\,\geq \,\min \{\mu_\syst,\,\mu_\cosyst\}\;.\]
Moreover, the code $\code$ is $\rho$-qLTC where
\begin{equation}\label{eq:qltc-bound}
 \rho\,\geq\, \left(\frac{1}{\log_2 q}\right)\min \Big\{\frac{D_i}{D_{i-1}}\frac{\eps_\cyc}{M_i},\; \frac{D_i}{D_{i+1}}\frac{\eps_\cocyc}{M_i} \Big\}\;.\end{equation}
\end{lemma}

We note that the bound in~\eqref{eq:qltc-bound} may be somewhat loose, due to the interplay between the use of the Hamming weight and the block Hamming weight, which leads to the factor $\frac{1}{M_i}$ on the right-hand side. In our constructions, $M_i$ is thought of as polylogarithmic, or even constant, in the length of the code.

\begin{proof}
The formula for the dimension $k$ is clear from~\eqref{eq:css-dim} $k = \dim\ker \delta_i - \dim\,\im\partial_i^T$ and the fact that for the linear operator $\delta_i:\F_2^n\to\F_2^m$, $\dim\ker\delta_i=n-\dim\,\im\delta_i=n-\dim\,\im\delta_i^T$. Intuitively, this is saying that the number of logical qubits is equal to the number of physical qubits minus the number of independent stabilizer generators.

For the distance, the bound follows from the definition~\eqref{eq:css-dist} of the distance of a quantum CSS code, the definition of the parameters $\mu_\syst$ and $\mu_\cosyst$ and the fact that the block weight $|c|$ is a lower bound on the Hamming weight of $c$.

To show the bound on the soundness parameter $\rho$, consider first the soundness of the classical code $\code_X=\ker H_X=\ker\partial_i$. Recall that $\partial_i:C_i(X,\sheaf)\to C_{i-1}(X,\sheaf)$. Let $c\in C_i(X,\sheaf)$.
Then
\begin{align*}
\frac{1}{D_{i-1}} |H_X c|_H &\geq \frac{1}{D_{i-1}} |\partial_i c|\\
&\geq \frac{1}{D_{i-1}} \,\eps_\cyc\,d(c,\ker \partial_i)\\
&\geq \frac{D_{i}}{D_{i-1}} \frac{\eps_\cyc}{M_i} \frac{d_H(c,\code_X)}{D_i}\;.
\end{align*}
Here the first inequality follows since $|\cdot|_H\geq |\cdot|$, the second inequality uses the definition of $\tilde{\eps}_i$, and the third inequality uses that $|x|\geq \frac{1}{M_i}|x|_H$ for $x\in C_i(X)$.
Let $|c|_{\F_2}$ be the Hamming weight of $c$ when viewed as a vector over $\F_2$.
Because $|c|_{\F_2} \le |c|_H \le (\log_2 q) |c|_{\F_2}$, we obtain the lower bound on the soundness of the code $\code_X$.

The argument showing a lower bound on the soundness of the code $\code_Z=\ker H_Z$ is analogous, and we omit it.

\end{proof}

\section{The complex}
\label{sec:complex}

Let $t\geq 1$ be an integer. Our construction is based on a graded incidence poset that is obtained from a $t$-dimensional cubical complex
\[X = X(0)\cup\cdots \cup X(t)\;.\]
The cubical complex can be generated from a set $G$ of size $N=|G|$, and finite subsets of permutations $A_1,\ldots,A_t$ of $G$ of size $n=|A_i|$. The sets $A_i$ should be closed under inverse and such that permutations taken from different sets commute. We denote this complex $X=X(G;\{A_i\})$. The requirement that the sets $A_i$ pairwise commute can be understood as a requirement that allows us to iterated the \emph{balanced product} construction from~\cite{BE}, which a priori was only introduced for the product of two graphs (or more generally chain complexes) with naturally commuting left and right group actions.

A concrete example that one may keep in mind is the case where $G$ is a finite abelian group and $A_1\ldots,A_t\subseteq G$ are generating subsets. Another example, when $t=2$, is where $G$ is any finite group and $A_1,A_2\subset G$ act by multiplication on the left and right respectively. The best construction that we identified is based on a third example, derived from an expander construction of \cite{JeronimoMO0T22}. This construction was pointed out to us by Louis Golowich. We describe it in more detail in Section~\ref{sec:ex2}.

\subsection{Geometry} \label{sec:geometry}

We describe our general construction of a complex $X$, in the process introducing notation that will be used throughout. See also Figure \ref{fig:geometry}.

\paragraph{Faces} For $0\leq k \leq t$ the $k$-dimensional faces of $X$, i.e.\ the elements of $X(k)$, are partitioned according to their \emph{type}, which can be any subset $S\subseteq\{1,\ldots,t\}$ of size $|S|=k$. Let $\ol{S}$ denote the complement of $S$ in $\{1,\ldots,t\}$. For every such $S$, a face $f$ of type $S$ is uniquely specified as
 \begin{equation}\label{eq:face-def}
f\,=\, [\,g;\ (a_j)_{j\in S},\ (b_j)_{j\in \ol{S}}\,]\;,
\end{equation}
where $g\in G$, for each $j\in S$, $a_j\in A_j$, and for each $j\in\ol{S}$, $b_j\in \{0,1\}$. We let $X(S)$ denote the set of faces of type $S$. Thus $|X(S)|=|G| n^k 2^{t-k}$. Geometrically, the $k$-face $f=[g;a,b]$ contains the $2^{k}$ $0$-faces (vertices)
    \begin{equation}\label{eq:face-vertices}
		\big\{ \big(g\cdot\prod_{j:b'_j=1} a_j\,;\, b'\conc b\big) \,\big|\, b'\in \{0,1\}^S\big\}\;,
		\end{equation}
    where we use the notation $b'\conc b$ to denote string concatenation with re-ordering of indices in the obvious way.
		Here and in the following we use the notation $g\cdot a$ to denote the element $a(g)$, for $a\in A_j$ a permutation of $G$. Note that the element $g\cdot\prod_j a_j$ is well-defined because we assumed that permutations $a_j$ taken from different sets $A_j$ pairwise commute.

For a face $f = [g;a,b]$ we write $S=\type(f)$ for its type, and (slightly overloading notation compared to~\eqref{eq:face-def}) write
\[ f_0 = g\qquad\text{and}\qquad f_j = \begin{cases}
  a_j  & j\in S\\
  b_j & j\not\in S
\end{cases}\]
for $j=1,\ldots,t$.

\begin{figure}
    \centering
    \includegraphics[scale=0.35]{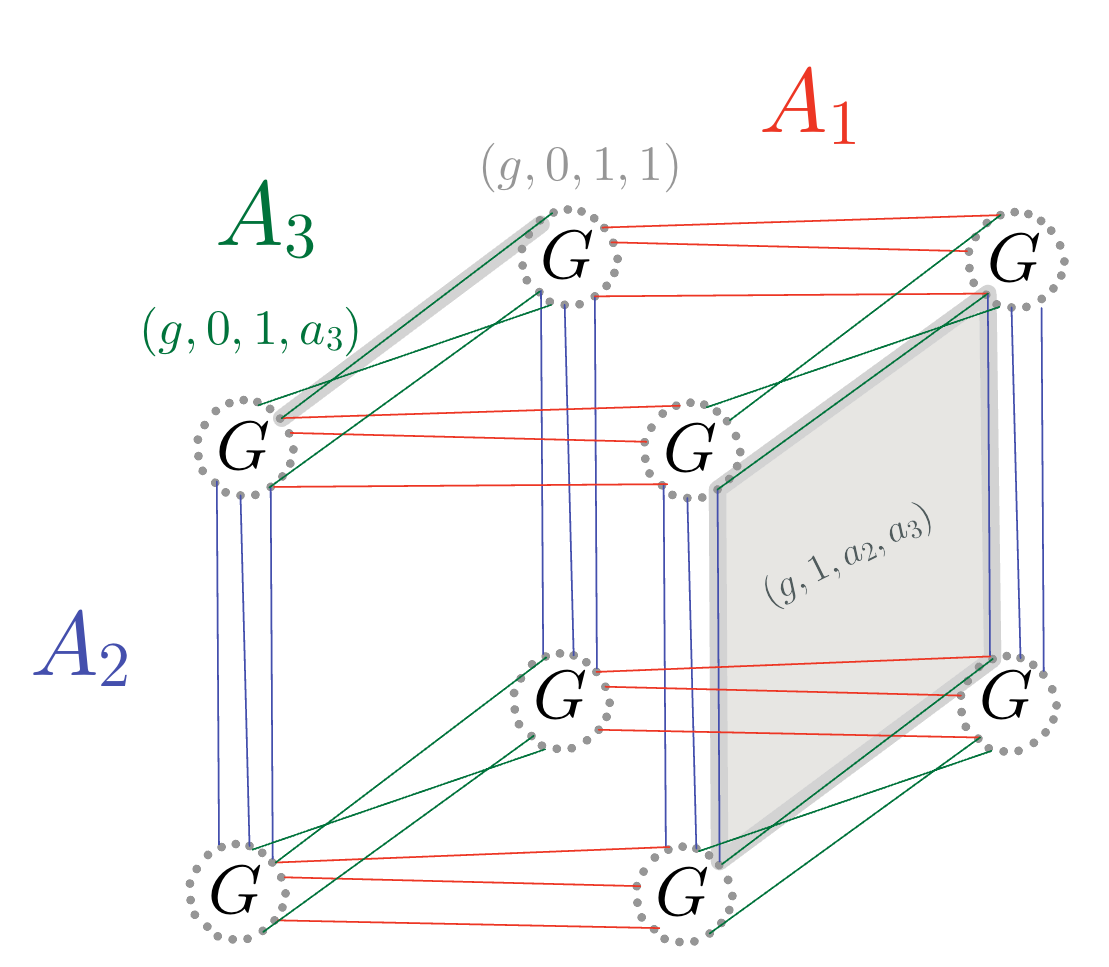}
    \caption{A cubical complex with $t=3$ and three sets of permutations $A_1,A_2,A_3$. The underlying graph can be seen to be $2^t=8$-partite. We label the faces as $(f_0,f_1,f_2,f_3)$.}\label{fig:geometry}
\end{figure}
\paragraph{Partial order}
We introduce the following partial order on faces of $X$. This partial order is the one that is induced by set inclusion, when each face is seen as a set of vertices as in~\eqref{eq:face-vertices}. Formally, given two faces $f,f'$ and $j\in\{1,\ldots,t\}$ we write $f' \precdot_j f$ if all the following hold:
    \begin{itemize}
        \item $f'_j \in \{0,1\}$ while $f_j\in  A_j$,
        \item $f_i = f'_i$ for all $i\notin \{0,j\}$,
        \item $f'_0= \begin{cases}
        f_0 & f'_j = 0\\
        f_0\cdot f_j & f'_j = 1\\
    \end{cases} $ .
        \end{itemize}
			        We sometimes suppress $j$ from the notation and write $f'\precdot f$ to indicate that there is some $j$ for which $f'\precdot_j f$. Note that if $f' \precdot f$ then necessarily $f$ is one dimension higher than $f'$. We also write $f' \prec  f$ if there is a sequence $f' \precdot  \cdots \precdot f$, and $f'\preceq f$ if $f'\prec f$ or $f'=f$.

\begin{lemma}
The relation $\prec$ defines a transitive graded partial order on faces. Moreover, with this relation $X$ is an incidence poset.
\end{lemma}

\begin{proof}
The relation $\preceq$ is reflexive and transitive. It is antisymmetric because if $f'\preceq f$ and $f\neq f'$ then $f$ must be of dimension at least one more than $f'$. To show that it is graded, we verify that the covering relation associated with $\preceq$ is precisely $\precdot$.

For the moreover part, note that if $f\precdot f'\precdot f''$, then there are exactly two distinct positive $i\neq j$ such that $f''_j\neq f_j$ and $f''_i\neq f_i$, and this determines exactly two possible values for $f'$.
\end{proof}

\begin{definition}\label{def:x-def}
Given the set $G$ and the sets of permutations $A_1,\ldots,A_t$, we let $X(G;\{A_i\})$ denote the graded incidence poset $X=\cup_{0\leq k \leq t} X(k)$. Here the rank function is given by $\rho(f)=|\type(f)|$, for $f\in X$.
\end{definition}

We end the section by introducing convenient notation.

\paragraph{Incidence maps}
The up incidence map $\up$ takes a $k$-face $f$ to all $(k+1)$-faces that cover it,
\[ \up(f) = \sett{f'}{ f' \succdot f }\;,\]
and is naturally extended to sets of faces by taking unions.
Similarly,
\[ \down(f) = \sett{f'}{ f' \precdot f }\;.\]

\paragraph{Links}
The upwards link of a face $v\in X$ is a sub-complex consisting of all faces above $v$,
\[X_{\geq v} = \sett{f\in X}{f\succeq v}.\]
Note that since the structure of $X$ is not simplicial, we do not remove $v$ from a face $f\succeq v$ when defining the link.
Similarly, the downwards link of a face $v\in X$ is a sub-complex consisting of faces below $v$,
\[X_{\leq v} = \sett{f\in X}{f\preceq v}.\]
For a subset $S\subseteq\{1,\ldots,t\}$ we denote $X_{\geq v}(S) = X_{\geq v}\cap X(S)$ and $X_{\leq v}(S) = X_{\leq v}\cap X(S)$. Similarly, $X_{\geq v}(k)=X_{\geq v}\cap X(k)$ and $X_{\leq v}(k)=X_{\leq v}\cap X(k)$.

\subsection{Local coefficients}
\label{sec:loc-coeffs}

To each face we associate a local coefficient space. For $i\in \{1,\ldots,k\}$ let $m_i\geq 1$ be an integer and $\hat{A}_i=\{1,\ldots,m_i\}$. The local coefficient space associated with a face of type $S$ is the vector space over a finite field $\F_q$ of characteristic $2$\footnote{Our construction can be applied more generally to arbitrary fields. However, we focus on finite fields with characteristic $2$, so that the corresponding code can be viewed as a code over $\F_2$. Moreover, the construction of the (co)boundary maps becomes more straightforward in this context, and we will comment further in Section~\ref{sec:co-maps}.}
\begin{equation}\label{eq:def-coeffs}
V_S = \F_q^{\prod_{j\in \ol{S}}\hat A_j}.
\end{equation}
So the local coefficient space of a $t$-dimensional face is $V_{\set{1,\ldots,t}}=\F_q$, and the local coefficient space of a vertex is ($\prod_{j=1}^t m_j$)-dimensional. We write $V_f = V_{\type(f)}$ for the local coefficient space of the face $f$, and $\sheaf=\{V_f\}$ for the collection of local coefficient spaces.

For $i\in\{0,\ldots,t\}$ we denote by $C_i(X,\mF)$, or $C_i(X)$ for short, the vector space of $i$-chains, that assign to every $i$-face $f$ a ``coefficient'' from $V_{f}$:
\[ C_i(X,\mF) \,=\, \bigoplus_{f\in X(i)} V_f\;.\]
Similarly, we use $C^i(X)$ to denote the $i$-cochains
\[ C^i(X,\mF) \,=\, \bigoplus_{f\in X(i)} V_f^*\;,\]
where $V_f^* = \{g:V_f \to \F\}$ is the dual vector space. Throughout we fix an identification of each $V_S^*$ with $V_S$ through the use of a self-dual basis which we label  $\{a:\,a\in \prod_{j\in \ol{S}}\hat{A}_j\}$, so $V_f^*\simeq V_f$ canonically.

\begin{definition}\label{def:x-complex}
Given a set $G$ and subsets of permutations $A_1,\ldots,A_t$ we let $C_*(G;\{A_i\};\{h_i\})$ be the length-$(t+1)$ chain complex
\[C_*(G;\{A_i\};\{h_i\}) \,=\, \bigoplus_i C_i(X,\sheaf)  \;.\]
We often write $C_*(X,\sheaf)$ or even $C_*(X)$ for this complex, when the parameters of the complex are clear from context. We use $C^*(X) = \oplus_i C^i(X)$ to denote the dual complex.
\end{definition}
\begin{figure}
    \centering
    \includegraphics[scale=0.9]{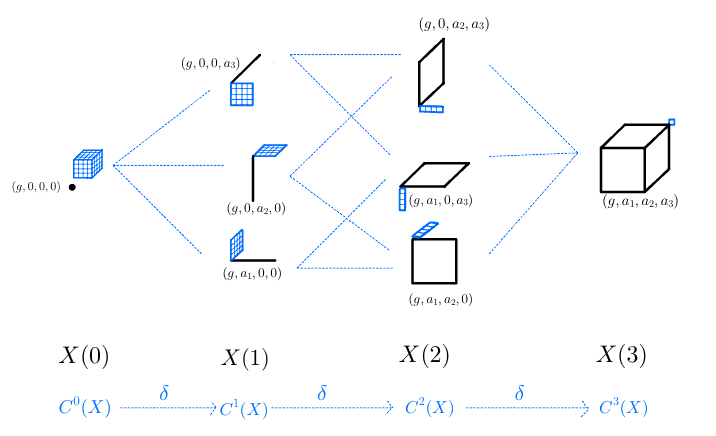}
    \caption{The co-chain complex, localized at a vertex $f=(g,0,0,0)$; with $t=3$. The geometric elements of $X_{\ge f}$ are drawn in black. The associated coefficient spaces are drawn in blue. }\label{fig:chain}
\end{figure}

The (co-)boundary maps associated with $C_*$ are defined in the next sub-section. Before proceeding, we can already compute the dimension of the space of $i$-chains $C_i(X,\sheaf)$.

\begin{lemma}\label{lem:x-dim}
Let $C_*(X)=C_*(G;\{A_i\};\{h_i\})$, $0\leq i \leq t$ and $D_i=\dim\, C_i(X,\sheaf)$. Then
\[ D_i \,=\, N \,  n^i 2^{t-i} \sum_{\substack{T\subseteq\{1,\ldots,t\}\\|T|=t-i}}\prod_{j\in T} m_j\;.\]
\end{lemma}

\begin{proof}
We have that $D_i = \sum_{f\in X(i)} \dim V_f$. For each $f\in X(i)$, $|\type(f)|=i$ and hence by~\eqref{eq:def-coeffs}, $\dim\, V_f = \prod_{j\notin \type(f)} m_j$. Moreover, for $|S|=i$ we have $|X(S)|=N  n^i 2^{t-i}$, where $N$ is to choose $g\in G$, $n^i$ is to choose $a_j$ for $j\in S$, and $2^{t-i}$ is to choose $b_j$ for $j\notin S$.
\end{proof}

\subsection{(Co)boundary maps}
\label{sec:co-maps}

For $i\in\{1,\ldots,t\}$ let  $h_i \in \F_q^{m_i\times n}$. Assume that $m_i\leq n_i$ and that $h_i$ has full row rank, i.e.\  the rows of $h_i$ are linearly independent. We index the rows of $h_i$ using the set $\hat{A}_i=\{1,\ldots,m_i\}$, and the columns of $h_i$ using the set $A_i$, and view $h_i$ as a linear map $h_i:\F_q^{A_i}\to\F_q^{\hat A_i}$.

\paragraph{Restriction and co-restriction maps}

Given $i\notin \type(f')$ and $f' \precdot_i f$, define the \emph{co-restriction map} 
\[ \corest_{f',f}\,:\quad V_{f'}^*\to V_f^*\]
as follows. For any $z\in V_{f'}^*$, $\corest_{f',f}(z)$ is obtained by applying $\Id_{-i}\otimes h_{i}^T$, resulting in a vector in $\F_q^{( \prod_{j\not \in \type(f')} \hat A_j)\times A_i}$, and then setting the $A_i$-coordinate to $f_i$. Formally,
\begin{equation}\label{eq:co-rest-1}
 \corest_{f',f}(z) =
    \big((\Id_{-i}\otimes h_{i}^T)z\big) [ \cdots ,f_i, \cdots]\;.
\end{equation}

Similarly, for $i\in \type(f')$ and $f' \succdot_i f$ define the \emph{restriction map}
\[ \rest_{f',f}: V_{f'}\to V_f\]
for any $z\in V_{f'}$ by
\begin{equation}\label{eq:rest-1}
 \rest_{f',f}(z) = z \otimes (h_i f'_i)\;,
\end{equation}
    which is an element of $\F_q^{ \prod_{j\notin \type(f')} \hat A_j} \otimes \F_q^{\hat A_i}\cong \F_q^{ \prod_{j\notin \type(f)} \hat A_j}$.
    We abuse the notation and write $f'_i$ as the canonical basis vector of $\F_q^{A_i}$ corresponding to $f'_i \in A_i$.
    We verify below that the maps $\corest$ and $\rest$ are adjoint to each other (see~\eqref{eq:adjoint-1}). We extend the notation by defining, for any sequence $f_i\precdot f_{i+1}\cdots \precdot f_k$ where $f_j\in X(j)$,
    \[  \rest_{f_k,f_i}(z) \,=\, \rest_{f_{i+1},f_i}(\cdots  \rest_{f_k,f_{k-1}}((z)))\;.\]
We note that this leads to a well-defined quantity, i.e.\ $\rest_{f_k,f_i}(z) $ is independent of the path $f_i\precdot f_{i+1}\cdots \precdot f_k$, as can easily be verified from the definition~\eqref{eq:rest-1}. This is because application of a linear map along different coordinates, or restriction on different coordinates, are operations that commute. A similar extension is performed to define co-restriction maps $\corest_{f',f}$ for general $f'\prec f$.

\paragraph{The coboundary map}%
The coboundary map $\delta: C^*(X)\to C^*(X)$ is defined as
\begin{equation}\label{eq:def-delta}
\delta(z) (f) = \sum_{f'\precdot f} \corest_{f',f}( z(f'))\;,
\end{equation}
where $\corest_{f',f}$ is as in~\eqref{eq:co-rest-1}.
In words, at a face $f\in X$, we go over all faces $f'$ below $f$ and map $z(f')$ to an element in $V_f^*$ via the restriction map $\rest_{f',f}$, and then sum  all the contributions thus obtained.
Because $\F_q$ has characteristic $2$, $\delta \circ \delta = 0$, as required by the definition of the coboundary map.\footnote{For general fields $\F_q$ with characteristics not equal to $2$, the maps requires additional signs so that $\delta \circ \delta = 0$.
Specifically, $\delta(z) (f) = \sum_{f'\precdot f} (-)^{s(f,f')} \corest_{f',f}( z(f'))$ for some $s$.}
\paragraph{The boundary map}
The boundary map $\partial:C_*(X)\to C_*(X)$ is defined as
\[
\partial(z) (f) = \sum_{f'\succdot f} \rest_{f',f}( z(f')).
\]

The next lemma shows that the map $\partial$ is the adjoint of the map $\delta$ with respect to the natural inner product on $C(X)$.

\begin{lemma}\label{lem:adjoint}
For any $z\in C_*(X)$ and $z'\in C^*(X)$ it holds that
\[\langle z,\delta(z')\rangle \,=\, \langle \partial(z),z'\rangle \;.\]
\end{lemma}

\begin{proof}
It is enough to verify the equality for $z\in V_f$ and $z'\in V_{f'}^*$. By definition $\delta:V_{f'}^* \to \oplus_{f''\succdot f'} V_{f''}^*$ and $\partial: V_{f}:\oplus_{f''\precdot f} V_{f''}$. Thus both sides of the equality are non-zero only if $f'\precdot f$. Suppose that this is the case, and let $i$ be such that $\type(f)=\type(f')\cup\{i\}$. It remains to verify that
\begin{equation}\label{eq:adjoint-1}
\langle z,\corest_{f',f}(z')\rangle \,=\, \langle \rest_{f,f'}(z),z'\rangle\;.
\end{equation}
According to~\eqref{eq:co-rest-1}
\[ \corest_{f',f}(z') = (\Id_{-i}\otimes h_{i}^T)z' [ \cdots, f_i, \cdots]\;,\]
and according to~\eqref{eq:rest-1}
\[ \rest_{f,f'}(z) = z \otimes (h_i f_i) \;.\]
Thus the left-hand side of~\eqref{eq:adjoint-1} evaluates to
\begin{align*}
\langle z,(\Id_{-i}\otimes h_{i}^T)z' [ \cdots, f_i, \cdots]\rangle
&= \langle z\otimes f_i, (I_{-i} \otimes h_i^T)z' \rangle\\
&=\langle (I_{-i} \otimes h_i)(z\otimes f_i), z' \rangle\\
&=\langle z \otimes (h_i f_i),z'\rangle\;,
\end{align*}
which matches the right-hand side.
\end{proof}

\subsection{Statement of results}
\label{sec:results}

Having described our construction, we can state our main result about it.

\begin{theorem}\label{thm:main}
Let $(G;\{A_i\})$ be a set and collections of pairwise commuting permutations on it such that the graphs $\Cay(G,A_i)$ are each $\lambda$-expanding up to size $r |G|$ (see Definition~\ref{def:x-expand}). Let $N=|G|$ and $n=|A_1|=\cdots=|A_t|$. Let $(h_1,\ldots,h_t)$ be a collection of matrices $h_j\in \F_q^{m_j\times n}$ such that the family $(h_1,\ldots,h_t)$ is two-way $\kappa$-robust (see Definition~\ref{def:h-robust}). Let $C=C_*(G;\{A_j\};\{h_j\})$ be the length-$(t+1)$ chain complex over $\F_q$ from Definition~\ref{def:x-complex}. Then there is a $c=\kappa^t \exp(-O(t^2))$ such that $C$ has the following properties.
\begin{enumerate}
\item The space $C_k$ of $k$-chains has dimension $D_k=  N \  n^k 2^{t-k} \sum_{\substack{T\subseteq\{1,\ldots,t\}\\|T|=t-k}}\prod_{j\in T} m_j$.
\item The co-chain complex $C^*$ has  co-systolic distance $\mu_\cosyst(k)$ for every $0\leq k \leq t-1$ and co-cycle expansion $\eps_\cocyc(k)$  for every $0\leq k \leq t-2$, where
\begin{equation}\label{eq:main-1}
  \mu_\cosyst(k) \,\geq\,   \frac{c-\lambda}{2^{2t}n^t} r |X(k)|  \qquad\text{and}\qquad \eps_\cocyc(k) \,\geq\,   \frac{c-\lambda}{(2nt)^{2t+1}} r\;.
\end{equation}
\item The chain complex $C_*$ has systolic distance  $\mu_\syst(k)$ for every $1\leq k \leq t$ and cycle expansion  $\eps_\cyc(k)$ for every $2\leq k \leq t$, where
  \begin{align}
    \mu_\syst(k) &\geq   \frac{1}{(2nt)^t} \mu_\cosyst(t-k)\;, \notag\\
    \eps_\cyc(k) &\geq \min\Big\{ \frac{ \eps_\cocyc(t-k)}{2(t^22^{2t} n^{t+1})^t}\,,\; \frac{1}{(2nt)^t}\frac{\mu_\cosyst(t-k+1)}{|X(k)|}\Big\}\;.\label{eq:main-2}
  \end{align}
\end{enumerate}
\end{theorem}

We note that we did not attempt to overly optimize the bounds in Theorem~\ref{thm:main}; in particular the dependence on $t$ may be loose. To understand the bounds, one can consider that the main asymptotic parameter is $N=|G|$, which grows to infinity. We think of $t$ as a constant, such as $t=4$.
The parameter $n$ is also ideally a constant, but the parameter $r$ is not, $r = \Omega((\log N)^{-(t-1)})$.
This is chosen such that $\lambda$ is sufficiently small compared to $\kappa$ so that the bounds in~\eqref{eq:main-1} are positive.
This is the origin of the loss of the polylogarithmic factor in distance and soundness.


\begin{proof}[Proof of Theorem~\ref{thm:main}]
The first item is shown in Lemma~\ref{lem:x-dim}. 
 For the second item, we first note the bound
\[ \dist_\coloc(k) \,\ge\, \frac{1-\lambda C_1}{C_2} r |X(k)|\;,\]
valid for $0\leq k < t$, that follows from Proposition~\ref{prop:codistance-t}. The prefactor on the right-hand side can be re-written as
\begin{align*}
  \frac{1-\lambda C_1}{C_2} &= \big(1/C_1-\lambda \big) \frac{C_1}{C_2}\\
  &\geq \frac{1}{2^{2t}n^t} ( \kappa^{t}2^{-4t^2} - \lambda)\;,
\end{align*}
where the second line uses $t^2 2^{2t}\leq 2^{3t^2}$ for $t\geq 1$. We then use Lemma~\ref{lem:cosys-exp} to conclude the claimed bounds. For $\eps(k)$, we also use $|X(k+1)|/|X(k)| \leq 2nt$ and $|X_{\geq v}(k)|\leq (2n)^t$, which follow from Lemma~\ref{lem:geom-2}. Finally, the bounds on $\tilde{\eps}(k)$ and $\tilde{\mu}(k)$ follow immediately from Proposition~\ref{prop:distance}.
 \end{proof}

We can apply Theorem~\ref{thm:main} to obtain quantum codes with parameters described in the following corollary. Since a quantum CSS code can be defined from any length-$3$ complex, we have the choice to locate the qubits on the space of $k$-chains for any $i\in \{1,\ldots,t-1\}$. However, to obtain a bound on the local testability soundness of the quantum code, we need a complex of length at least $5$. To simplify the discussion of parameters we assume that, of the codes $\mC_i=\ker h_i$, exactly $i$ have rate close to $1$, and the remaining $(t-i)$ have rate close to $0$. This ensures that most of the dimension of the space $C_i$ is concentrated on a single type of $i$-face, and guarantees that the resulting code has positive dimension. Of course, other settings of parameters are possible.

\begin{corollary}\label{cor:main-code}
With the same setup as in Theorem~\ref{thm:main}, assume $t\geq 4$ and fix $2\leq i \leq t-2$. Assume that the parameters $m_1,\ldots,m_t$ and $n$ are chosen such that $m_j=\Theta(n)$ for all $j$, and furthermore $m_j/m_{j'}$ is at least a sufficiently large constant (depending on $t$) when $j\in\{1,\ldots,i\}$ and $j'\in \{i+1,\ldots,t\}$. Assume that, for these parameters there is a family $(h_1,\ldots,h_t)$ of two-way $\kappa$-robust codes, where $\kappa$ is small enough, as a function of $\lambda$ and $t$, so that the lower bounds in~\eqref{eq:main-1} are positive.\footnote{We emphasize that one can generally expect to achieve $\lambda$ that is an arbitrarily small constant, or even an inverse polynomial in the degree $n$. For constant $t$, the requirement on $\kappa$ is simply that it should be a small enough constant, or even just that $\kappa\geq 1/n^\eps$ for small enough $\eps>0$.}
Let $\code = CSS(\delta_i,\partial_i)$ be the CSS code associated with the space $C_i$ and the boundary and co-boundary operators on it,
  where $C_i$ is viewed as a $\F_2$-vector space and $\delta_i,\partial_i$ are viewed as $\F_2$-linear maps.
Then $\code$ is an $[\ell_i,k_i,d_i]$ CSS code with parity checks of weight $O((\log_2 q) tn)$ and such that
\begin{enumerate}
\item $\ell_i = (\log_2 q) D_i = (\log_2 q) N \cdot \Theta(n)^t$,
\item $k_i = \Omega((\log_2 q) \ell_i)$,
\item $d_i = r N\cdot (nt)^{-O(t)} \exp(-O(t^2))$.
\end{enumerate}
Moreover, $\code$ is a qLTC with soundness $\rho \geq r (\log_2 q)^{-1} (nt)^{-O(t)} \exp(-O(t^2))$.
\end{corollary}

\begin{proof}
The first item follows by definition of $\code$. For the second item, we will prove the lower bound on $k_i$ in Section~\ref{sec:rate} (see end of Section~\ref{sec:shift}). As shown in Lemma~\ref{lem:soundness-param} the distance is $d_i \geq\min \{\mu_\cosyst(i),\,\mu_\syst(i)\}$, giving the stated bound by Theorem~\ref{thm:main}. Similarly, the soundness is bounded using Lemma~\ref{lem:soundness-param}.
\end{proof}

As discussed earlier, if $n,t,q,r$ are all constant then we obtain a family of quantum LTC with constant weight parity checks, constant rate, relative distance and soundness. While the parameters $n,t,q$ can be chosen to be constants, we do not know how to construct the required $4$-dimensional cubical complex, with the right expansion properties, for constant $r$. However, as discussed in Section~\ref{sec:particular-instance-of-complex} we can construct it for $r = \Omega((\log |G|)^{-(t-1)})$.

\subsection{Instantiating $G$ and the $\{A_i\}$}
\label{sec:particular-instance-of-complex}

We provide two examples of $G$ and pairwise commuting $A_1, ..., A_t$. The first example is easier to understand; the second example has better parameters, and in particular leads to our main application, stated as Theorem~\ref{thm:qltc} in the introduction.

For both constructions, given target rates $(n-m_j)/n$ for the local codes $\code_j$ we fix a family $(h_1,\ldots,h_t)$ of parity check matrices for codes $\code_1,\ldots,\code_j\subseteq\F_q^n$ such that the $\code_j$ have the desired rate and moreover the family $(h_1,\ldots,h_t)$ is two-way $\kappa$-robust according to Definition~\ref{def:h-robust}. It follows from Theorem~\ref{thm:PK-future} and Lemma~\ref{lemma:prodexp-robust} that by choosing $q$ as a large enough power of $2$, a random $t$-tuple of codes satisfies this property for some $\kappa>0$ independent of $n$.

\subsubsection{First construction}

The first construction is obtained from Cayley graphs over $(\Z/2\Z)^{\log N}$.

Let $G = (\Z/2\Z)^{\log N}$.
Let $n=\lceil c \log N\rceil$, for a sufficiently large constant $c$ such that there exist sets $A_1,\ldots,A_t$ of size $n=|A_j|$ such that the graphs $\Cay(G,A_j)$ have spectral expansion at most $\lambda$, where $\lambda$ is sufficiently small compared to the robustness parameter $\kappa$. Since $\kappa$ is a constant independent of $n$ then we can always make $\lambda$ small enough by increasing $c$, independently of $n$. This is because it is well-known that for any target $\lambda$ (possibly depending on $\kappa$) there is a constant $c(\lambda)=O(1/\lambda^2)$ and subsets $A_j$ of $G$ of size $|A_j|=c(\lambda)\log|G|$ such that $\Cay(G,A_j)$ is $\lambda$-expanding. (This can be shown by a probabilistic argument~\cite{alon1994random}, and also by explicit constructions~\cite{AlGoHaPe}.)

Note that this construction is not ``fully explicit,'' because we do not have a way to check in time $\poly(N)$ that a given choice of codes $(h_1,\ldots,h_t)$ is indeed two-way $\kappa$-robust. This is because according to Theorem~\ref{thm:PK-future} we need to choose $q=2^{(n+3)^t}$, which is super-polynomial in $N$.

\subsubsection{Second example}
\label{sec:ex2}

The second example is beyond Cayley graphs, but can be described as abelian lifts and is tightly related to Cayley graphs over $\Z/\ell \Z$.

Let $G_0=(V_0,E_0)$ be an $n$-regular Cayley expander graph on $n'$ vertices. Based on known explicit constructions of expanders, $n$ can be chosen any large enough constant, and $n'$ arbitrarily large.
In \cite[Theorem 1.3]{JeronimoMO0T22} the authors describe an $H$-lift of this graph for an abelian group $H$, $|H| = \exp(\Theta(n'))$, which is given by associating with each edge $e\in E_0$ an element $s\in H$. This graph, which we call $R$, has $N = |V_0| |H|$ vertices and has second largest eigenvalue $O(\sqrt{n} \log n)$, making it nearly-Ramanujan.
Furthermore, the graph $R$ can be constructed efficiently in time polynomial in $N$.

Let $G_1=(V_1,E_1)$ be the double cover of $G_0$, with vertex set $V_1=V_0\times\set{0,1}$. Furthermore, let us assume a fixed numbering of the $n$ edges out of each vertex, so that for each $i\in[n]$ we have an action on $H\times V_0\times \set{0,1}$ obtained by taking an element $(h,v,b)$ to $(s\cdot h,v',1-b)$ where $(v',1-b)$ is the $i$-th neighbor of $(v,b)$ in $G_1$ and $s\in H$ is associated with the $i$-th edge out of $v$. This set of $n$ actions gives a $n$-regular graph that is the $H$-lift of $G_1$ given by the association of $H$-elements to edges.

We now move to the $t$-fold product, letting $G= H\times V_0^t$, and letting $X(0) = G \times \set{0,1}^t$. We have $t$ sets $A_1,\ldots,A_t$ of permutations on $X(0)$, where $A_j= \set{s_j^1,\ldots,s_j^d}$ and
\[ (h,v_1,\ldots,v_t,b_1\ldots,b_t) \stackrel {s_j^i}{\to} (s\cdot h,v_1,\ldots ,v'_j,\ldots,v_t,b_1,\ldots ,1-b_j,\ldots,b_t)  \] where $(v'_j,1-b_j)$ is the $i$-th neighbor of $(v_j,b_j)$ in $G_1$, and $s$ is associated with the $i$-th neighbor of $v_j$.

It is easy to check that every pair of permutations $s\in A_j$ and $s'\in A_{j'}$ pairwise commute, and therefore $X_0$ together with $A_1,\ldots,A_t$ give rise to a $t$-dimensional complex $X$ as described in Section~\ref{sec:geometry}.

We make the following two observations about this construction. Firstly, for any fixed $j\in [t]$ the graph obtained by connecting every $v\in X(0)$ to the $n$ vertices obtained by the $n$ permutations from $A_j$ is a graph consisting of $V_1^{t-1}$ copies of the expander graph $R$. Since $|V_1|=O(\log N)$, we obtain that the graphs $\Cay(G,A_i)$ are $\lambda$-expanding up to size $r|G|$, where $\lambda$ is the normalized second eigenvalue of $R$ and $r=\Omega(1/(\log N)^{t-1})$. Secondly, $n$ can be chosen as any constant large enough such that if $\lambda$ is the normalized second eigenvalue of $R$ then $\lambda$ is small enough compared to $\kappa$ for the lower bounds in~\eqref{thm:main} to be positive. This is always possible because the universal constant $\kappa$ obtained from Theorem~\ref{thm:robust} does not depend on $n$ (only on $t$), whereas $\lambda$ goes to zero as $n$ grows. Because $n$ is a constant, we can exhaustively search for a tuple $(h_1,\ldots,h_t)$ that is two-way $\kappa$-robust.

Combining these observations lead to the following main instantiation of Theorem~\ref{thm:main}.

\begin{corollary}\label{cor:main-cor}
Let $t\geq 2$ be an integer. There is a polynomial-time algorithm that on input an integer $N$ (in unary) returns a description of $(G;\{A_i\})$ and $\{h_i\}$ as above such that the following hold.
\begin{enumerate}
\item The space $C_k$ of $k$-chains has dimension $D_k= \Theta(N)$\ ;
\item The co-chain complex $C^*$ has  co-systolic distance $\mu_\cosyst(k) = \Theta(|X(k)|/(\log N)^{t-1})$ for every $0\leq k \leq t-1$ and co-cycle expansion $\eps_\cocyc(k)=\Theta(1/(\log N)^{t-1})$  for every $0\leq k \leq t-2$\ ;
\item The chain complex $C_*$ has systolic distance  $\mu_\syst(k) = \Theta(|X(k)|/(\log N)^{t-1})$ for every $1\leq k \leq t$ and cycle expansion  $\eps_\cyc(k) = \Theta(1/(\log N)^{t-1})$ for every $2\leq k \leq t$.
\end{enumerate}
Here, the constants implicit in the $\Theta(\cdot)$ notation can depend on $t$.
\end{corollary}

We also obtain a corresponding instantiation of Corollary~\ref{cor:main-code}, which leads to the previously stated Theorem~\ref{thm:qltc}.

\begin{corollary}\label{cor:main-code-explicit}
There is a polynomial-time algorithm that on input an integer $N$ (in unary) returns explicit descriptions of parity check matrices $H_X$ and $H_Z$ on $\Theta(N)$ qubits such that the following hold:
\begin{enumerate}
\item The weight of each row of $H_X$ and $H_Z$ is $\Theta(1)$\ ;
\item The code $\code=CSS(H_X,H_Z)$ has dimension $\Theta(N)$ and distance $\Theta(N/(\log N)^3)$\ ;
\item $\code$ is a qLTC with soundness $\rho=\Omega(1/(\log N)^3)$.
\end{enumerate}
\end{corollary}

\section{A lower bound on the code rate}
\label{sec:rate}

A natural approach to bounding the code dimension is to count the number of qubits and the number of X and Z checks.
The number of qubits, X and Z checks are
  $\dim_{\F_2} C_i(X, \cF)$, $\dim_{\F_2} C_{i-1}(X, \cF)$, and $\dim_{\F_2} C_{i+1}(X, \cF)$, respectively.
This leads to a lower bound on the code dimension:
  $k \ge \dim_{\F_2} C_i(X, \cF) - \dim_{\F_2} C_{i-1}(X, \cF) - \dim_{\F_2} C_{i+1}(X, \cF)$.
However, when applied to our construction,
  this bound yields a negative value, making it trivial.

To obtain a nontrivial lower bound on the code dimension,
  we need to take into account the linear dependencies among the X and Z checks.
Rather than simply counting the total number of checks,
  we refine our estimate by counting only the number of linearly independent checks.
We formalize this refinement by establishing a correspondence between two chain complexes:
  our original complex, where the $i$-cochain consists of $i$-cells,
  and an alternative complex that combines cells of different dimensions.
This relation is analogous to the relation between
  the constructions in \cite{DHLV} and \cite{PK2}.
Conceptually, the new chain complex reorganizes the qubits and checks to remove redundant checks that are linearly dependent on others.
Within this modified framework,
  the counting of qubits and checks leads to a nontrivial lower bound on the code dimension.

\subsection{Shifting the complex}
\label{sec:shift}

We now describe the construction of the new chain complex.
Recall that our main construction defines
  $C^i(X, \cF) = \bigoplus_{f \in X(i)} V_f$
  where we take all cells of dimension $i$ to form $C^i$.
Notice that the dimension of a cell $f$ can be expressed as $|\type(f)|$
  and the space of all possible types $\type(f) \subseteq [t]$
  can be viewed as a $t$-dimensional hypercube.
Therefore, the process of forming a chain complex from a cubical complex
  can be rephrased as choosing a diagonal direction of the $t$-dimensional hypercube.
In the case of $C^i(X, \cF)$,
  the diagonal direction goes from $(0, 0, ..., 0)$ to $(1, 1, ..., 1)$.

More generally, given a subset $T \subseteq [t]$,
  it specifies a direction from $1_T$ to $1_{[t] - T}$,
  where $1_T$ is the vector whose $i$-th coordinate is $1$ if $i \in T$ and $0$ otherwise.
This induces the following chain complex:
\begin{equation}
  C^i(X, \cF, T) = \bigoplus_{f \in X, |\type(f) \Delta T| = i} V_f
\end{equation}
where $\Delta$ is the symmetric difference of the two sets.
In particular, $C^i(X, \cF) = C^i(X, \cF, \emptyset)$.

The coboundary map of this new chain complex,
  $\delta: C^i(X, \cF, T) \to C^{i+1}(X, \cF, T)$,
  is again a map between faces with incidence relations,
  just like the original chain complex.
In particular, the coboundary maps can be decomposed into maps consist of
  $\delta_{f, f'}: V_f \to V_{f'}$,
  for $f, f' \in X$,
  $|\type(f) \Delta T| = i$,
  and $|\type(f') \Delta T| = i+1$.
If $f \succdot f'$,
  we set the map to be $\corest_{f, f'}$,
  as defined in~\eqref{eq:co-rest-1}.
If $f \precdot f'$,
  we set the map to be $\rest_{f, f'}$,
  as defined in~\eqref{eq:rest-1}.
If there is no incidence relation between $f$ and $f'$,
  we set the map to be $0$.

The coboundary maps satisfy $\delta \circ \delta = 0$
  because, in some sense, the corestriction and restriction maps
  commute with themselves and with each other.
To check $\delta \circ \delta = 0$,
  it is equivalent to check
  $\sum_{f' \in X} \delta_{f', f''} \delta_{f, f'} = 0$
  for every $f, f'' \in X$.
The LHS is nontrivial only when $f$ and $f''$ are both incident to a common face and
\begin{itemize}
  \item $\type(f'') = \type(f) \cup \{j, k\}$,
  \item $\type(f'') = \type(f) \cup \{j\} - \{k\}$,
  \item $\type(f'') = \type(f) - \{j, k\}$,
\end{itemize}
for $j, k \in [t]$ and $j \ne k$.
In each case, there are exactly two $f'$, say $f_1', f_2'$, with nontrivial $\delta_{f', f''}$ and $\delta_{f, f'}$.
What remains is to verify is the identity
  $\delta_{f_1', f''} \delta_{f, f_1'} = \delta_{f_2', f''} \delta_{f, f_2'}$,
  which follows from the definitions of the restriction and corestriction maps given in Section~\ref{sec:co-maps}.

This new family of chain complexes is, in a certain sense, equivalent to
  the original chain complex from our main construction.
In particular, for any subset $T \subseteq [t]$,
  there exists a sheaf $\cF_T$ (which we will specify below)
  such that the chain complex $C^*(X, \cF)$ is chain homotopy equivalent to
  $C^*(X, \cF_T, T)$.
Furthermore, this chain homotopy equivalence is ``sparse'',
  which implies that the code parameters induced from the two chain complexes
  are the same up to constant factors.
That means if we only care about
  the asymptotic scaling of the code parameters
  it is sufficient to study $C^*(X, \cF)$.
See \cite[App A]{hastings2021fiber} and \cite[Sec 2.1]{lin2024transversal}
  for more detail on the notion of chain homotopy equivalence.

We now formally define $\cF_T$.
We write $\cF(\{h_i\}_{i=1}^t)$ to denote the sheaf constructed from
  $h_i: \F_q^{A_i} \to \F_q^{\hat A_i}$.
Recall that $h_i^\perp: \F_q^{A_i} \to \F_q^{A_i - \hat A_i}$
  is a linear map that satisfies $h_i^\perp h_i^T = 0$,
  i.e. the code with parity check matrix $h_i^\perp$ is dual to the code with parity check matrix $h_i$.
Let $\cF_T$ be the sheaf $\cF(\{h_i'\}_{i=1}^t)$,
  where $h'_i = h_i$ if $i \notin T$
    and $h'_i = h_i^\perp$ if $i \in T$.
In some sense,
  the structures are flipped when $i \in T$.

Rather than establishing a full chain homotopy equivalence,
Instead, for the purpose of deriving a lower bound on the dimension
  it suffices to prove a weaker statement:
  that the homologies of the two complexes have the same dimension. In the next subsection we show the following.

\begin{theorem} \label{thm:dimension}
  For any $0 \le i \le t$ and $T \subseteq [t]$
  \begin{equation}
    \dim\, H^i(X, \cF) = \dim\, H^i(X, \cF_T, T).
  \end{equation}
\end{theorem}

Using the theorem, we can show the lower bound on the rate $k_i$ claimed in Corollary~\ref{cor:main-code}. For this we require the following simple estimate.

\begin{lemma}
  Let $D_i' = \dim\,C^i(X, \cF(h'), T)$. Then
  \begin{equation}
    D_i' = N \sum_{S \subseteq [t], |S \Delta T| = i} n^{|S|} 2^{t-|S|} \prod_{j \notin S} m'_j
  \end{equation}
  where $m'_j$ is the dimension of the codomain of $h'_j$.
\end{lemma}
\begin{proof}
  We have that $D_i' = \sum_{f \in X, |\type(f) \Delta T| = i} \dim\, V_f$.
  For each $f$, $\dim\, V_f = \prod_{j \notin \type(f)} m'_j$.
  Moreover, we have $|X(\type(f))| = N n^{|\type(f)|} 2^{t-|\type(f)|}$,
    where $N$ is to choose $g \in G$,
    $n^{|\type(f)|}$ is to choose $a_j$ for $j \in \type(f)$,
    and $2^{t-|\type(f)|}$ is to choose $0$ or $1$ for $j \not\in \type(f)$.
\end{proof}

The results above lead to the bound on the code dimension claimed in Corollary~\ref{cor:main-code}.

\begin{proof}[Proof of the lower bound on $k_i$ in Corollary~\ref{cor:main-code}]
  \label{proof:rate}
  Choose $m_j = \nu n$ for $1 \le j \le i$
    and $m_j = (1-\nu) n$ for $i+1 \le j \le t$.
  Let $T = \{i+1, ..., t\}$.
  We have the following bound on $k_i$
  \begin{equation}
    k_i = \dim\, H^i(X, \cF) = \dim\, H^i(X, \cF_T, T) \ge D_i' - D_{i-1}' - D_{i+1}'
  \end{equation}
  Notice that $m_j' = \nu n$ for $1 \le j \le t$.
  Thus, $D_i' \ge N n^t$ from the term $S = \{1, ..., t\}$ alone.

  Notice that $\sum_{j=0}^t D_j' = N \prod_{j=1}^t (n + 2 m_j') = N n^t (1 + 2\nu)^t$.
  Thus, $D_{i-1}' + D_{i+1}' \le \sum_{j=0}^t D_j' - D_i' \le N n^t (1 + 2\nu)^t - N n^t$
    which implies $k_i \ge D_i' - D_{i-1}' - D_{i+1}' \ge 2 N n^t - N n^t (1 + 2\nu)^t$.
  Therefore, as long as $\nu < (2^{1/t} - 1) / 2$, we have $k_i = \Theta(\ell_i)$.
  When evaluating the dimension over $\F_2$,
    this leads to $\Theta((\log_2 q)\ell_i)$.
\end{proof}

\subsection{Proof of Theorem~\ref{thm:dimension}}

We prove Theorem~\ref{thm:dimension} using the following claims.

\begin{claim} \label{claim:ineq}
  For any $0 \le i \le t$ and $0 < j \le t$,
  \begin{equation}
    \dim\, H^i(X, \cF_{[j-1]}, [j-1]) \le \dim\, H^i(X, \cF_{[j]}, [j]).
  \end{equation}
\end{claim}

\begin{claim} \label{claim:eq}
  For any $0 \le i \le t$ and $0 \le j \le t$,
  \begin{equation}
    \dim \,H^i(X, \cF_{[j]}, [j]) = \dim\, H^{t-i}(X, \cF^\perp_{[t] - [j]}, [t] - [j]),
  \end{equation}
  where $\cF^\perp$ is the sheaf where every $h_k$ is replaced with $h_k^\perp$.
\end{claim}

\begin{proof}[Proof of Theorem~\ref{thm:dimension}]
  We first show the following corollary.
  \begin{corollary} \label{cor:eq}
    For any $0 \le i \le t$ and $0 < j \le t$,
    \begin{equation*}
      \dim\, H^i(X, \cF_{[j-1]}, [j-1]) = \dim\, H^i(X, \cF_{[j]}, [j]).
    \end{equation*}
  \end{corollary}
  \begin{proof}[Proof of Corollary~\ref{cor:eq}]
    By Claim~\ref{claim:ineq},
      we have
      $\dim H^i(X, \cF_{[j-1]}, [j-1]) \le \dim H^i(X, \cF_{[j]}, [j])$.
    Additionally, we have
      \begin{equation*}
        \dim H^{t-i}(X, \cF^\perp_{[t]-[j]}, [t]-[j]) \le \dim H^{t-i}(X, \cF^\perp_{[t]-[j-1]}, [t]-[j-1]),
      \end{equation*}
      where we reverse the ordering of the $t$ directions of the cubical complex,
      which maps $[t-j]$ to $[t] - [j]$ and $[t-j+1]$ to $[t]-[j-1]$.
    Together with Claim~\ref{claim:eq}, we have
      \begin{equation*}
        \dim H^i(X, \cF_{[j]}, [j]) \le \dim H^i(X, \cF_{[j-1]}, [j-1]).
      \end{equation*}
    Thus, $\dim H^i(X, \cF_{[j-1]}, [j-1]) = \dim H^i(X, \cF_{[j]}, [j])$
      as desired.
  \end{proof}

  By applying the corollary inductively, for any $0 \le j \le t$, we have
  \begin{equation*}
    \dim H^i(X, \cF) = \dim H^i(X, \cF_{\emptyset}, \emptyset) = \dim H^i(X, \cF_{[j]}, [j]).
  \end{equation*}
  Rearranging the order of the $t$ directions of the cubical complex
    maps $[j]$ to $T$,
    which gives the desired result,
    $\dim H^i(X, \cF) = \dim H^i(X, \cF_{T}, T)$.
\end{proof}

We now prove the two claims.
We first prove Claim~\ref{claim:eq}, which is more straightforward.
\begin{proof}[Proof of Claim~\ref{claim:eq}]
  The main observation is that
    the two chain complexes $C^*(X, \cF_{[j]}, [j])$
    and $C_{t-*}(X, \cF^\perp_{[t]-[j]}, [t]-[j])$
    are exactly the same.
  In particular,
  \begin{align*}
    &\phantom{{}={}} C^i(X, \cF_{[j]}, [j]) \\
    &= \bigoplus_{f \in X, |\type(f)\symdiff [j]| = i} V_f \\
    &= \bigoplus_{f \in X, |\type(f)\symdiff [j]| = i}
      \F_q^{\prod_{k \in \type(f) \setminus [j]} \hat A_k}
      \otimes \F_q^{\prod_{k \in \type(f) \cap [j]} A_k - \hat A_k}.
  \end{align*}
  On the other hand,
  \begin{align*}
    &\phantom{{}={}} C_{t-i}(X, \cF^\perp_{[t]-[j]}, [t]-[j]) \\
    &= \bigoplus_{f \in X, |\type(f)\symdiff ([t]-[j])| = t-i} V^\perp_f \\
    &= \bigoplus_{f \in X, |\type(f)\symdiff [j]| = i} V^\perp_f \\
    &= \bigoplus_{f \in X, |\type(f)\symdiff [j]| = i}
      \F_q^{\prod_{k \in \type(f) \setminus ([t]-[j])} A_k - \hat A_k}
      \otimes \F_q^{\prod_{k \in \type(f) \cap ([t]-[j])} \hat A_k} \\
    &= \bigoplus_{f \in X, |\type(f)\symdiff [j]| = i}
      \F_q^{\prod_{k \in \type(f) \cap [j]} A_k - \hat A_k}
      \otimes \F_q^{\prod_{k \in \type(f) \setminus [j]} \hat A_k}.
  \end{align*}
  Thus, $C^i(X, \cF_{[j]}, [j]) \cong C_{t-i}(X, \cF^\perp_{[t]-[j]}, [t]-[j])$.
  Similarly, one can verify that the coboundary maps and the boundary maps agree.
  Therefore,
  \begin{equation*}
    \dim H^i(X, \cF_{[j]}, [j]) = \dim H_{t-i}(X, \cF^\perp_{[t] - [j]}, [t] - [j]).
  \end{equation*}
  Because the dimension of homology equals to the dimension of cohomology,
    we have the desired result,
    $\dim H^i(X, \cF_{[j]}, [j]) = \dim H^{t-i}(X, \cF^\perp_{[t] - [j]}, [t] - [j])$.
\end{proof}

We now prove Claim~\ref{claim:ineq}.
\begin{proof}[Proof of Claim~\ref{claim:ineq}]
  To establish the result,
    we first construct a family of maps
      $\alpha^i: C^i(X, \cF_{[j-1]}, [j-1]) \to C^i(X, \cF_{[j]}, [j])$,
    and then show that these maps form a chain map, i.e. $\delta \alpha^i = \alpha^{i+1} \delta$.
    Finally, we prove that the induced map on cohomology is injective,
    which implies the desired result
    $\dim H^i(X, \cF_{[j-1]}, [j-1]) \le \dim H^i(X, \cF_{[j]}, [j])$.

  We start by constructing the maps
    $\alpha^i: C^i(X, \cF_{[j-1]}, [j-1]) \to C^i(X, \cF_{[j]}, [j])$,
    i.e. $\bigoplus_{f \in X, |\type(f) \Delta [j-1]| = i} V_{[j-1], f} \to \bigoplus_{f' \in X, |\type(f') \Delta [j]| = i} V_{[j], f'}$,
      where $V_{[j-1], f}$ is the local coefficients induced from $\cF_{[j-1]}$.
  The map can be decomposed into maps consist of $V_{[j-1], f} \to V_{[j], f'}$
    between $|\type(f) \Delta [j-1]| = i$ and $|\type(f') \Delta [j]| = i$.
  Most of the maps are the zero maps
    except for a few pairs of $f$ and $f'$ with a certain geometric relation.

  To describe such a geometric relation,
    we establish a natural coordinate system in $\R^t$,
    where we assign each cell a point $p_f$,
    which is its center and takes value in $\{0, \frac{1}{2}, 1\}^t$.
  In particular,
    for the cell $f = [g; (a_k)_{k \in S}, (b_k)_{k \in \overline{S}}]$,
    the $k$-th coordinate of its center is
    \begin{itemize}
      \item $0$ if $k \notin S, b_k = 0$;
      \item $\frac{1}{2}$ if $k \in S$;
      \item $1$ if $k \notin S, b_k = 1$.
    \end{itemize}

  The map $\alpha: V_{[j-1], f} \to V_{[j], {f'}}$ is nonzero
    only when the faces $f, f'$ are incident to each other
    and $p_{f'} - p_f = \frac{1}{2} \cdot e_j$,
    where $e_j$ is the standard vector with value $0$ at every coordinate except the $j$-th coordinate which has value $1$.
  More explicitly,
    suppose $f = [g; (a_k)_{k \in S}, (b_k)_{k \in \overline{S}}]$.
  If $j \in S$,
    then $f' = [g \cdot a_j; (a_k)_{k \in S - \{j\}}, (b_k)_{k \in \overline{S}}, b_j = 1]$.
  Otherwise, if $j \notin S$,
    then $f' = [g; (a_k)_{k \in S}, a'_j, (b_k)_{k \in \overline{S} - \{j\}}]$
    for some $a'_j \in A_j$.
  In the first case $f \succdot f'$ and $p_{f, j} = \frac{1}{2}$, $p_{f', j} = 1$;
  In the second case $f \precdot f'$ and $p_{f, j} = 0$, $p_{f', j} = \frac{1}{2}$.


  For each pair $f, f'$ that satisfies the above condition,
    we now define the map $\alpha: V_{[j-1],f} \to V_{[j], {f'}}$
    i.e. $\F_q^{\prod_{k \in \overline{S}} \hat A_{[j-1], k}} \to \F_q^{\prod_{k \in \overline{S \Delta \{j\}}} \hat A_{[j], k}}$,
    where $\hat A_{[j], k}$ is $\hat A_k$ if $k \notin [j]$ and $A_k - \hat A_k$ if $k \in [j]$.
  Note that
    $\hat A_{[j-1], k} = \hat A_{[j], k}$ when $k \ne j$
    and $\hat A_{[j-1], j} = \hat A_j$, $\hat A_{[j], j} = A_j - \hat A_j$.
  We again split into two cases.
  If $j \in S$,
    the domain and codomain simplify into
    $\F_q^{\prod_{k \in \overline{S}} \hat A_{[j], k}} \to \F_q^{\prod_{k \in \overline{S}} \hat A_{[j], k}} \otimes \F_q^{A_j - \hat A_j}$,
    and the map is obtained by
    extending the $A_j$-coordinate from $a_j$ to $\F_q^{A_j}$ and applying $h_j^\perp: \F_q^{A_j} \to \F_q^{A_j - \hat A_j}$.
  Otherwise, if $j \notin S$,
    the domain and codomain simplify into
    $\F_q^{\prod_{k \in \overline{S}} \hat A_{[j], k}} \otimes \F_q^{\hat A_j} \to \F_q^{\prod_{k \in \overline{S}} \hat A_{[j], k}}$,
    and the map is obtained by
    applying $h_j^T: \F_q^{\hat A_i} \to \F_q^{A_i}$ and restricting the $A_j$-coordinate to $a'_j$.
  The first case is like applying the restriction map,
    except that it uses $h^\perp$ instead of $h$,
  while the second case is like applying the co-restriction map.
  This completes the construction of the maps $\alpha$.

  Next, we show that these maps form a chain map.
  Specifically, we need to verify that
  \begin{equation*}
    \delta \alpha = \alpha \delta.
  \end{equation*}
  The proof is similar to the proof for $\delta \circ \delta = 0$.
  What we need to check is
  \begin{equation*}
    \sum_{f' \in X} \delta_{f', f''} \alpha_{f, f'} = \sum_{f' \in X} \alpha_{f', f''} \delta_{f, f'}
  \end{equation*}
    for every $f, f'' \in X$.
  Since the maps are nonzero only when the faces are incident to each other,
    there are three possible relative positions between $f$ and $f''$
    \begin{itemize}
      \item $p_{f''} - p_f = \frac{1}{2} \cdot e_j \pm \frac{1}{2} \cdot e_k$ for some $k \ne j$,
      \item $p_{f''} - p_f = e_j$,
      \item $p_{f''} - p_f = 0$.
    \end{itemize}
  The first case was analyzed in the proof of $\delta \circ \delta = 0$.
  For the second case, there is only one possible $f'$ on both sides.
  So we need to check $\delta_{f', f''} \alpha_{f, f'} = \alpha_{f', f''} \delta_{f, f'}$ for that $f'$.
  Indeed, both sides are written as $\rest^{\perp}_{f', f''} \corest_{f, f'}$,
    where $\rest^{\perp}_{f', f''}$ is the restriction map using $h^\perp$.

  For the third case,
    we will show that both sides are $0$.
  Let us first focus on the left hand side $\sum_{f' \in X} \delta_{f', f''} \alpha_{f, f'}$.
  Suppose that for some $f'$, $\delta_{f', f''} \alpha_{f, f'}$ is nonzero,
    then $p_{f'} - p_f = \frac{1}{2} \cdot e_j$ by the construction of $\alpha$.
  Since $p_{f''} - p_f = 0$, this implies $p_{f''} - p_{f'} = - \frac{1}{2} \cdot e_j$.
  If we review the definition of $\delta$ for $C^*(X, \cF_{[j]}, [j])$,
    either $p_{f''} - p_{f'} = \pm \frac{1}{2} \cdot e_k$ for some $k \ne j$
    or $p_{f',j} = \frac{1}{2}$ to $p_{f'',j} = 0$ or $1$.
  Thus, the map could be nonzero only when $p_{f,j} = p_{f'',j} = 0, p_{f',j} = \frac{1}{2}$.
  In this case,
    $\alpha$ acts as a co-restriction map using $h_j$
    and $\delta$ acts as a restriction map using $h^\perp_j$.
  Overall,
    we have $\sum_{f' \in X} \delta_{f', f''} \alpha_{f, f'} = (h^\perp_j h_j^T) \otimes I_{...} = 0$,
    where $I_{...}$ is the identity operator on the other components.

  The analysis is similar for the right hand side $\sum_{f' \in X} \alpha_{f', f''} \delta_{f, f'}$.
  Suppose that for some $f'$, $\alpha_{f', f''} \delta_{f, f'}$ is nonzero,
    then $p_{f''} - p_{f'} = \frac{1}{2} \cdot e_j$ by the construction of $\alpha$.
  Since $p_{f''} - p_f = 0$, this implies $p_{f'} - p_{f} = - \frac{1}{2} \cdot e_j$.
  If we review the definition of $\delta$ for $C^*(X, \cF_{[j-1]}, [j])$,
    either $p_{f'} - p_{f} = \pm \frac{1}{2} \cdot e_k$ for some $k \ne j$
    or $p_{f',j} = \frac{1}{2}$ to $p_{f,j} = 0$ or $1$.
  Thus, the map could be nonzero only when $p_{f,j} = p_{f'',j} = 1, p_{f',j} = \frac{1}{2}$.
  In this case,
    $\delta$ acts as a co-restriction map using $h_j$
    and $\alpha$ acts as a restriction map using $h^\perp_j$.
  Overall,
    we have $\sum_{f' \in X} \alpha_{f', f''} \delta_{f, f'} = (h^\perp_j h_j^T) \otimes I_{...} = 0$.

  Finally, we show that the induced map
    $\alpha^i: H^i(X, \cF_{[j-1]}, [j-1]) \to H^i(X, \cF_{[j]}, [j])$ is injective.
  In particular, we need to show that
    for any $x \in Z^i(X, \cF_{[j-1]}, [j-1])$,
    if $\alpha(x) \in B^i(X, \cF_{[j]}, [j])$,
    then $x \in B^i(X, \cF_{[j-1]}, [j-1])$.
  The remainder of the proof relies on two claims.
  \begin{claim}\label{claim:1}
    For $0 \le i \le t$.
    If $\alpha(x) \in B^i(X, \cF_{[j]}, [j])$,
      then there exists $z \in C^{i-1}(X, \cF_{[j-1]}, [j-1])$,
      such that $\alpha(x) = \delta(\alpha(z))$.
  \end{claim}
  \begin{claim}\label{claim:2}
    For $0 \le i \le t$.
    If $w \in Z^i(X, \cF_{[j-1]}, [j-1])$ satisfies $\alpha(w) = 0$,
    then $w \in B^i(X, \cF_{[j-1]}, [j-1])$.
  \end{claim}
  (Note that $C^i(...)$ is defined to be $\{0\}$ for $i < 0$.)

  We first prove the theorem assuming the two claims.
  By Claim~\ref{claim:1},
    there exists $z$ such that
    $\alpha(x) = \delta(\alpha(z)) = \alpha(\delta(z))$,
    where we use the fact that $\alpha$ is a chain map.
  This implies, $\alpha(w) = 0$
    for $w = x - \delta(z) \in C^i(X, \cF_{[j-1]}, [j-1])$.
  Because $x \in Z^i(X, \cF_{[j-1]}, [j-1])$
    and $\delta(z) \in B^i(X, \cF_{[j-1]}, [j-1])$,
    we have $w \in Z^i(X, \cF_{[j-1]}, [j-1])$.
  By Claim~\ref{claim:2},
    this implies $w \in B^i(X, \cF_{[j-1]}, [j-1])$.
  Since $x = w + \delta(z)$,
    we have the desired result $x \in B^i(X, \cF_{[j-1]}, [j-1])$.

  We now prove the two claims in order.
  \begin{proof}[Proof of Claim~\ref{claim:1}]
    For $i = 0$, $B^i(X, \cF_{[j]}, [j]) = \{0\}$,
      which means $\alpha(x) = 0$.
    Thus, we set $z = 0$.

    For $i > 0$, since $\alpha(x) \in B^i(X, \cF_{[j]}, [j])$,
      there exists $y \in C^{i-1}(X, \cF_{[j]}, [j])$
      such that $\alpha(x) = \delta(y)$.
    We claim that $y$ can be expressed as $y' + \delta(s)$,
      where $s \in C^{i-2}(X, \cF_{[j]}, [j])$
      and $y'$ is supported only on faces $f$
      such that $p_{f, j} \ne 0$.
    (When $i = 1$, $y$ already has the desired support.
      So we can set $y' = y$.)
    This follows from the fact that part of the $\delta$ map
      that relates faces with the $j$-th coordinate equal to $\frac{1}{2}$
      to those faces with the $j$-th coordinate equal to $0$
      is surjective.
    In particular,
      for every $f$ with $p_{f, j} = 0$,
      consider the map $\bigotimes_{f' \succdot_j f} V_{[j], f'} \to V_{[j], f}$
      restricted from $\delta$.
    By the definition of $\delta$ on $C^*(X, \cF_{[j]}, [j])$,
      the map can be written as $h^\perp_j \otimes I_{...}$.
    Such a map is surjective, because $h^\perp_j$ is surjective.
    Therefore, we can find $s \in C^{i-2}(X, \cF_{[j]}, [j])$,
      that cancels out the contribution of $y$
      on faces $f$ where $p_{f, j} = 0$.

    Since $\alpha(x) = \delta(y) = \delta(y')$,
      it suffices to find $z \in C^{i-1}(X, \cF_{[j-1]}, [j-1])$
      such that $\alpha(z) = y'$.
    In our construction, $z$ is chosen to be supported on the faces
      where the $j$-th coordinate is either $0$ or $\frac{1}{2}$.
    We first describe the faces with $p_{f, j} = \frac{1}{2}$.
    To define the value of $z$ on those faces,
      we utilize the restriction of $\alpha$ to
      $\bigotimes_{f \succdot_j f'} V_{[j-1], f} \to V_{[j], f'}$
      where $p_{f', j} = 1$.
    By the definition of $\alpha$,
      this map can be written as $h^\perp_j \otimes I_{...}$,
      which is surjective.
    Therefore, we can demand $z|_{f \succdot_j f'}$ to be the preimage
      of $y'|_{f'}$.

    We now describe the faces with $p_{f, j} = 0$.
    To define the value of $z$ on those faces,
      we utilize the restriction of $\alpha$ to
      $V_{[j-1], f} \to \bigotimes_{f' \succdot_j f} V_{[j], f'}$
      where $p_{f', j} = \frac{1}{2}$.
    By the definition of $\alpha$,
      this map can be written as $h^T_j \otimes I_{...}$,
      which is not surjective.
    Nevertheless, we claim that $y'|_{f' \succdot_j f}$ is in the image of $h^T_j \otimes I_{...}$,
      which allows us to define $z|_f$ as its preimage.

    To show that $y'|_{f' \succdot_j f}$ is in the image of $h^T_j \otimes I_{...}$,
      we use the property that $\alpha(x) = \delta(y')$
      and the fact about their support.
    By construction $\alpha(x)$ and $y'$ are not supported on
      the faces $f$ where $p_{f, j} = 0$.
    That means $\alpha(x)|_f = (h^\perp_j \otimes I_{...}) (y'|_{f' \succdot_j f})$
      which is equal to $0$.
    Therefore, $y'|_{f' \succdot_j f}$ is in the kernel of $h^\perp_j \otimes I_{...}$,
      which means it is in the image of $h^T_j \otimes I_{...}$,
      as claimed.
  \end{proof}
  \begin{proof}[Proof of Claim~\ref{claim:2}]
    The goal is to construct $v \in C^{i-1}(X, \cF_{[j-1]}, [j-1])$
      such that $w = \delta(v)$.
    For $i = 0$, we set $v = 0$.
    For $i > 0$,
      $v$ is chosen to be supported on the faces where the $j$-th coordinate is $1$.
    For every $f$ with $p_{f, j} = 1$,
      consider the map $V_{[j-1], f} \to \bigotimes_{f' \succdot_j f} V_{[j-1], f'}$
      restricted from $\delta$.
    By the definition of $\delta$ on $C^*(X, \cF_{[j-1]}, [j-1])$,
      the map can be written as $h_j^T \otimes I_{...}$,
      which is not surjective.
    Nevertheless, we claim that $w|_{f' \succdot_j f}$ is in the image of $h_j^T \otimes I_{...}$,
      which allows us to define $v|_f$ as its preimage.

    To show that $w|_{f' \succdot_j f}$ is in the image of $h_j^T \otimes I_{...}$,
      we use the property that $\alpha(w) = 0$.
    The restriction of the domain of $\alpha$ to $\{f' \succdot_j f\}$ gives the map
      $\bigotimes_{f' \succdot_j f} V_{[j-1], f'} \to V_{[j], f}$.
    By the definition of $\alpha$,
      the map can be written as $h_j^\perp \otimes I_{...}$.
    Therefore, $w|_{f' \succdot_j f}$ is in the kernel of $h_j^\perp \otimes I_{...}$,
      which means it is in the image of $h_j^T \otimes I_{...}$,
      as claimed.

    We now show that $\delta(v) = w$.
    We check that the equality at faces with $j$-th coordinate equal to $0, \frac{1}{2}, 1$ in order.
    For faces with $p_{f,j} = 0$, we claim that both sides have a value of $0$.
    We have $\delta(v)|_f = 0$
      because $v$ is supported on faces where the $j$-th coordinate is $1$.
    This implies that $\delta(v)$ is supported on faces where the $j$-th coordinate is either $1$ or $\frac{1}{2}$.
    To check $w|_f = 0$, we use the fact that $\alpha(w) = 0$.
    Consider the map $V_{[j-1], f} \to \bigotimes_{f' \succdot_j f} V_{[j], f'}$
      restricted from $\alpha$.
    By the definition of $\alpha$,
      the map can be written as $h_j^T \otimes I_{...}$,
      which is injective.
    Since $\alpha(w)|_{f' \succdot_j f} = 0$,
      we obtain $w|_f = 0$.

    We now study the faces with $p_{f,j} = \frac{1}{2}$.
    For every $f'$ with $p_{f',j} = 1$,
      by the construction of $v$,
      $w|_{f \succdot_j f'} = (h_j^T \otimes I_{...})(v|_{f'})$.
    By the definition of $\delta$,
      $(h_j^T \otimes I_{...})(v|_{f'}) = \delta(v)|_{f \succdot_j f'}$.
    Thus, $\delta(v)$ and $w$ agrees on the faces with $p_{f,j} = \frac{1}{2}$.

    Finally, we study faces with $p_{f,j} = 1$.
    Let $w' = w - \delta(v)$.
    To show $w'|_f = 0$, we utilize the fact that $\delta(w') = 0$,
      and that $w'$ is supported within the faces with $p_{f,j} = 1$.
    $\delta(w') = 0$ follows from $w \in Z^i(X, \cF_{[j-1]}, [j-1])$,
      and the second statement follows from the previous parts of the proof.
    Consider the map $V_{[j-1], f} \to \bigotimes_{f' \succdot_j f} V_{[j-1], f'}$
      restricted from $\delta$.
    By the definition of $\delta$,
      the map can be written as $h_j^T \otimes I_{...}$,
      which is injective.
    Since $\delta(w')|_{f' \succdot_j f} = 0$,
      we obtain $w'|_f = 0$,
      which implies $\delta(v)$ and $w$ agrees on faces with $p_{f,j} = 1$.
    This completes the proof.
  \end{proof}
\end{proof}

\section{The local chain }
\label{sec:loc-chain}

In this section we define a ``local'' chain complex $C(L_S)$, for each subset $S\subseteq\{1,\ldots,t\}$. This chain complex will be used as a tool in the analysis of the complex $C_*(X,\mF)$ introduced in the previous section. In particular, as we will show in Lemma~\ref{lem:loc-glob}, the chain complex $C(L_S)$ is isomorphic to the ``local view'' of $C_*(X,\mF)$ from any face $f$ such that $\type(f)=\ol{S}$, i.e.\ the subcomplex $C_*(X_{\geq f};\mF)$. In Section~\ref{sec:robustness} we introduce an important property of the local chain complex, \emph{robustness}.

\subsection{The complex}
\label{sec:loc-complex}

We start by defining a $1$-dimensional complex $L_i=L_i(0)\cup L_i(1)$, where $i\in\{1,\ldots,t\}$, as follows. There is a single $0$-dimensional face $L_i(0)=\{\emptyset\}$ and there are $n_i$ $1$-dimensional faces which we label using elements of $A_i$, $L_i(1)=A_i$. The incidence structure is the obvious one, $\emptyset \precdot a$ for any $a\in A_i$.

Next we introduce coefficient spaces $V_\emptyset = \F_q^{\hat{A}_i}$ and $V_a=\F_q$ for each $a\in A_i$.
 The coboundary map is defined, for $x\in \F_q^{\hat{A}_i}$ and (with a slight abuse of notation, which we will repeat) $a_i$ the $a_i$-th canonical basis vector of $\F_q^{A_i}$, as
\[ \delta_{\{i\}}(x)(a_i) \,=\, (h_i^T(x))(a_i)\;.\]

The boundary map $\partial_{\{i\}}=\delta_{\{i\}}^T$ is then
\[ \partial_{\{i\}}(y)(\emptyset) \,=\, h_i\, y\;,\]
where $y \in \oplus_a V_a \simeq \F_q^{A_i}$.

Above we used a subscript $\{i\}$ to distinguish them from the (co)boundary maps associated with the complex $C(X)$. We denote the resulting complex as $C(L_{\{i\}})$.

More generally, for $S\subseteq \{1,\ldots,t\}$ we define an $|S|$-dimensional complex $C(L_S)$ as follows. For $k\in\{0,\ldots,|S|\}$ the $k$-faces are $L_S(k)=\cup_{T\subseteq S, |T|=k} L_S(T)$, where the faces $L_S(T)$ \emph{of type $T$} are in bijection with elements ${\prod_{i\in T} A_i}$.
 The incidence structure is given by $f'\succdot f$ iff $f'$ is of type $T'$ such that $|T'|=|T|+1$ and $T\subseteq T'$. The local coefficient space associated to a face $f$ of type $T$ is $V_f = \F_q^{\prod_{i\in S- T} \hat{A}_i}$. For $x\in V_f$, the coboundary map is defined as
\[ \delta_S(x) \,=\, \sum_{j\in S- T} (I_{-j} \otimes h_j^T)(x) \,\in \bigoplus_{f'\succdot f} V_{f'}\;,\]
and the boundary map is
\[ \partial_S(x) \,=\, \sum_{j\in T} (I_{-j} \otimes h_j)(x) \,\in \bigoplus_{f'\precdot f} V_{f'}\;.\]

\begin{remark}
One can verify that the complex $C(L_{S})$ is the homological product~\cite{bravyi2014homological} of the $1$-dimensional complexes $C(L_{\{i\}})$, for $i\in S$. In particular, we see that $C_0(L_S) \simeq \otimes_{i\in S} C_0(L_{\{i\}})$, $C_1(L_S) \simeq \oplus_{i\in S} (C_1(L_{\{i\}})\otimes_{j\in S-\{i\}} C_0(L_{\{j\}}))$, etc.
\end{remark}

\begin{lemma}[Local and global]\label{lem:loc-glob}
  For any face $f=[g;a,b]$ of type $S$ there is a natural isomorphism between $C(X_{\geq f})$ and $C(L_{\ol{S}})$. In particular, for every $T\subseteq \ol{S}$, $C(X_{\geq f}(S\cup T))\simeq C(L_{\ol{S}}(T))$.
\end{lemma}

\begin{proof}
A face $g\in X_{\geq f}(S\cup T)$ is given by $a_T\in A_T$ and can be written as $g = (g_i)$ where
$g_i=\begin{cases}
  f_i& i\not\in T\\
  a_i& i\in T
\end{cases}$. In shorthand, $g=(f_{-T}||a_T)$. By definition, for a face $g$ of type $S\cup T$, $V_g = \F_q^{\prod _{j\notin (S\cup T)} \hat{A}_j}=V_{S\cup T}$ so
  \[C(X_{\geq f}(S\cup T)) = \bigoplus_{g\in X_{\geq f}(S\cup T)}V_g = \bigoplus_{a\in A_T} V_{(f_{-T}||a_T)}
  = \bigoplus_{a\in A_T} V_{S\cup T}\;.
  \]
    An element of $C(X_{\geq f}(S\cup T))$ is thus a tuple $x=(x_{\bar a})_{\bar a\in A_T}$ where $x_{\bar a} \in \F_q^{\prod_{j\notin (S\cup T)}\hat{A}_j}$. Such a tuple is naturally identified with $x'\in \F_q^{\prod_{j\in T} A_j \prod_{j\in \ol{S}\setminus T} \hat{A}_j} = C(L_{\ol{S}}(T))$ by setting $x'_{\bar a,\bar a'}=(x_{\bar a})_{\bar a'}$ for $\bar a\in\prod_{j\in T} A_j $ and $\bar a'\in \prod_{j\in \ol{S}\setminus T} \hat{A}_j$.

    It is easy to check that under this isomorphism the boundary and coboundary maps of $C(L_{\ol{S}})$ and of $C(X_{\geq f})$ coincide.
\end{proof}

We will use the following.

\begin{lemma}\label{lem:exactness}
For any $S\subseteq \{1,\ldots,t\}$ the chain complex $C(L_S)$ is exact at $i=0,\ldots,|S|-1$, i.e.\ for any $x\in C_i(L_S)$ such that $\partial_S(x)=0$ there is a $y\in C_{i+1}(L_S)$ such that $\partial_S(y)=x$.
\end{lemma}

\begin{proof}
This follows from the K\"unneth formula, using that $C(L_S)$ is the homological product of $1$-dimensional complexes $C(L_{\{i\}})$ which satisfy $\dim H_0(L_{\{i\}}) = 0$.
\end{proof}

\begin{lemma}\label{lem:tensor-code}
For any $S\subseteq \{1,\ldots,t\}$ and for any $x\in C_{|S|}(L_S)$ such that $\partial_S(x)=0$, $x$ must be in the tensor code $\bigotimes_{j\in S} \ker(h_j)$.
\end{lemma}
\begin{proof}
Observe that $C_{|S|}(L_S)\cong \F_q^{\prod_{j\in S}A_j}$, so an element $x\in C_{|S|}(L_S)$ is an $|S|$-dimensional tensor. The condition $\partial_S x=0$ for $x\in C_{|S|}(L_S)$ implies that for every $j\in S$ and every $(a_j)_{j\in S}\in \prod_{j \in S} A_j$, the column vector $x(a_{-j},\cdot)$ where we fix all coordinates except the $j$-th belongs to $\ker(h_j)$. This is exactly the definition of the space $\bigotimes_{j\in S} \ker(h_j)$.
\end{proof}

\subsection{Robustness}
\label{sec:robustness}

Let $S\subseteq\{1,\ldots,t\}$ and $C(L_S)$ be the dimension-$|S|$ complex described in the previous section. For $i\in S$ let
\[ \code_i^\perp = \im(h_i^T) \subseteq \F_q^{A_i}\;.\]
Let $d_i^\perp$ be the minimum distance of $\code_i^\perp$, i.e.\ $d_i^\perp = \min\{|x|_H: x\in \code_i^\perp\}$, where recall that $|\cdot|_H$ denotes the Hamming weight. For a $k$-face $f\in L_S(k)$ of type $T$ and $x\in V_f = \F_q^{\prod_{i\in S\setminus T} \hat{A}_i}$ we let $|x| = 1_{x\neq 0}$. For $x\in C^k(L_S)$ we let $|x| = \sum_{f\in L_S(k)} |x(f)|$, i.e.\ the number of faces $f$ such that $x(f)$ is nonzero. We refer to $|\cdot|$ as the \emph{block-(Hamming)-weight}.

\begin{definition}\label{def:min}
Let $k\in\{1,\ldots,|S|\}$. An element $x\in C^k(L_S)$ is called \emph{minimal} if for any $y\in C^{k-1}(L_S)$, $|x+\delta_S(y)|\geq |x|$.
\end{definition}

\begin{definition}[Robustness]\label{def:robust}
Let $S\subseteq\{1,\ldots,t\}$ and $\ell=|S|$. For $0\leq k \leq \ell-1$ let $\kappa_{\ell,k}$ be a positive real. We say that $C(L_S)$ is \emph{$\kappa_{\ell,k}$-robust} (implicitly, \emph{at level $k$}) if for any $x\in C^{k}(L_S)$ such that $x$ is minimal it holds that
\[ |\delta_S(x)| \,\geq\, \kappa_{\ell,k}  \,n\, |x|\;.\]
\end{definition}

\begin{remark}
The condition of $\kappa$-robustness is stronger  than the notion of product-expansion from~\cite{kalachev2022two} because it applies to all levels $k< |S|$ as opposed to only $k = |S|-1$. Nevertheless, in Lemma \ref{lemma:prodexp-robust} below we show a reduction from the later to the former. Furthermore, both notions are equivalent to \emph{co-boundary} expansion of the complex $C(L_S)$. For $k\in\{1,\ldots,|S|-1\}$ the co-boundary expansion coefficient $h^k(L_S)$ is defined as
\[ h^k(L_S)\,=\, \min_{x\in  C^k(L_S) - \im(\delta_S)} \frac{ | \delta_S(x)|}{\min_{y\in \im(\delta_S)} |x+y|}\;.\]
We have that $h^{k}(L_S)\leq \kappa_{\ell,k} n$ because if $x$ is minimal then $\min_{y\in \im(\delta_S)} |x+y|=|x|$. Conversely, if $x$ achieves the minimum in the definition of $h^k(L_S)$ then $x'=x+y$ for any $y\in \im(\delta_S)$ that minimizes $|x+y|$ will be such that $|\delta_S(x')| = h^k(L_S) |x'|$.
\end{remark}

\begin{remark}\label{rk:robust-distance}
For the case $|S|=1$, the condition of $x\in C^{0}(L_S)$ being locally minimal is vacuous, because there are no $(-1)$-faces. Moreover, in that case $S=\{i\}$ for some $i$, $\delta_S(x)=h_i^T(x)$, $|x|=1_{x\neq 0}$ and $|\delta_S(x)|=|\delta_S(x)|_H$, the Hamming weight. Thus in this case the condition of robustness is equivalent to the condition of distance, i.e.\ for each $i\in\{1,\ldots,t\}$ we have that $L_{\{i\}}$ is $(d_i^\perp / n)$-robust, where $d_i^\perp$ is the distance of the dual code $\code_i^\perp$.
\end{remark}

\begin{definition}\label{def:h-robust}
Let $\kappa>0$. We say that the family $\{h_i\}$ is \emph{two-way $\kappa$-robust} if the following conditions hold.
\begin{itemize}
\item When based on the matrices $\{h_1,\ldots,h_t\}$, the complex $C(L_S)$ is $\kappa$-robust for each $S\subseteq\{1,\ldots,t\}$ and at each level $k=0,\ldots,|S|-1$.
\item When based on the matrices $\{\hp_1,\ldots,\hp_t\}$, the complex $C(L_S)$ is $\kappa$-robust for each $S\subseteq\{1,\ldots,t\}$ and at each level $k=0,\ldots,|S|-1$.
\end{itemize}
\end{definition}

We show the following:

\begin{theorem}\label{thm:robust}
    For each collection of intervals $I_1,\ldots,I_\ell\subseteq(0,1)$, there exists $\kappa>0$ such that for all $n\in \N$ there exist parity check matrices $h_i \in \F_q^{(1-\rho_i)n\times n}$, where $q = 2^{(n+3)^\ell}$, such that $\rho_i\in I_i$ and $\{h_1,\ldots,h_t\}$ is two-way $\kappa$-robust.
\end{theorem}

We note that the theorem was previously known for $\ell=1$ and $\ell=2$ over fixed $q$ (for example $\F_2$). The case for $\ell=1$ is known as the Gilbert–Varshamov bound and the case for $\ell=2$ was shown recently in \cite{kalachev2022two,DHLV} (and earlier for a sub-constant robustness parameter in \cite{PK2}).
Our proof crucially relies on a recent result by Panteleev and Kalachev on the existence of code tuples with product expansion, provided that the field size $q$ is large enough.

\begin{theorem}[Panteleev and Kalachev \cite{PK-future}]\label{thm:PK-future}
    For each collection of intervals $I_1,\ldots,I_\ell\subseteq(0,1)$, there exists $\rho>0$ such that for all $n\in \N$ there exist codes $\C_1\ldots,C_\ell\subset\F_q^n$, where $q = 2^{(n+3)^\ell}$, such that $\frac 1 n \dim(\C_i)\in I_i$ and both collections $\C_1,\ldots,\C_\ell$ and $\C_1^\perp,\ldots, C^\perp_\ell$ are $\rho$-product expanding.
\end{theorem}

In Lemma \ref{lemma:prodexp-robust} in the following section we show that any $\rho$-product expanding family is also $\rho'$-robust, for positive $\rho'$. This proves Theorem~\ref{thm:robust}.

\subsection{Product expansion implies robustness}\label{sec:product-robust}

In this section we show that product expansion  implies robustness. Recall the definition.
\begin{definition}[Product expansion, \cite{kalachev2022two}]\label{def:prexp}
    A family of codes $\set{\C_i}_{i\in [\ell]}$ is said to have $\rho$-product expansion iff the co-cycle expansion of $L_{[\ell]}$ at level $\ell-1$ is at least $\rho$, namely, if $\eps_{\cocyc}(\ell-1)\geq n\cdot \rho$.
\end{definition}
To be precise, the definition in \cite{kalachev2022two} is given in slightly different terms, but it is exactly equivalent to the definition above, as shown in \cite[Appendix B]{kalachev2022two}.

For our results, we need a bit more: as stated in Definition~\ref{def:h-robust}, we need not only coboundary expansion at level $\ell-1$, but also for {\em all} levels $k\leq \ell-1$. In the next lemma we show that the latter follows from the former. For convenience, assume for the remainder of this section that the family $\set{\C_i}_i$ is fixed, and we omit it from our notation.

\begin{lemma}\label{lemma:prodexp-robust}
  Let $\ell\geq 1$ and $S$ such that $|S|=\ell$.
For every $\rho>0$ there is some $\rho'>0$ (depending on $\rho$ and $\ell$) such that the following holds. Suppose that for every $S'\subsetneq S$ and $\ell'=|S'|\geq 1$ the complex $L_{S'}$ satisfies
$\kappa_{\ell',\ell'-1} \geq \rho$. Then for all $0\leq k\leq \ell-1$, $\kappa_{\ell,k} > \rho'$.
\end{lemma}

The proof of the lemma is given in the next two subsections.

\subsubsection[Induction]{Proof of Lemma \ref{lemma:prodexp-robust}}

We prove Lemma~\ref{lemma:prodexp-robust} by induction on $\ell$. For $\ell=1$ it is true because the only relevant case, $k=0$, is covered by product expansion.  We assume for convenience that $S= \set{1,\ldots,\ell} = [\ell]$.
Assume that we have a lower bound on $\kappa_{\ell,k}$ for all $0\leq k<\ell$. We show a lower bound on $\kappa_{\ell+1,k}$, for any $k<\ell$ (the case $k=\ell$ is covered by product expansion).

Recall the sheaf  $\sheaf_\ell$, which is such that for any $s\in L_{[\ell]}$, $\sheaf_\ell(s) = \F_q^{\prod_{i\in [\ell]\setminus \type(s)} \hat A_i}$. Let $C(L_{[\ell+1]},\sheaf_{\ell+1})$ denote the chain complex. We partition $L_{[\ell+1]}$ into two parts, $L_0$ and $L_1$, as follows.
\begin{enumerate}
    \item $L_0$ consists of all faces whose type is contained in $[\ell]=\{1, \ldots, \ell\}$. Clearly $L_0 \cong L_{[\ell]}$. Moreover, for each face $s\in L_0$,
    \[ (\sheaf_\ell(s))^{\hat A_{\ell+1}}= (\F_q^{\prod_{i\in [\ell]\setminus\type(s)} \hat A_i})^{\hat A_{\ell+1}} \cong\F_q^{\prod_{i\in [\ell+1]\setminus \type(s)}\hat A_i}=\sheaf_{\ell+1}(s).\]

    So
    \begin{equation}\label{eq:ind-equ-1}
       C^k(L_{0},\sheaf_{\ell+1})\,\cong\, C^k(L_{[\ell]},\sheaf_\ell^M)\;,
    \end{equation}
    for $M=|\hat A_{\ell+1}|$. Here, $\sheaf_\ell^M$ is the sheaf where each space $\sheaf_\ell(s)$ is replaced by the cartesian product $\sheaf_\ell(s)^M$.
    We will rely on Proposition \ref{prop:DPcbdry} below that upgrades coboundary expansion with respect to coefficients from $\sheaf$ to the same with respect to $\sheaf^M$ for any $M\in\N$.
    \item $L_1$ consists of faces whose type contains $\ell+1$. These faces can be written as $(s;j)$ for a face $s\in L_{[\ell]}$ and an element $j\in A_{\ell+1}$. Thus $L_{1}(k)$ is partitioned into $n=|A_{\ell+1}|$ copies of $L_{[\ell]}(k)$. Put differently, for each $s\in L_{[\ell]}$ we have an $n$-tuple $(\sheaf_{\ell+1}((s;j)))_{j\in A_{\ell+1}}$ of coefficients that naturally lives in $\sheaf_\ell(s)^n$. So
    \begin{equation*}
        C^k(L_1,\sheaf_{\ell+1}) \cong C^{k}(L_{[\ell]},\sheaf_\ell^n).
    \end{equation*}
    
    Here we are using the fact that there is a face $(s;j)\in L_1$ for each $j\in A_{\ell+1}$, and the coefficient space of this face satisfies
    \[ \sheaf_{\ell+1}((s;j)) = \F_q^{\prod_{i\in [\ell+1]\setminus \type((s;j))}\hat A_i} = \F_q^{\prod_{i\in [\ell]\setminus \type(s)}\hat A_i} = \sheaf_{\ell}(s).\]
    For  $x\in C^k(L_1,\sheaf_{\ell+1})$, we can decompose $x=\sum_{s\in L_1} x(s)$ according to the value $j\in A_{\ell+1}$ that $s$ takes on the $(\ell+1)$-st coordinate, as
    \begin{equation}\label{eq:directsum}
        x = \sum_{j\in A_{\ell+1}} y_j \circ j\;,\qquad y_j\in C^{k}(L_{[\ell]},\sheaf_\ell)
    \end{equation}
where $y_j \circ j\in C^k(L_{[\ell+1]},\sheaf_{\ell+1})$ is defined by
    \[\forall s\in L_{[\ell]}(k)\;,\quad \forall i\in A_{\ell+1}\;,\qquad (y_j\circ j )(s;i) =\begin{cases}
  y_j(s) & i=j\\
  0& i\neq j
\end{cases}.\]
\end{enumerate}

Let $x = \delta y \in C^{k+1}(L_{[\ell+1]},\sheaf_{\ell+1})$ be such that $|x|=\epsilon |X(k+1)|$. Write $x_0 = x|_{L_0}\in C^k(L_0,\sheaf_{\ell+1})$ and observe that $x_0= (\delta y) |_{L_0} = \dt ( y|_{L_0} )$, where $\dt$ denotes the coboundary map in $C(L_{[\ell]}, \sheaf_{\ell+1})\cong C(L_{[\ell]}, \sheaf_\ell^{M})$. (This is because $\delta_\ell$ only adds coordinates in one of the first $\ell$ directions.)
%
%
To complete the argument we need a bound on $\bar\kappa_{\ell,k}$, the coboundary expansion constant of the chain complex $C^k(L_{[\ell]},\sheaf_\ell^M)$. Let $\bar\kappa_{\ell ,k}$ be the infimum of all coboundary expansion constants of $C^k(L_{[\ell]},\sheaf_\ell^M)$, taken over all natural integers $M$.
We call $\bar\kappa_{\ell ,k}$ the \emph{collective robustness}
  because of its similarity to the notion of \emph{collective cosystolic expansion} studied in \cite[Definition 1.2]{kaufman2021new}.

In Section~\ref{sec:proof-sheaf} below we show the following.

\begin{proposition}\label{prop:DPcbdry}
    Let $\kappa = \min_{i<\ell'\leq \ell} \kappa_{\ell',i}$. There is a constant $A>0$ depending only on $\ell$ such that $ \bar \kappa_{\ell,k} \geq A\cdot \kappa^{2\ell}$.
\end{proposition}

In words, this proposition roughly says that if the chain complex $C^*(L_{[\ell]}, \sheaf_\ell)$ has constant robustness $\kappa_{\cdot,\cdot}$, then the chain complex has constant collective robustness $\bar\kappa_{\cdot,\cdot}$.

By the inductive hypothesis we have a lower bound for $\kappa_{\ell',i}$ for all $i<\ell'\leq \ell$, so by the proposition, there is $\tilde{y}_0\in C^{k}(L_{[\ell]}, \sheaf_\ell^{M})$ such that
$\dt\tilde{y}_0 = x_0$ and
\begin{equation}\label{eq:sheaf-3}
    \bar\kappa_{\ell,k} n\cdot |\tilde y_0  |  \leq  |x_0 | \leq  |x| .
\end{equation}
Let $x' = x-\delta\tilde y_0\in C^k(L_1,\sheaf_{\ell+1})$. By construction $x'\in \im \delta$ since $x = \delta ( y-\tilde y_0)$, and also, using $n=|A_{\ell+1}|$,
\begin{equation}\label{eq:sheaf-4}
|x'| \leq |x| + |\delta \tilde y_0| \leq |x| + n\ell\cdot |\tilde y_0| \leq |x| + \frac {n\ell}{\bar\kappa_{\ell,k}n}|x| \leq (1+ \frac {\ell}{\bar\kappa_{\ell,k}})|x| \;.
\end{equation}
\begin{claim}
There are chains $z_j \in C^{{k-1}}(L_0,\sheaf_{\ell})$ for each $j\in A_{\ell+1}$, such that
\begin{equation}\label{eq:sheaf-xprime}
  x' = \sum_{j\in A_{\ell+1}} (\dt z_j) \circ j\;.
\end{equation}
\end{claim}
\begin{proof}
We have $x' = \delta y'$ for $y' = y-\tilde y_0$. Since $x'|_{L_0}=0$, we deduce that $y'|_{L_0} \in \ker \dt$. By exactness of $\dt$ there is some $z'$ such that $y'|_{L_0} = \dt z'$. Letting $y'' = y'-\delta z'$, we still have $\delta y'' = x'$ but now $y''|_{L_0} = y'|_{L_0} - \delta z'|_{L_0} = \dt z' - \delta z'|_{L_0}=0$. In conclusion, $y''\in C^{k}(L_1)$, so by \eqref{eq:directsum} it can be written as $y'' = \sum_{j\in A_{\ell+1}} z_j\circ j$, and then we can write $x' = \delta y'' = \sum_{j\in A_{\ell+1}} \dt z_j\circ j$.
\end{proof}
Let $y'' = \sum_{j\in A_{\ell+1}} \tilde z_j\circ j$, where
$\tilde z_j$ are chosen to have smallest weight while satisfying $\dt \tilde z_j = \dt z_j$. By induction $\kappa_{\ell,{k-1}} n \cdot|\tilde z_j |\leq |\dt z_j|$, and
\begin{align}
  |y''| = \sum_j |\tilde z_j| &\leq \sum_j \frac{|\dt z_j| }{ n\kappa_{\ell,k-1}} \nonumber \\
  &= \frac{|x'|}{\kappa_{\ell,k-1} n} \nonumber \\
  &  \leq \frac 1 {\kappa_{\ell,k-1} n}(1+ \frac {\ell}{\bar \kappa_{\ell,k}}) |x|\;, \label{eq:sheaf-5}
\end{align}
where the second line is by~\eqref{eq:sheaf-xprime} and the third is by~\eqref{eq:sheaf-4}.
This allow us to define $\tilde{y} = \tilde{y}_0 + y''$, which satisfies $\delta \tilde{y} = x$ and
\begin{equation*}
  |\tilde{y}| \le |\tilde{y}_0| + |y''| \le \left(\frac{1}{\bar \kappa_{\ell,k} n} + \frac 1 {\kappa_{\ell,k-1} n}(1+ \frac {\ell}{\bar \kappa_{\ell,k}})\right) |x|
\end{equation*}
which follows from~\eqref{eq:sheaf-3} and~\eqref{eq:sheaf-5}.
Thus,
\begin{equation*}
  \kappa_{\ell+1,k} \geq  \frac{1}{\frac{1}{\bar \kappa_{\ell,k}} + \frac 1 {\kappa_{\ell,k-1}}(1+ \frac {\ell}{\bar \kappa_{\ell,k}})}
\end{equation*}
which concludes the induction step.
\qed

\subsubsection[From $\sheaf$ to $\sheaf^M$]{Proof of Proposition \ref{prop:DPcbdry}}\label{sec:proof-sheaf}

In this subsection we prove Proposition
\ref{prop:DPcbdry}, which lower bounds the coboundary expansion of a direct product sheaf based on the coboundary expansion of the original sheaf. To lighten the notation, we fix $\ell$ and let $X=L_{\ell}$, and we fix $M\in \mathbb{N}$ and denote by $\bar\sheaf$ the sheaf that to each face $f$ of $L_{\ell}$ associates the coefficient space $\bar\sheaf(f) = (\sheaf(f))^M$. The coboundary operators in this complex simply operate component-wise.
Moreover, the weight of a chain is {\em still} the number of faces with non-zero values. Let $\bar\kappa_{\ell,k}$ be the robustness of $C(X,\bar\sheaf)$ at level $k$, as in Definition~\ref{def:robust}.

\begin{definition}\label{def:heavy}
  Let $0<p\leq 1$ and $0\leq \r\leq i$.
Given an $i$-chain $z$ we say that an $\r$-face $e$ is {\em $p$-heavy} for $z$ if
\[|z(X_{\geq e})| \,\geq\, p\cdot |X_{\geq e}(i)|\;.\]
If $e$ is not $p$-heavy then we say that it is {\em $p$-light}.
\end{definition}

\begin{proof}[Proof of Proposition \ref{prop:DPcbdry}]
    Let $\y_1 \in C^{k}(L_{[\ell]},\bar\sheaf)$ and let $x = \delta \y_1$. We will find $ \y\in C^{k}(L_{[\ell]},\bar\sheaf)$ such that $\delta\y = x$ and $|\y| \leq \frac \epsilon { \bar\kappa_{\ell,k}}\cdot |X(k)|$, where $|x| = \epsilon\cdot |X(k+1)|$. We may assume wlog that $\epsilon \leq \frac 1 3 (\frac \kappa 4)^{2\ell}$. This is because if not, then as long as $A \leq \frac 1{3(16)^\ell}$ the proposition follows trivially as
    \[\bar \kappa_{\ell,k} n\cdot |\y_1| \leq \bar \kappa_{\ell,k} |X(k+1)|\leq \epsilon|X(k+1)|=|x| \;.\]
    For each $\alpha \in \F_q^M$, we denote $\alpha\cdot \y_1$ the chain in $C^{k}(X,\sheaf)$ that satisfies $(\alpha\cdot \y_1)(s) = \alpha\cdot( \y_1(s))$ for every $s\in X(k)$. By linearity, $\delta(\alpha \cdot \y_1) = \alpha \cdot (\delta \y_1) = \alpha \cdot x$. Observe that for any face $s$, $x(s)=0$ implies $\alpha\cdot x(s)=0$ so $|\alpha\cdot x| \leq |x|$. Let $\hat y_\alpha \in C^{k}(L_{[\ell]},\sheaf)$ be the minimum-weight chain such that $\delta \hat y_\alpha = \alpha \cdot x$. By coboundary expansion of $C(L_{[\ell]},\sheaf)$,
    \begin{equation}\label{eq:yalpha}
     \forall \alpha,\quad \kappa_{\ell,k} n\cdot |\hat y_\alpha | \leq |\delta \hat y_\alpha| = |\alpha\cdot x| \leq |x|.
        \end{equation}
    Let $\e_1,\ldots,\e_M \in \F_q^M$ be the standard basis and define
    \begin{equation}\label{eq:defy}\forall s\in X(k),\quad
        \y(s) = (\hat y_{\e_1}(s),\ldots,\hat y_{\e_M}(s)).
    \end{equation}
    Clearly $\delta \y=(\e_1\cdot x,\ldots,\e_M\cdot x)=x$.
        To conclude the proof we will show the following claim.

        \begin{claim}\label{claim:H-set}
          Let $\rho = \frac 1 3 (\frac {\kappa}4)^{2\ell}$. Let $H\subseteq X(k)$ be the set of $k$-faces that are above some $\rho$-heavy face of $x$:
          \begin{equation}\label{eq:def-heavy}
              H \,=\, \sett{s\in X(k)}{\exists  e \prec s, \; e \hbox{ is $\rho$-heavy for }x}\;.
          \end{equation}
    Then
    \begin{equation}\label{eq:linear}
        \forall s\not\in H,\forall \alpha,\beta\in \F_q^M,\qquad \hat y_\alpha (s) + \hat y_\beta (s) =\hat y_{\alpha+\beta} (s)\;.
    \end{equation}
    Moreover,
    \begin{equation}\label{eq:H-size}
    |H| \,\leq\,   2^k \ \frac\epsilon{\rho}\cdot |X(k)|\;.
    \end{equation}
  \end{claim}

  We first show how the claim can be used to conclude the proof.
    For any face $s\not \in H$ and for any $\alpha\in\F_q^M$, \eqref{eq:linear} implies that $\hat y_{\alpha}(s) = \sum_i \alpha_i \hat y_{\e_i}(s)$. If $\y(s)\neq 0$ then there is some $i$ such that $\hat y_{\e_i}(s)\neq 0$ which means that $\Pr_{\alpha}[\hat y_\alpha(s) \neq 0] \geq 1/2$, where the probability is taken over a uniformly random choice of $\alpha$. We conclude that for every $s\in \supp(\y)\setminus H$, at least half of the $\alpha$ have $\hat y_\alpha(s)\neq 0$. Reversing the order of the quantifiers, there is some $\alpha$ such that $\hat y_\alpha(s)\neq 0$ for at least half of $s\in \supp(\y)\setminus H$. Namely $\frac 1 2|\supp(\y)\setminus H| \leq |\hat y_\alpha|$, or
    \begin{equation}\label{eq:bary}
        |\y| \leq 2 |\hat y_\alpha| + |H|.
    \end{equation}
    Using the bound~\eqref{eq:H-size} on the size of $H$,
    \begin{align*}
      n|\y| \leq 2n|\hat y_\alpha| + n|H| &\leq 2(\kappa_{\ell,k})^{-1} |x| + 2^k  n \frac\epsilon{\rho} |X(k)|\\
      & \leq  \Big(\frac{2}{\kappa_{\ell,k}} + \frac{2^k t}{\rho}\Big) |x|\;,
  \end{align*}
where the last line uses $(n/t) |X(k)|\leq |X(k+1)|$
and the definition of $\eps$. Since by definition $\kappa_{\ell,k}\geq \kappa$, this shows that $\bar\kappa_{\ell,k}=\Omega_k(1)\kappa^{2\ell}$, as desired.

It remains to prove Claim~\ref{claim:H-set}. We start by showing the upper bound~\eqref{eq:H-size} on the size of $H$. Towards this, consider the random process of choosing a face in $X(k+1)$ by first choosing uniformly at random an $\r$-face $e$ and then a random $(k+1)$-face containing it, $s\succeq e$. If $e$ is $\rho$-heavy for $x$, then the probability that $x(s)\neq 0$ is at least $\rho$. The fraction of $\r$-faces that are heavy for $x$ is thus at most $\epsilon/\rho$.
    A $k$-face has exactly $\binom{k}{\r}$ $\r$-faces below it, so it has this many chances to contain a heavy $\r$-face. By a union bound, the probability of this event is at most $\binom{k}{\r}\cdot\eps/\rho$. Summing over all $\r\leq k$,
    \[
    |H| \leq \sum_{\r=1}^{k} \binom{k}\r \frac\epsilon{\rho}\cdot |X(k)| = 2^k\ \frac{\epsilon}{\rho}\cdot |X(k)|\;.
    \]
    It remains to prove \eqref{eq:linear}.
    First observe that for all $\alpha,\beta$,
    \[\delta \hat y_\alpha  + \delta \hat y_\beta  -\delta \hat y_{\alpha+\beta}  = \alpha\cdot x + \beta \cdot x - (\alpha+\beta)\cdot x = 0.\]
    So
    $ \hat y_\alpha +  \hat y_\beta  - \hat y_{\alpha+\beta}\in \ker \delta_{k} = \im\delta_{k-1}$ and therefore we can write
    \begin{equation}\label{eq:abc}
            \hat y_\alpha  +  \hat y_\beta - \hat y_{\alpha+\beta} = \delta z
    \end{equation}
 for some $z\in C^{{k-1}}(X,\sheaf)$. Let $z$ be such a chain with minimal weight.
We assume that $\delta z(s_0)\neq 0$ for some $s_0\in X(k)$, and show that $s_0\in H$. This will show~\eqref{eq:linear}.

We begin with the following claim that will help find a face $e$ below $s_0$ that is heavy for $\delta z$.

\begin{claim}\label{claim:heavy}
Let $1\leq \r\leq i<\ell$ and let $z\in C^i(X)$. Define $f_{i,r},g_{i,r}:[0,1]\to[0,1]$ by
\[f_{i,r}(p)\,=\, \frac{ \kappa_{\ell-r,i-\r}}{2} \cdot
\frac{i-\r+1}{\r(\ell-\r+1)}\cdot p\;,\quad\text{and}\quad g_{i,r}(p) \,=\,\frac{ \kappa_{\ell-r,i-\r}}{2} \cdot \frac{i-\r +1}{\ell-i} \cdot p \;.\]
Then for any $0<p\leq 1$, if $e\in X(\r)$ is $p$-heavy for $z$ yet each $e'\precdot e$ is $f_{i,r}(p)$-light for $z$, then $e$ is $g_{i,r}(p)$-heavy for $\delta z$.
\end{claim}

We defer the proof of the claim to the end of this section. Since $\delta z(s_0)\neq 0$, there must be some $e_0\precdot s_0$ such that $z(e_0)\neq 0$. Clearly $e_0\in X(k-1)$ itself is $p_0=1$-heavy for $z \in C^{k-1}(X)$.
If all $e\precdot e_0$ are $p_1=f_{k-1,k-2}(p_0)$-light for $z$, we let $s_1=e_0$. By Claim~\ref{claim:heavy}, $s_1$ is $q_1=g_{k-1,k-2}(1)$-heavy for $\delta z$.
Otherwise, choose $e_1\in X(k-2)$, $e_1\precdot e_0$ a $p_1$-heavy face for $z$. If all $e\precdot e_1$ are $p_2=f_{k-1,k-3}(p_1)$-light, we let $s_1=e_1$. By Claim~\ref{claim:heavy}, $s_1$ is $q_2=g_{k-1,k-3}(p_1)$-heavy for $\delta z$.
Otherwise we continue taking $e_j\precdot \cdots e_1\precdot e_0$ such that
$e_j\in X(k-1-j)$ is $p_{j}$-heavy for z, where
\begin{align*}
  p_{j} &= f_{k-1,k-j-1} \circ\cdots\circ f_{k-1,k-2}(1)\\
  &= \kappa_{\ell-k+2,1}\cdots\kappa_{\ell-k+j+1,j}\cdot\frac{1}{2^j} \cdot\frac{2\cdots (j+1)}{\ell \cdots j(\ell-j+1)}\\
  &=\frac{1}{2^j}\cdot \kappa_{\ell-k+2,1}\cdots\kappa_{\ell-k+j+1,j} \cdot \frac{j+1}{(\ell-j+1)\cdots \ell}\;.
\end{align*}
 We can continue this process until either $j=k$ or for some $j< k$ all $e\precdot e_j$ are $p_{j+1}$-light for $z$. In the latter case Claim~\ref{claim:heavy} applies and we find $s_1$ that is $q_{j+1}=g_{k-1,k-j-2}(p_{j})$-heavy for $\delta z$. Note that if $j=k$ then there is a single face $e_{j+1}\in X(0)$ such that $e_{j+1}\preceq e_0$, and being $p_{k-1}$-light is the same as the fractional weight of $x$ being smaller than $p_{k-1}$, which holds by assumption (given the assumption made on $\eps$ at the start of the proof).

 To summarize the argument, letting $q$ be the minimum of all $q_j$, for $j\in\{k-1,\ldots,1\}$, we have found $s_1 \preceq e_0$ such that $s_1$ is at least $q$-heavy for $\delta(z)$. Moreover, $q=\Omega(\kappa^\ell)$, with the implicit constant depending on $\ell$ only.


By \eqref{eq:abc}, and relying on the  triangle inequality, one of $\hat y_\alpha,\hat y_\beta$, or $\hat y_{\alpha+\beta}$ have at least $1/3$ of the  weight of $\delta z$ above $s_1$. Assume without loss of generality that it is $\hat y_\alpha$ , so $s_1$ is $q/3$-heavy for $\hat y_\alpha$.

We now repeat the same argument as before. This yields a sequence $s_2=e_{j'}\precdot e_{j'-1} \precdot\cdots \precdot s_1$ such that $s_2$ is $\Omega(\kappa^\ell q) = \Omega(\kappa^{2\ell})$-heavy for $\delta\hat y_\alpha=\alpha\cdot x$. Thus $s_2$ is also heavy for $x$. Since $s_0\succeq s_1\succeq s_2$ this implies $s_0\in H$,as desired.
\end{proof}
\begin{proof}[Proof of Claim \ref{claim:heavy}]
By definition,
\begin{equation}\label{eq:heavy-z}
    |z(X_{\geq e}(i))| \geq p \cdot|X_{\geq e}(i)| = p\cdot\binom{\ell-\r}{i-\r}\cdot n^{i-\r}.
\end{equation}
Let $\delta_e$ denote the coboundary operator of $X_{\geq e}$ and notice that $X_{\geq e}(i)\cong L_{[\ell-\r]}(i-\r)$. (Note that we allow $\r=i$, in which case this complex has a single element $\phi$ in level $i-\r=0$.) By coboundary expansion of $L_{[\ell-\r]}(i-\r)$,
\begin{align}
|\delta_{e}z(X_{\geq e}(i))| &\geq n\cdot \kappa_{\ell-\r,i-\r}\cdot |z(X_{\geq e}(i))| \notag\\
&\geq n\cdot \kappa_{\ell-\r,i-\r} \cdot p \cdot |X_{\geq e}(i)|\notag\\
&= \kappa_{{\ell-\r},i-\r} \cdot \frac{\binom{\ell-\r}{i-\r}}{\binom{\ell-\r}{i-\r+1} }\cdot p \cdot | X_{\geq e}(i+1)|\notag\\
&= 2 \cdot g_{i,r}(p)\cdot | X_{\geq e}(i+1)|\;. \label{eq:heavy-bound}
\end{align}
%
%
We need to compute the weight of $\delta z$ inside $X_{\geq e}(i+1)$. Recall that for $z\in C^{i}(X)$ and $s\in X_{\geq e}(i+1)$ the coboundary map is defined as
\[ \delta z (s) \,=\, \sum_{s'\precdot s} \corest_{s',s}(z) \,=\,\sum_{s'\precdot_j s} (I_{-j} \otimes h_j^T)(z(s'))[...,s_j,...]\;.\]
Separating $i$-faces into $s'\succ e$ and $s'\not\succ e$,
\begin{equation}\label{eq:cob-link}
\delta z(s) = \sum_{e\not\prec s'\precdot s} \corest_{s',s}z + \sum_{e\prec s'\precdot s} \corest_{s',s}z = \sum_{e\not\prec s'\precdot s} \corest_{s',s}z + \delta_{e}z_{e}(s)\;.
\end{equation}
We claim that the first summand on the right hand side is non-zero only for few $s\in X_{\geq e}(i)$. The key is that for each fixed $s'\not\succ e$, $\corest_{s',s}(z)$ can be non-zero for exactly one $s\in X_{\geq e}(i+1)$, namely, the one that is above both $s'$ and $e$.
Each such $s'$ can be `charged' to $e'$, the $(\r-1)$-face directly below both $e$ and $s'$ (this is the empty face in case $\r=1$). However, by assumption on $e$, no $e'\precdot e$ is $f_{i,r}(p)$-heavy for $z$. So the number of faces $s'$ that are charged to $e'$ is at most
$|z_{e'}| < f_{i,r}(p)|X_{\geq e'}(i)|$.
Summing over all $\r$ possible $e'\precdot e$, this gives at most
\[\r\cdot  f_{i,r}(p) \cdot |X_{\geq e'}(i)| =  |X_{\geq e}(i + 1)|\cdot g_{i,r}(p)\]
non-zero faces that come from the first summand in \eqref{eq:cob-link}. %
Observe that we are comparing two numbers of the same order of magnitude: the number of $(i+1)$-faces above $e$, $|X_{\geq e}(i+1)|=\binom{\ell-\r}{i+1-\r}n^{i+1-\r}$, and the number of $i$-faces above $e'$, $|X_{\geq e'}(i)|=\binom{\ell-\r+1}{i-\r}n^{i+1-\r}$.

Plugging in \eqref{eq:heavy-bound} we deduce
\begin{align*}
    |(\delta z)(X_{\geq e}(i+1))|&\geq |\delta_e (z(X_{\geq e}(i)))| - \r f_{i,r}(p) |X_{\geq e'}(i)| \\
    &\geq 2g_{i,r}(p)|X_{\geq e}(i+1)| - \r f_{i,r}(p) |X_{\geq e'}(i)| \\
    &=   g_{i,r}(p)\cdot  | X_{\geq e}(i+1)|\;,
\end{align*}
so $e$ is $g_{i,r}(p)$-heavy for $\delta z$.
\end{proof}


%
%

\section{Expansion properties of the geometric complex}
\label{sec:expansion}

In this section we we define some random walks on $X$ that will play an important role in the analysis, and show that all these random walks have good expansion properties, assuming only that the ``$1$-dimensional'' random walks on the Cayley graphs $\Cay(G,A_i)$ are expanding. Before proceeding, we start by formulating some useful incidence properties of the geometric complex $X$.

\subsection{Incidence properties}

\begin{lemma}\label{lem:geom-2}
Let $0\leq i \leq t$. Let $f\in X(i)$. Then
\begin{equation}\label{eq:down-up-val-0a}
|\down(f)|=2i \qquad \text{and}\qquad |\up(f)|= (t-i) n\;.
\end{equation}
More generally, if $0\leq \ell <i<k\leq t$ then
\begin{equation}\label{eq:down-up-val-0b}
|X_{\geq f}(k)|={t-i \choose k-i} n^{k-i}\qquad\text{and}\qquad |X_{\leq f}(\ell)| = {i \choose \ell}2^{i-\ell}\;.
\end{equation}
As a consequence,
\[ |X(k)| \,=\, \frac{1}{2^k} \sum_{v\in X(0)}  |X_{\geq v}(k)| \,=\, {t \choose k} 2^{t-k} n^k |G|\;.\]
\end{lemma}

\begin{proof}
An $i$-face is $f=[g;a,b]$ of type $S$ such that $|S|=i$. To specify an element in $\down(f)$ we specify (i) an index $j\in S$ ($i$ choices) and (ii) a value $b_i$ ($2$ choices). To specify an element in $\up(f)$ we specify (i) an index $j\notin S$ ($t-i$ choices) and (ii) a value for $a_i$ ($n$ choices).

More generally, to specify an element of $X_{\geq f}(k)$ we specify (i) a subset $S'\subseteq \ol{S}$ of size $|S'|=k-\ell$, and (ii) a value for $a_{S'}$. To specify an element of $X_{\leq f}(\ell)$ we specify (i) a subset $S'\subseteq S$ of size $|S'|=i-\ell$ and (ii) an element $b\in \{0,1\}^{S'}$.

Finally for the ``consequence'', for the first equality note that each $f\in X(k)$ lies in $X_{\geq v}(k)$ for exactly $2^k$ vertices $v\in X(0)$. For the second equality, we have $|X(0)|=2^t|G|$.
\end{proof}

\subsection{Markov operators}

We recall some basic notation and facts.

\begin{definition}
Let $V$ be a finite set. Let $\R^V$ denote the vector space of real functions on $V$ and $\langle \cdot,\cdot \rangle$ the standard inner product. For $M:\R^V\to\R^V$ linear we say that
\begin{enumerate}
\item $M$ is \emph{symmetric} if $\langle M\phi,\psi\rangle = \langle \phi,M\psi\rangle$ for any $\phi,\psi\in \R^V$;
\item $M$ is \emph{Markov} if $M\phi \geq 0$ whenever $\phi\geq 0$ (entrywise) and $M 1_V=1_V$, with $1_V$ the constant one function;
\item $M$ is \emph{$\lambda$-expanding}, for some $0<\lambda<1$, if $\langle M \phi,\phi\rangle \leq \lambda \langle \phi,\phi\rangle$ for all $\phi\in \R^V$ such that $\langle \phi,1_V\rangle=0$.
\end{enumerate}
\end{definition}

We will use the following well-known lemma.

\begin{lemma}\label{lem:ac}
Let $M:\R^V\to\R^V$ be a symmetric, Markov, $\lambda$-expanding linear operator. Then for any $x\in \R^V$,
\[ \langle x,Mx\rangle \leq \lambda \|x\|_2^2 + \frac{\|x\|_1^2}{|V|}\;,\]
where $\|x\|_1 = \sum_{i\in V}|x_i|$ and $\|x\|_2 = (\sum_{i\in V} x_i^2)^{1/2}$ denote the $1$- and $2$-norm of $x$ respectively. 
\end{lemma}

\begin{proof}
Decompose $x$ as $x=\alpha u + \beta v$, where $u,v$ are orthogonal vectors of norm $1$ such that $u$ is colinear to the all-$1$ vector. Then
\begin{align*}
\langle x,Mx\rangle &= \alpha^2 +\beta^2 \langle v,Mv \rangle \\
&\leq \alpha^2 + \lambda \beta^2\\
&= \alpha^2(1-\lambda)+\lambda\|x\|_2^2\;,
\end{align*}
where the inequality uses that $v\bot u$ and the assumption that $M$ is $\lambda$-expanding. Moreover,
\begin{align*}
\alpha^2 &= |\langle x,u\rangle|^2 \\
&\leq \frac{\|x\|_1^2}{|V|}\;,
\end{align*}
because $u_i = 1/\sqrt{|V|}$ for all $i$. Combining both inequalities completes the proof.
\end{proof}

Given a graph $G=(V,E)$, we say that $G$ is $\lambda$-expanding if the Markov operator that is given by the normalized adjacency matrix of $G$ is $\lambda$-expanding.
This notion is only suitable for graphs that are connected.
However, the constant degree complex described in Section~\ref{sec:particular-instance-of-complex} is not connected, so we have to generalize the notion.
We say $G$ is $\lambda$-expanding up to size $r |V|$ if the graph is a disjoint union of graphs of size $\ge r |V|$ where each graph is $\lambda$-expanding.
In particular, the constant degree complex satisfies $r = \Omega((\log |G|)^{-(t-1)})$.
For each set $A_j$, we let $G_j=\Cay(G,A_j)$ denote the graph whose vertex set is $V_j = G$, and such that there is an edge between $(g,g')$ if and only if $g'= g\cdot a_j$ for some $a_j\in A_j$. Since we assumed that $A_j$ is closed under inverse, $\Cay(G,A_j)$ is an $n$-regular undirected graph (which for simplicity one may assume does not contain any self-loops, although our arguments extend verbatim if there are).

\begin{definition}\label{def:x-expand}
For a set of permutations $A$ on $G$, we say that the pair $(G;A)$ is \emph{$\lambda$-expanding} if the Markov operator $M:\R^G \to \R^G$ such that
\begin{equation}\label{eq:def-Markov-cay}
  M(e_g)\,=\,\frac{1}{|A|}\sum_{a\in A} e_{g\cdot a}
\end{equation} is $\lambda$-expanding, where $e_g$ denotes the $g$-th canonical basis vector of $\R^G$.

More generally, for $0<r\leq 1$ we say that $(G;A)$ is \emph{$\lambda$-expanding up to size $r |G|$} if the graph is a disjoint union of graphs of size $\ge r |G|$ where each graph is $\lambda$-expanding.
\end{definition}

Let $M_j$ be the Markov operator associated with the normalized adjacency matrix of $G_j=\Cay(G,A_j)$ as in~\eqref{eq:def-Markov-cay}.
Hereafter, we assume that $G_j$ is $\lambda$-expanding up to size $r |G|$ for all $j$.

Let $D:\R^X\to \R^X$ be the normalized unsigned boundary operator, i.e.\
\[ D\phi(f) \,=\, \frac{1}{|\up\{f\}|}\sum_{f'\succdot f} \phi(f')\;.\]
Note that $|\up\{f\}|$ is independent of $f$ (see Lemma~\ref{lem:geom-2}), and so this is a linear operator.
Let $U$ be the normalized unsigned coboundary operator, i.e.\
\[ U\phi(f) \,=\, \frac{1}{|\down\{f\}|}\sum_{f'\precdot f} \phi(f')\;.\]

\begin{claim}\label{claim:d-u-adjoint}
The operators $D$ and $U$ satisfy that for any $\ell$ and $\phi:X(\ell)\to \R$, $\phi':X(\ell-1)\to\R$,
\[ \E_{f\in X(\ell-1)}  \phi'(f) D \phi(f) \,=\, \E_{f'\in X(\ell)}  U\phi'(f') \phi(f')\;,\]
where both expectations are uniform.
\end{claim}

\begin{proof}
  To verify this, for any $\phi,\phi'$ we have
  \begin{align*}
  \frac{1}{X(\ell-1)}\sum_{f\in X(\ell-1) } \phi'(f) D\phi(f)
  &=  \frac{1}{X(\ell-1)}\sum_{f\in X(\ell-1) } \frac{1}{\up\{f\}} \sum_{ f'\succdot f} \phi'(f) \phi(f') \\
  &= \frac{1}{X(\ell)} \sum_{f'\in X(\ell)} \frac{1}{\down\{f'\}}\sum_{f\precdot f'}\phi'(f) \phi(f') \\
  &= \frac{1}{X(\ell)}\sum_{f'\in X(\ell) } U\phi'(f') \phi(f')
  \;.
  \end{align*}
  Here the second line is because for any $f\precdot f'$, $\frac{X(\ell-1)}{X(\ell)} = \frac{\down\{f'\}}{\up\{f\}}$.
  \end{proof}

\subsection{Expanding random walks}
\label{sec:rw}

We define and analyze some random walks on the vertex set $X(k)$.

\begin{definition}\label{def:walk-W}
  For $0\leq \ell \leq k$ define the following walk $W^{(k,\ell)}$ on $X(k)$. Starting at $f\in X(k)$,
  \begin{enumerate}
  \item Choose a uniformly random $v \in X(\ell)$ such that $v\in \down^{(k-\ell)}\{f\}$. (If $\ell=k$ then set $v=f$.)
  \item Write $v = [g;a,b]$ and $S=\type(v)$.
  Choose a uniformly random $i \in \ol{S}$ and $a_i \in A_i$.
  Let $g' = g a_i$, $a'_j=a_j$ for $j\in S$, $b'_i=1-b_i$, and $b'_j=b_j$ for $j\in \ol{S} - \{i\}$. Let $v'=[g';a',b']$.
  \item Return a uniformly random $f'\in \up^{(k-\ell)}\{v'\}$.
  \end{enumerate}
  We write $W^{(k,\ell)} : \R^{X(k)}\to\R^{X(k)}$ for the Markov operator associated with this walk
    and $W^{(k,\ell)}_{adj} : \R^{X(k)}\to\R^{X(k)}$ for the adjacency matrix associated with this walk.
\end{definition}

\begin{definition}
  Let $0 \le \ell \le k < t$ and $v \in X(\ell)$.
  Recall the definition of the set
  \[(\down\circ\up)X_{\geq v}(k) = \big\{ f\in X(k)\;\big|\; \exists f'\in X_{\geq v}(k+1)\,,\; f\precdot f'\big\}\;.\]
  In words, these are $k$-faces that are included in a $(k+1)$-face that contains $v$.
  Such faces can either contain $v$, or intersect but not contain $v$, or not intersect $v$.
  We write the set of $k$-faces that intersect but do not contain $v$ as
  \[\Nb_v(k) = \big\{f\in X(k)\;\big| \; \exists f'\in X_{\geq v}(k+1)\,,\; f\precdot f'\,,\; f\cap v \neq \emptyset, v \nprec f \big\}\;,\]
  and those that do not intersect $v$ as
  \[\Op_{v}(k) = \big\{ f\in X(k)\;\big|\;  \exists f'\in X_{\geq v}(k+1)\,,\; f\precdot f'\,,\; f\cap v = \emptyset\big\}\;.\]
\end{definition}

From the definition, it is clear that
\begin{equation}\label{eq:d-u}
  (\down\circ\up)X_{\geq v}(k) \,=\, X_{\geq v}(k) \cup \Nb_v(k) \cup \Op_{v}(k)\;.
\end{equation}
The following lemma states a useful ``covering'' property for the neighborhoods $\Nb_v(k)$. First, for any $0\leq \ell < k < t$ we define
\begin{equation}\label{eq:def-akl}
   a_{k, \ell} \,=\, \max_{\substack{v_{\ell}\in X(\ell),\, v_k,v'_k\in X(k):\\ v_\ell \prec v_k,\,v_\ell\prec v'_k}} \big|\big\{v_{\ell+1}\in X(\ell+1):\, v_{\ell} \prec v_{\ell+1} \preceq v_k\, \hbox{ and }\, v_k' \in \Nb_{v_{\ell+1}}(k)\big\}\big|\;.
\end{equation}
We observe the crude bound
\begin{equation}\label{eq:akl-bound}
  a_{k,\ell}\,\leq\, |X_{\leq v_k}(\ell+1)| \,\leq\, 2^t
\end{equation}
for all $\ell,k$. Here the first inequality is because $v_{\ell+1}$ in~\eqref{eq:def-akl} must be in $X_{\geq v_k}(\ell+1)$, and the second inequality is by Lemma~\ref{lem:geom-2}. Furthermore, we observe the following claim (which will not be used in the proof).

\begin{claim}\label{claim:akl}
If $k=t-1$ then $a_{t-1,\ell}=1$ for all $\ell<t-1$.
\end{claim}

\begin{proof}
Fix $v_\ell\in X(\ell)$ and $v_k,v'_k\in X(k)$ such that $v_\ell \prec v_k$ and $v_\ell\prec v'_k$. We first show that any $v_{\ell+1}\in X(\ell+1)$ such that $v_\ell\prec v_{\ell+1} \preceq v_k$ and $v'_k\in \Nb_{v_{\ell+1}}(k)$ has $\type(v_{\ell+1})=\type(v_\ell)\cup ([t]\setminus\type(v'_k))$. This is because using the first condition necessarily $\type(v_{\ell+1})\supseteq \type(v_{\ell})$ and has exactly one more element; furthermore, using the second condition that element must be the unique element of $[t]$ that is not in $\type(v'_k)$. Once its type has been fixed, $v_{\ell+1}$ is uniquely determined by the condition $v_\ell\prec v_{\ell+1}\preceq v_k$.
\end{proof}

\begin{lemma}\label{lem:Nb-bound}
  For all $v_k \in X(k)$ and $\ell<k$,
    \begin{equation*}
    \bigsqcup_{v_{\ell+1} \preceq v_k} \Nb_{v_{\ell+1}}(k) \,\subseteq\, a_{k, \ell} \cdot \bigsqcup_{v_\ell \prec v_k} X_{\ge v_\ell}(k)\;,
  \end{equation*}
  where $\sqcup$ denotes the union of multisets, i.e.\ elements are counted with multiplicity, and $a\cdot S$ is all elements of the multiset $S$ with multiplicity multiplied by $a$.
\end{lemma}

\begin{proof}
Let $v_k\in X(k)$. For $v_{\ell+1}\in X(\ell+1)$, each face $v_k' \in \Nb_{v_{\ell+1}}(k)$ intersects $v_{\ell+1}$ on some $\ell$-dimensional face.
  This is because, since $\Nb_{v_{\ell+1}}(k) \subseteq (\down\circ\up)X_{\geq v_{\ell+1}}(k)$ by definition,
    the types of $v_k$ and $v_k'$ must differ by at most one element.
  Additionally, the intersection could not be an $\ell+1$-dimensional face, otherwise $v_{\ell+1} \preceq v_k'$ which implies $v_k' \notin \Nb_{v_{\ell+1}}(k)$.
  Therefore,
  \begin{align*}
    \bigsqcup_{v_{\ell+1} \preceq v_k} \Nb_{v_{\ell+1}}(k)
    &= \bigsqcup_{v_{\ell+1} \preceq v_k} \bigsqcup_{v_k' \in \Nb_{v_{\ell+1}}(k)} \{v_k'\} \\
    &= \bigsqcup_{v_{\ell+1} \preceq v_k} \bigsqcup_{v_k' \in \Nb_{v_{\ell+1}}(k)} \bigsqcup_{v_{\ell} \prec v_{\ell+1}} 1_{ v_{\ell+1} \cap v_k' = v_{\ell}} \cdot \{v_k'\} \\
    &= \bigsqcup_{v_{\ell} \prec v_k} \bigsqcup_{v_k' \succ v_{\ell}} \bigsqcup_{v_{\ell+1}: v_{\ell} \prec v_{\ell+1} \preceq v_k}  1_{ v_{\ell+1} \cap v_k' = v_{\ell}} \cdot 1_{v_k' \in \Nb_{v_{\ell+1}}(k)}\cdot \{v_k'\} \\
    &\subseteq \bigsqcup_{v_{\ell} \prec v_k} \bigsqcup_{v_k' \succ v_{\ell}} a_{k, \ell} \cdot \{v_k'\} \\
    &= a_{k, \ell} \cdot \bigsqcup_{v_\ell \prec v_k} X_{\ge v_\ell}(k)\;.
  \end{align*}
\end{proof}

\begin{definition}\label{def:opp-walk}
  For $0\leq \ell \leq k$ define the following walk $\Op^{(k,\ell)}$ on $X(k)$. Starting at $f\in X(k)$,
  \begin{enumerate}
  \item Choose a uniformly random $v \in X(\ell)$ such that $v\in \down^{(k-\ell)}\{f\}$. (If $\ell=k$ then set $v=f$.)
  \item Return a uniformly random $f'\in \Op_{v}(k)$.
  \end{enumerate}
  We write $\Op^{(k,\ell)} : \R^{X(k)}\to\R^{X(k)}$ for the Markov operator associated with this walk and $\Op^{(k,\ell)}_{adj} : \R^{X(k)}\to\R^{X(k)}$ for the adjacency matrix associated with this walk.
\end{definition}

The random walk $\Op^{(k,\ell)}$ is used in the proof of co-systolic distance. For the analysis, it will be convenient to relate it to the walk $W^{(k,\ell)}$. This is done through the following claim.

\begin{claim}\label{claim:Op-to-W}
For any $\mA \subseteq X(k)$,
  \begin{equation*}
    \langle 1_\mA, \Op^{(k,\ell)}_{adj} 1_\mA\rangle
    \,\le\, \langle 1_\mA, W^{(k,\ell)}_{adj} 1_\mA\rangle
    \,=\, {k \choose \ell}{t-\ell \choose k-\ell} (t-\ell) 2^{k-\ell} n^{k+1-\ell} \langle 1_\mA, W^{(k,\ell)} 1_\mA\rangle\;.
  \end{equation*}
\end{claim}

\begin{proof}
  It is not hard to check that $\langle 1_\mA, \Op^{(k,\ell)}_{adj} 1_\mA\rangle \le \langle 1_\mA, W^{(k,\ell)}_{adj} 1_\mA\rangle$ from the definitions.
  To show that $W^{(k,\ell)}_{adj} = {k \choose \ell}{t-\ell \choose k-\ell} (t-\ell) 2^{k-\ell} n^{k+1-\ell} W^{(k,\ell)}$,
  we observe that the normalization factor for $W^{(k,\ell)}$ is, using notations from Definition~\ref{def:walk-W},
  \begin{align*}
    |X_{\le f}(\ell)| |\ol{S}| |A_i| |X_{\ge v'}(k)|
    &= \left({k \choose \ell} 2^{k-\ell}\right) (t-\ell) (n) \left({t-\ell \choose k-\ell} n^{k-\ell}\right) \\
    &= {k \choose \ell}{t-\ell \choose k-\ell} (t-\ell) 2^{k-\ell} n^{k+1-\ell}\;,
  \end{align*}
  where the first equality uses bounds from Lemma~\ref{lem:geom-2}.
\end{proof}

One would hope that $W^{(k,\ell)}$ has spectral expansion.
However, $W^{(k,\ell)}$ shares similarities to the random walk on a hypercube,
  which is not a good spectral expander.
Nevertheless, $W^{(k, \ell)}$ has small-set expansion, as the next lemma shows.

\begin{lemma}\label{lem:graph-expansion}
  Assume that for each $j\in \{1,\ldots,t\}$ the pair $(G;A_j)$ is $\lambda$-expanding. Let $0\leq \ell \leq k$ and $\mA\subseteq X(k)$.
  Then
  \begin{equation}\label{eq:graph-expansion}
    \langle1_\mA, W^{(k,\ell)} 1_\mA\rangle
    \,\le\, \lambda |\mA| + {t \choose \ell} 2^{t-\ell-1} n^\ell \frac{|\mA|^2}{r|X(k)|}\;.
  \end{equation}
\end{lemma}

Note that $\lambda |\mA|$ is the dominating term when $\frac{|\mA|}{|X(k)|}$ is a small constant.

\begin{proof}
  Let $M$ be the Markov operator associated to step 2 of the walk $W^{(k,\ell)}$ introduced in Definition~\ref{def:walk-W}.
  We first discuss the structural properties of $M$.
  Notice that the type of the face does not change after a step of the walk $M$.
  Thus, $M$ is a direct sum of Markov operators $M_S$, for $S$ a subset of size $\ell$, each of which corresponds to the restriction of $M$ to faces of type $S$.
  $M_S$ can be further decomposed into a sum of operators $M_{S, i}$, for $i\in\ol{S}$, associated to the walk in the $i$-th direction. This decomposition is not a direct sum, but an average.
  Finally, $M_{S, i}$ is isomorphic to the direct sum of $n^\ell 2^{t-\ell-1}$ disjoint copies of the walk on the double cover of $\Cay(G, A_i)$. This is because $a_S$ and $b_{\ol{S}-i}$ are not changed by $M_{S,i}$; only $g$ and $b_i$ are.
  Denote the random walk on the double cover of $\Cay(G, A_i)$ as $P_i$.
  To summarize,
  \[M \,=\, \sum_{|S| = \ell} M_{S}\;,\qquad M_{S} \,=\, \frac{1}{t-\ell} \sum_{i \in \ol{S}} M_{S, i}\;,\qquad\text{and}\qquad M_{S,i} \,=\, P_i^{\oplus n^\ell 2^{t-\ell-1}}\;.\]
We now prove the lemma. For $P_i$,
    by applying the expander mixing lemma, Lemma~\ref{lem:ac}, on the double cover of $\Cay(G, A_i)$
    we have
  \begin{equation*}
    \langle x, P_i x\rangle \,\le\, \lambda \norm{x}_2^2 + \frac{\norm{x}_1^2}{2r|G|}
  \end{equation*}
  for all real functions $x$ on the vertices.
  The factor $2r|G|$ appears in the denominator because each connected component is $\lambda$-expanding with size at least $2r|G|$.
  Because $M_{S,i}$ is a direct sum of $P_i$, the same inequality carries over:
  \begin{equation*}
    \langle x, M_{S, i} x\rangle \,\le\, \lambda \norm{x}_2^2 + \frac{\norm{x}_1^2}{2r|G|}\;.
  \end{equation*}
  Because $M_{S} = \frac{1}{t-\ell} \sum_{i \in \ol{S}} M_{S, i}$,
  \begin{equation*}
    \langle x, M_{S} x\rangle
    \,=\, \frac{1}{t-\ell} \sum_{i \in \ol{S}} \langle x, M_{S, i} x\rangle
    \,\le\, \lambda \norm{x}_2^2 + \frac{\norm{x}_1^2}{2r|G|}\;.
  \end{equation*}
Finally, since $M$ is a direct sum of $M_S$, we obtain that for any real function $x$ on the faces $X(\ell)$,
  \begin{equation}\label{eq:m-bound-1}
    \langle x, M x\rangle \,\le\, \lambda \norm{x}_2^2 + \frac{\norm{x}_1^2}{2r|G|}\;.
  \end{equation}
We can now evaluate
  \begin{align*}
    \langle 1_\mA, W^{(k,\ell)} 1_\mA\rangle
    &= \langle 1_\mA, U^{\circ(k-\ell)} M D^{\circ(k-\ell)} 1_\mA\rangle \\
    &= \frac{|X(k)|}{|X(\ell)|} \langle D^{\circ(k-\ell)} 1_\mA, M D^{\circ(k-\ell)} 1_\mA\rangle \\
    &\le \frac{|X(k)|}{|X(\ell)|} \lambda \langle D^{\circ(k-\ell)} 1_\mA, D^{\circ(k-\ell)} 1_\mA\rangle + \frac{|X(k)|}{|X(\ell)|} \frac{\langle 1_{X(\ell)}, D^{\circ(k-\ell)} 1_\mA\rangle^2}{2|G|} \\
    &=\frac{|X(k)|}{|X(\ell)|} \lambda \langle D^{\circ(k-\ell)} 1_\mA, D^{\circ(k-\ell)} 1_\mA\rangle + \frac{|X(k)|}{|X(\ell)|} \frac{\langle D^{\circ(k-\ell)} 1_{X(k)}, D^{\circ(k-\ell)} 1_\mA\rangle^2}{2r|G|} \\
    &\le \lambda \langle 1_\mA, 1_\mA\rangle + \frac{|X(\ell)|}{|X(k)|} \frac{\langle1_{X(k)}, 1_\mA\rangle^2}{2r|G|} \\
    &\le \lambda |\mA| + \frac{|X(\ell)|}{|X(k)|} \frac{|\mA|^2}{2r|G|} \\
    &\le \lambda |\mA| + \frac{{t \choose \ell} 2^{t-\ell} n^\ell |G|}{2r|G|} \frac{|\mA|^2}{|X(k)|}\;.
  \end{align*}
Here the second line uses Claim~\ref{claim:d-u-adjoint}.
  The third line uses~\eqref{eq:m-bound-1} and the fact that $D^{\circ(k-\ell)} 1_\mA$ has nonnegative entries,
    which implies that $\norm{D^{\circ(k-\ell)} 1_\mA}_1 = \langle1_{X(\ell)}, D^{\circ(k-\ell)} 1_\mA\rangle$. The fourth and fifth lines use that $D$ is an averaging operator, and the regularity properties of the complex.
  The last line uses $|X(\ell)| = {t \choose \ell} 2^{t-\ell} n^\ell |G|$ by Lemma~\ref{lem:geom-2}.
\end{proof}

\section{Locally co-minimal distance}
\label{sec:codistance}

The main result of this section is that for any $0\leq k <t$, the co-chain complex $C^*(X,\sheaf)$ is a co-systolic expander in dimension $k$ according to Definition~\ref{def:cosyst-exp}, assuming the expansion parameter $\lambda$ (see~\ref{def:x-expand}) is small enough as a function of $t$ and the robustness parameters introduced in Section~\ref{sec:robustness}.

Quantitatively, we first prove a lower bound on the locally co-minimal distance of $C^k(X)$. This is introduced in Section~\ref{sec:prelim-chains}. For convenience we recall the definition
\begin{equation}\label{eq:def-coloc}
   \dist_\coloc(k) \,=\, \min\big\{ |x| \,:\, x\in \ker\delta_k - \{0\},\, \text{$x$ is locally co-minimal}\big\}\;.
\end{equation}
Bounds on the co-cycle expansion parameters follow immediately from a lower bound on $\dist_\coloc(k)$, as shown in Lemma~\ref{lem:cosys-exp}.

\begin{proposition}\label{prop:codistance-t}
  Let $0\leq k < t$.
Assume that there exists $\kappa_{t-i,k-i} >0$, for each $0\leq i \leq k$, such that for each $S\subseteq \{1,\ldots,t\}$ with $|S|=t-i$ the local chain complex $C(L_S)$ is $\kappa_{t-i,k-i}$-robust (at level $k-i$) according to Definition~\ref{def:robust}.  Then
\begin{equation*}
  \dist_\coloc(k) \,\ge\, \frac{1-\lambda C_1}{C_2} r |X(k)|\;,
\end{equation*}
where
\begin{align*}
  C_1 \,=\, \frac{t^2 2^{t^2+3t}}{\prod_{i=0}^{k} \kappa_{t-i,k-i}}\qquad\text{and}\qquad
  C_2 \,=\, \frac{t^2 2^{t^2+5t} n^t}{\prod_{i=0}^{k} \kappa_{t-i,k-i}}\;.
\end{align*}
\end{proposition}

\begin{remark}\label{remark:better-bound}
  The constants $C_1$ and $C_2$ are simplifications
    of constants $C_1'$ and $C_2'$ which provide a slightly tighter bound, see~\eqref{eq:c1-c2} in the proof of Proposition~\ref{prop:codistance-t}, where the coefficients $a_{k,\ell}$ are defined in~\eqref{eq:def-akl}. To simplify these bounds, we
    we apply the estimate $a_{k, \ell} \le 2^t$ from~\eqref{eq:akl-bound}.\footnote{The bound is often much tighter, for example $a_{t-1,\ell} = 1$ for all $\ell$ as shown in Claim~\ref{claim:akl}. Such tighter bounds could be of interest if one cares about the dependence of $C'_1$, $C'_2$ on the dimension $t$.}
  Then,
  \begin{equation*}
    C'_1
    = \sum_{\ell=0}^k \frac{{k \choose \ell}{t-\ell \choose k-\ell} (t-\ell) 2^{k-\ell} \prod_{i=\ell}^{k-1} a_{k, i}}{\prod_{i=\ell}^{k} \kappa_{t-i,k-i}}
    \le t \frac{2^t 2^t t 2^t (2^t)^t}{\prod_{i=0}^{k} \kappa_{t-i,k-i}}
    = C_1\;,
  \end{equation*}
  and
  \begin{equation*}
    C'_2
    = \sum_{\ell=0}^k \frac{{t \choose \ell} 2^{t-\ell-1} n^\ell {k \choose \ell}{t-\ell \choose k-\ell} (t-\ell) 2^{k-\ell} \prod_{i=\ell}^{k-1} a_{k, i}}{\prod_{i=\ell}^{k} \kappa_{t-i,k-i}}
    \le t \frac{2^t 2^t n^t 2^t 2^t t 2^t (2^t)^t}{\prod_{i=0}^{k} \kappa_{t-i,k-i}}
    = C_2\;.
  \end{equation*}
\end{remark}

\begin{remark}
  We explicitly work out the requirements on $\lambda$ and the $\kappa_{\ell,k}$ that suffice to guarantee a linear (in $X(k)$) co-systolic distance, i.e.\ such that $1-\lambda C'_1 >0$, for certain cases of interest.
    We use the shorthand $\kappa_i = \kappa_{i, i-1}$ for the robustness at the last level of the chain complex.
  Note that when $k = t-1$, by Claim~\ref{claim:akl} we have that $a_{k, i} = 1$.

  When $t = 2, k = 1$ the sufficient condition is
  \begin{equation}
    \frac{1}{\kappa_1}
    + \frac{8}{\kappa_1 \kappa_2} \,<\, \frac{1}{\lambda}\;.
  \end{equation}
  This is the sufficient condition for obtaining good qLDPC using our construction.
  When $t = 4, k = 3$ the sufficient condition is
  \begin{equation}
    \frac{1}{\kappa_1}
    + \frac{24}{\kappa_1 \kappa_2}
    + \frac{108}{\kappa_1 \kappa_2 \kappa_3}
    + \frac{128}{\kappa_1 \kappa_2 \kappa_3 \kappa_4} < \frac{1}{\lambda}\;.
  \end{equation}
  This is the sufficient condition for obtaining good (again, up to polylog factors) qLTC using our construction.
  \end{remark}

Let $k< t$ and $x\in C^k(X)$ be such that $\delta(x)=0$. For $f\in X(k)$ we say that $f$ is \emph{active} if the projection $x(f)$ of $x$ on $V_f$ is nonzero. Let $\mA\subseteq X(k)$ be the set of {active faces}. At a high level, the proof of Proposition~\ref{prop:codistance-t} amounts to showing that, whenever $x$ is locally co-minimal, the set $\mA$ only expands a little under a suitable expanding random walk, which is based on the robustness of the local codes. However, when $\mA\neq\emptyset$, the spectral expansion of the graph implies to have such small expansion, the set must be large, leading to a lower bound on $|x|$.






\begin{proof}[Proof of Proposition~\ref{prop:codistance-t}]
Let $x\in C^k(X)$ be such that $\delta(x)=0$ and $x$ is locally co-minimal. We start with a preliminary claim, which shows that the locally co-minimal condition implies the following.

\begin{claim}\label{claim:active-rob}
  Let $\ell<k$ and $v\in X(\ell)$. Using the isomorphism from Lemma~\ref{lem:loc-glob}, the restriction $x(X_{\geq v}(k))$ can be seen as a $(k-\ell)$-chain of the dimension $(t-\ell)$ complex $C(L_{\ol{S}})$ where $S = \type\,{v}$. Let $z$ be that element. Then $z$ is minimal (according to Definition~\ref{def:min}).
\end{claim}

\begin{proof}
  It suffices to show that
  \begin{equation*}
   \forall y\text{ supported on } X_{\geq v}(k-1)\;,\qquad \big| x\big(X_{\geq v}(k)\big) + \delta_{\ol{S}}(y) \big| \,\geq\, \big| x\big( X_{\geq v}(k)\big)\big|\;.
  \end{equation*}
  Let $y$ be supported on $X_{\geq v}(k-1)$. Then  $\delta_{\ol{S}}(y) = \delta(y)(X_{\geq v}(k))$. Thus $x+\delta(y)$ decomposes as a sum
  \[ x+ \delta(y)\,=\, \big(x\big(X_{\geq v}(k)\big) + \delta_{\ol{S}}(y)\big) + \big(x-x\big(X_{\geq v}(k)\big)\big)\;,\]
  where the two elements on the right-hand side are supported on a disjoint set of faces: $X_{\geq v}(k)$ for the first and its complement in $X(k)$ for the second.
  Thus
  \begin{align*}
   \big| x\big(X_{\geq v}(k)\big) + \delta_{\ol{S}}(y) \big|
  &=   \big| x+ \delta(y) \big| - \big| x-x\big(X_{\geq v}(k)\big)\big|\\
  &\geq \big| x\big|- \big| x-x\big(X_{\geq v}(k)\big) \big|\\
  &=  \big| x\big(X_{\geq v}(k)\big)\big|\;,
  \end{align*}
  where the second line uses local minimality of $x$ (according to Definition~\ref{def:co-loc-min-dist}).
\end{proof}

The key step in the proof is the following claim, which uses the robustness assumption on the local chain complex $C(L_S)$.

\begin{claim}\label{claim:local-code-expansion}
Let $\ell\leq k$.  For all $v \in X(\ell)$,
  \begin{equation}\label{eq:local-code-expansion}
    \kappa_{t-\ell, k-\ell} n\ |x(X_{\ge v}(k))| \,\le\, |x(\Op_{v}(k))| + |x(\Nb_{v}(k))|\;.
  \end{equation}
\end{claim}

\begin{proof}
  We apply the robustness condition to $z = x(X_{\ge v}(k))$,
    seen as a $(k-\ell)$-chain of $C(L_{\ol{S}})$
    where $S = \type\,{v}$.
  By Claim~\ref{claim:active-rob}, $z$ is minimal.
  Therefore,
  \begin{equation*}
    \kappa_{t-\ell, k-\ell} n\ |x(X_{\geq v}(k))|\, \le\, |\delta_{\ol{S}}(x(X_{\geq v}(k)))|\;.
  \end{equation*}
  Because $\delta x = 0$ and $(\down\circ\up)X_{\geq v}(k) = X_{\geq v}(k) \cup \Nb_v(k) \cup \Op_{v}(k)$ by~\eqref{eq:d-u},
  each nonzero entry in $\delta_{\ol{S}}(x(X_{\geq v}(k)))$
  has to be canceled by some nonzero entries in $x(\Op_{v}(k))$ or $x(\Nb_{v}(k))$. Moreover, each nonzero entry of $x(\Op_{v}(k))$ or $x(\Nb_{v}(k))$ cancels at most one nonzero entry of $\delta_{\ol{S}}(x(X_{\geq v}(k)))$.
  Hence,
  \begin{equation*}
    |\delta_{\ol{S}}(x(X_{\geq v}(k)))| \le |x(\Op_{v}(k))| + |x(\Nb_{v}(k))|\;.
  \end{equation*}
  Combining the two inequalities above gives~\eqref{eq:local-code-expansion}.
\end{proof}

Informally, Claim~\ref{claim:local-code-expansion} for $\ell=k$ states that whenever a face $v\in X(k)$ is active, i.e.\ $|x(X_{\geq v}(k))|=|x(v)|=1$, it must have either many neighbors that are active, or many opposite neighbors that are active. In the latter case, we can make a step according to the random walk $\Op^{(k,k)}$ introduced in Defition~\ref{def:opp-walk}, and use the expansion properties of that walk to conclude. However, in the former case a step to $\Nb_v(k)$ would not be expanding (because neighbors share a lower-dimensional face). In that case, we instead first sample a lower-dimensional sub-face of $v$, and move to an opposite neighbor of the latter face. The following claim analyses this recursive process. Recall the definition of the coefficients $a_{k,\ell}$ in~\eqref{eq:def-akl}.

\begin{claim}\label{claim:global-code-expansion}
Let $v_k \in X(k)$. Then
  \begin{equation}\label{eq:global-code-expansion}
    |x(v_k)| \,\le\, \sum_{\ell=0}^k \sum_{v_\ell \preceq v_k} \frac{\prod_{i=\ell}^{k-1} a_{k, i}}{\prod_{i=\ell}^{k} (\kappa_{t-i,k-i} n)} |x(\Op_{v_\ell}(k))|\;.
  \end{equation}
\end{claim}

\begin{proof}
  We apply Claim~\ref{claim:local-code-expansion} recursively, yielding the following sequence of inequalities.
  \begin{align*}
    |x(v_k)|
    &\le \frac{|x(\Op_{v_k}(k))|}{\kappa_{t-k,0} n} + \frac{|x(\Nb_{v_k}(k))|}{\kappa_{t-k,0} n} \\
    &\le \frac{|x(\Op_{v_k}(k))|}{\kappa_{t-k,0} n} + a_{k, k-1} \sum_{v_{k-1} \prec v_k} \frac{|x(X_{\ge v_{k-1}}(k))|}{\kappa_{t-k,0} n} \\
    &\le \frac{|x(\Op_{v_k}(k))|}{\kappa_{t-k,0} n}
    + a_{k, k-1} \sum_{v_{k-1} \prec v_k} \left(\frac{|x(\Op_{v_{k-1}}(k))|}{\kappa_{t-k,0} n \cdot \kappa_{t-k+1,1} n}
    + \frac{|x(\Nb_{v_{k-1}}(k))|}{\kappa_{t-k,0} n \cdot \kappa_{t-k+1,1} n}\right) \\
    &\le \frac{|x(\Op_{v_k}(k))|}{\kappa_{t-k,0} n}
    + a_{k, k-1} \sum_{v_{k-1} \prec v_k} \frac{|x(\Op_{v_{k-1}}(k))|}{\kappa_{t-k,0} n \cdot \kappa_{t-k+1,1} n} \\
    &\hspace{14em}+ a_{k, k-1} a_{k, k-2} \sum_{v_{k-2} \prec v_k} \frac{|x(X_{\ge v_{k-2}}(k))|}{\kappa_{t-k,0} n \cdot \kappa_{t-k+1,1} n} \\
    &\le \cdots \\
    &\le \sum_{\ell=0}^k \sum_{v_\ell \prec v_k} \frac{\prod_{i=\ell}^{k-1} a_{k, i}}{\prod_{i=\ell}^{k} (\kappa_{t-i,k-i} n)} |x(\Op_{v_\ell}(k))|\;.
  \end{align*}
  Here the second inequality, and every other subsequent step thereafter, uses Lemma~\ref{lem:Nb-bound}.
  Note that the recursion terminates because $\Nb_{v_0}(k) = \emptyset$ for all $v_0 \in X(0)$.
\end{proof}

The previous claim draws the consequences of the robustness assumption on the local chains $C(L_S)$, which can be understood as a ``local'' expansion property. We now conclude the proof by combining this with the global expansion properties of the cubical complex, which are manifest in the expansion properties of the walks $W^{(k,\ell)}$
introduced in Section~\ref{sec:rw}, see Lemma~\ref{lem:graph-expansion}.
For any $v_k\in X(k)$, applying Claim~\ref{claim:global-code-expansion} we have
\begin{align*}
  \langle1_\mA, 1_{v_k}\rangle
  = |x(v_k)|
  &\le
  \sum_{\ell=0}^k \sum_{v_\ell \preceq v_k} \frac{\prod_{i=\ell}^{k-1} a_{k, i}}{\prod_{i=\ell}^{k} (\kappa_{t-i,k-i} n)} |x(\Op_{v_\ell}(k))| \\
  &=
  \sum_{\ell=0}^k \langle1_\mA, \Op^{(k,\ell)}_{adj} 1_{v_k}\rangle \frac{\prod_{i=\ell}^{k-1} a_{k, i}}{\prod_{i=\ell}^{k} (\kappa_{t-i,k-i} n)} \\
  &\le
  \sum_{\ell=0}^k \langle1_\mA, W^{(k,\ell)} 1_{v_k}\rangle \frac{{k \choose \ell}{t-\ell \choose k-\ell} (t-\ell) 2^{k-\ell} \prod_{i=\ell}^{k-1} a_{k, i}}{\prod_{i=\ell}^{k} \kappa_{t-i,k-i}}\;,
\end{align*}
where the third line follows from Claim~\ref{claim:Op-to-W}.
By averaging over $v_k \in X(k)$, we obtain
\begin{align*}
  \frac{|\mA|}{|X(k)|} = \frac{1}{|X(k)|} \langle1_\mA, 1_\mA\rangle
  &\le \frac{1}{|X(k)|} \sum_{\ell=0}^k \langle1_\mA, W^{(k,\ell)} 1_\mA\rangle \frac{{k \choose \ell}{t-\ell \choose k-\ell} (t-\ell) 2^{k-\ell} \prod_{i=\ell}^{k-1} a_{k, i}}{\prod_{i=\ell}^{k} \kappa_{t-i,k-i}} \\
  &\le \frac{|\mA|}{|X(k)|} \left(\lambda \sum_{\ell=0}^k \frac{{k \choose \ell}{t-\ell \choose k-\ell} (t-\ell) 2^{k-\ell} \prod_{i=\ell}^{k-1} a_{k, i}}{\prod_{i=\ell}^{k} \kappa_{t-i,k-i}}\right) \\
  &\qquad+ \frac{1}{r} \left(\frac{|\mA|}{|X(k)|}\right)^2 \left(\sum_{\ell=0}^k \frac{{t \choose \ell} 2^{t-\ell-1} n^\ell {k \choose \ell}{t-\ell \choose k-\ell} (t-\ell) 2^{k-\ell} \prod_{i=\ell}^{k-1} a_{k, i}}{\prod_{i=\ell}^{k} \kappa_{t-i,k-i}}\right)\;,
\end{align*}
where the last line follows from Lemma~\ref{lem:graph-expansion}.
Let
\begin{align}
  C_1' &= \sum_{\ell=0}^k \frac{{k \choose \ell}{t-\ell \choose k-\ell} (t-\ell) 2^{k-\ell} \prod_{i=\ell}^{k-1} a_{k, i}}{\prod_{i=\ell}^{k} \kappa_{t-i,k-i}}\;,\notag \\
  C_2' &= \sum_{\ell=0}^k \frac{{t \choose \ell} 2^{t-\ell-1} n^\ell {k \choose \ell}{t-\ell \choose k-\ell} (t-\ell) 2^{k-\ell} \prod_{i=\ell}^{k-1} a_{k, i}}{\prod_{i=\ell}^{k} \kappa_{t-i,k-i}}\;.\label{eq:c1-c2}
\end{align}
It follows that either $|\mA| = 0$ or $\frac{1-\lambda C_1'}{C_2'} \le \frac{|\mA|}{r|X(k)|}$, i.e.
\begin{equation*}
  |x| \,\ge\, \frac{1-\lambda C_1'}{C_2'} r|X(k)|\;.
\end{equation*}
The proof follows using $C'_1\leq C_1$ and $C'_2\leq C_2$, as discussed in Remark~\ref{remark:better-bound}.
\end{proof}

\begin{remark}\label{rem:decoding}
The proof of Proposition~\ref{prop:codistance-t} can be generalized to show that the chain complex has
  small-set co-boundary expansion, (sequential) co-decoder, and co-decoder with syndrome error by modifying Eq.~\eqref{eq:local-code-expansion}.
The rest of the proof follows similarly.

For small-set co-boundary expansion, we want to show that
  for $x$ with small weight,
  $|\delta x| \ge \beta\, dist(x, \im \delta)$ for some $\beta$.
We can reduce the problem to the case of locally co-minimal $x$ in a way simile to the proof of $\mu_\cosyst(i)\geq\dist_\coloc(i)$ in Lemma~\ref{lem:cosys-exp}. Here we want to show that for locally co-minimal $x$ with small weight,
  $|\delta x| \ge \beta\, |x|$.
This follows from the same argument as in the proof of Proposition~\ref{prop:codistance-t}.
The only modification needed is to replace Eq.~\eqref{eq:local-code-expansion} with
\begin{equation}
  \kappa_{t-\ell, k-\ell} n |x(X_{\ge v_\ell}(k))| \,\le\, |x(\Op_{v_\ell}(k))| + |x(\Nb_{v_\ell}(k))| + |y(X_{\ge v_\ell}(k+1))|\;,
\end{equation}
where $y = \delta x$.
The reason is that
\[\kappa_{t-\ell, k-\ell} n |x(X_{\ge v_\ell}(k))| \,\le\, |\delta_{\ol{S}}(x(X_{\geq v_\ell}(k)))| \,\le\, |x(\Op_{v_\ell}(k))| + |x(\Nb_{v_\ell}(k))| + |y(X_{\ge v_\ell}(k+1))|\;.\]
The first inequality follows from the local co-minimality of $x$ as before.
The second inequality is modified because now $y = \delta x \ne 0$;
  therefore, the remaining nonzero entries in $|\delta_{\ol{S}}(x(X_{\geq v_\ell}(k)))|$ can now be canceled from $x(\Op_{v_\ell}(k))$, $x(\Nb_{v_\ell}(k))$, or $|y(X_{\ge v_\ell}(k+1))|$.

For (sequential) co-decoder, we want to show that
  for a syndrome $z = \delta x$ obtained from a (qu)bit error $x$ with small weight,
  one can efficiently find $\tilde{x}$ which is homologous to the actual error,
  $\tilde{x} - x \in \im \delta$.
A natural co-decoder is the generalized small-set flip discussed in \cite{DHLV}.
The main challenge is to show that the generalized small-set flip co-decoder only terminates when $z = 0$
  and the key lemma is to show that a flip that reduces the syndrome always exists when the error has small weight.
To show such lemma, we suppose the flip does not exist, which implies
\begin{equation}\label{eq:co-decoder}
  \kappa_{t-\ell, k-\ell} n |x(X_{\ge v_\ell}(k))| \le 2 (|x(\Op_{v_\ell}(k))| + |x(\Nb_{v_\ell}(k))|)\;.
\end{equation}
By replacing Eq.~\eqref{eq:local-code-expansion} with the inequality above, the same argument in the proof of Proposition~\ref{prop:codistance-t} implies that $x = 0$.

The inequality~\eqref{eq:co-decoder} follows from
\[\kappa_{t-\ell, k-\ell} n |x(X_{\ge v_\ell}(k))| \,\le\, \delta_{\ol{S}}(x(X_{\geq v_\ell}(k))) \le 2 (|x(\Op_{v_\ell}(k))| + |x(\Nb_{v_\ell}(k))|)\;.\]
The first inequality follows from the local co-minimality of $x$ as before.
If the second inequality does not hold, then flipping $x(X_{\geq v_\ell}(k))$ reduces the weight of $z = \delta x$.
In particular, flipping $x(X_{\geq v_\ell}(k))$ only affects $z(X_{\geq v_\ell}(k+1))$.
Additionally,
  $z(X_{\geq v_\ell}(k+1))$ is induced from $\delta_{\ol{S}}(x(X_{\geq v_\ell}(k)))$, $x(\Op_{v_\ell}(k))$, and $x(\Nb_{v_\ell}(k))$.
Therefore, the current syndrome on $X_{\geq v_\ell}(k+1)$ has weight
  \[|z(X_{\geq v_\ell}(k+1))| \ge |\delta_{\ol{S}}(x(X_{\geq v_\ell}(k)))| - (|x(\Op_{v_\ell}(k))| + |x(\Nb_{v_\ell}(k))|)\;,\]
  while the new syndrome after the flip of $x(X_{\geq v_\ell}(k))$ has weight
  \[|z(X_{\geq v_\ell}(k+1)) + \delta_{\ol{S}}(x(X_{\geq v_\ell}(k)))| \le |x(\Op_{v_\ell}(k))| + |x(\Nb_{v_\ell}(k))|\;.\]
This means that if $|\delta_{\ol{S}}(x(X_{\geq v_\ell}(k)))| > 2 (|x(\Op_{v_\ell}(k))| + |x(\Nb_{v_\ell}(k))|)$
  then $|z(X_{\geq v_\ell}(k+1)) + \delta_{\ol{S}}(x(X_{\geq v_\ell}(k)))| < |z(X_{\geq v_\ell}(k+1))|$,
  which implies that the flip exists.
This concludes the proof of Eq.~\eqref{eq:co-decoder}.
More details of the co-decoder can be found in \cite{DHLV}.

For co-decoder with noisy syndrome, we want to show that
  given the noisy syndrome $z = \delta x + y$
    from (qu)bit error $x$ and syndrome error $y$ with small weights,
  one can efficiently find $\tilde{x}$ which is homologous to the actual (qu)bit error,
  $\tilde{x} - x \in \im \delta$.
We again use the generalized small-set flip decoder.
To guarantees a flip, we apply a similar argument but replace Eq.~\eqref{eq:co-decoder} with
\begin{equation}
  \kappa_{t-\ell, k-\ell} n |x(X_{\ge v_\ell}(k))| \le 2 (|x(\Op_{v_\ell}(k))| + |x(\Nb_{v_\ell}(k))| + |y(X_{\ge v_\ell}(k+1))|)
\end{equation}
to account for the syndrome error.
\end{remark}

\section{Distance}
\label{sec:distance}

\newcommand{\dis}{}

For integer $1\leq k \leq t$ we let $\mu_\syst(k)$ be the systolic distance of $C_k(X)$, that is
\[ \mu_\syst(k) \,=\, \min\big\{ |x| \,:\, x\in \ker\partial_k - \im \partial_{k+1}\big\}\;.\]
We also let $\eps_\cyc(k)$ be the cycle expansion,
\[ \eps_\cyc(k) \,=\, \min\Big\{ \frac{|\partial(x)|}{\min_{y\in\ker\partial_{k}}|x-y|} \,:\, x\in C_k(X)-\ker \partial_{k}\Big\}\;.\]

\def\xgb{x^{(0)}_{v}}
\def\zgb{z^{(0)}_{v}}

We show a lower bound on the systolic distance, and the cycle expansion, of $C_k(X)$ by reduction to the co-systolic distance at level $k'=t-k$ of a $t$-dimensional complex $C(\tilde{X})$ which we now define. The complex $C(\tilde{X})$ is defined exactly as $C(X)$, except that the matrices $h_i^T$ and $h_i$ that are used in the definition of the co-boundary and boundary maps for $C(X)$ respectively are replaced by matrices $({h}_i^\perp)^T$ and $h_i^\perp$, where ${h}_i^\perp\in \F_q^{k_i\times n_i}$ is a parity check matrix for the dual code $\cC_i^\perp$ (with linearly independent rows).

\begin{proposition}\label{prop:distance}
  For every $0\leq k'\leq t-1$, let $\mu_\cosyst(k')$ and $\eps_\cocyc(k')$ denote the co-systolic distance and co-cycle expansion of $C^{k'}(\tilde{X})$ respectively.
Then for any $1\leq k \leq t$,
\[ \mu_\syst(k)\,\geq\, \frac{1}{(2nt)^{t}} \cdot \mu_\cosyst(t-k)\;.\]
Furthermore, for $2\leq k \leq t$,
\[ \eps_\cyc(k) \,\geq\, \min\Big\{ \frac{1}{2(t^2 2^{2t}n^{t+1})^{t}}\cdot \eps_\cocyc(t-k)\,,\; \frac{1}{(2nt)^{t}}\cdot \frac{\mu_\cosyst(t-k+1)}{|X(k)|}\Big\}\;.\]
\end{proposition}

The proof of Proposition~\ref{prop:distance} is given in Section~\ref{sec:distance-proof}. Before giving it, we introduce a new sheaf on $X$, and associated linear maps. The sheaf will play an important role in the proof, facilitating a reduction from the distance back to the case of co-distance.

\subsection{Preliminaries}

We begin by introducing a new sheaf $\sheaf_k$ over $X$.

\begin{definition}\label{def:newsheaf}
Let $0\leq k \leq t$. We define a sheaf over $X$, $\sheaf_{k}$, by letting $\sheaf_k({f}) = C_k(X_{\geq f})$ for each $f\in X$. If $\dim(f)>k$ we have $\sheaf_k({f})=\set{0}$.
We define co-restriction maps
in the natural way, by restriction. Namely, for any $f\prec f'$, we let $\corest_{f,f'}:\sheaf_k({f})\to \sheaf_k({f'})$ be defined by
\[\corest_{f,f'}(x) := x|_{X_{\geq f'}(k)}\;. \]

The coboundary map $\Delta_k: C^i(X,\sheaf_k)\to C^{i+1}(X,\sheaf_k)$ is defined as follows.
For $y\in C^i(X,\sheaf_k)$
define $\Delta_k y\in C^{i+1}(X,\sheaf_k)$
by setting, for each $f'\in X(i+1)$
and $u\in X_{\geq f'}(k)$,
\begin{equation}\label{eq:def-deltak} \Delta_k y(f')[u] \, =\, \sum_{f\precdot f'} \corest_{f,f'}(y(f))[u]=\, \sum_{f\precdot f'} y(f)[u]\;,
\end{equation}
where we have used that $\corest_{f,f'}(y(f))[u] = y(f)[u]$ for all $f\prec f'\prec u $.
\end{definition}

Note that we define the (block) norm for $x \in C^i(X, \sheaf_k)$ as usual by $|x| = \sum_{f \in X(i)} |x(f)|$ and $|x(f)| = 1_{x(f) \ne 0}$.

The coefficient space assigned by $\sheaf_k$ to each face of $X$ is itself a space of $k$-chains. One possible way to think about $C^i(X,\sheaf_k)$ is that it is similar to the space $C_k(X)$, except that now each face $u\in X(k)$ may have different values, or ``opinions'', about it that are obtained from different $i$-faces $f\preceq u$.
 The map $\Delta_k$ measures the amount of inconsistency among the local views. In particular,

\begin{claim}\label{claim:cons}
  Let $z\in C^0(X,\sheaf_k)$ be such that $\Delta_{k} z=0$. Then there is $z'\in C_{k}(X)$ such that $z'|_{X_{\ge v}(k)}=z_v$ for all $v\in X(0)$.
    Furthermore, we can choose such a $z'$ which satisfies $|z'| \le 2^t n^t |z|$.
\end{claim}
\begin{proof}
    For each $u\in X(k)$ we choose some vertex $v\prec u$ and define $z'(u) := z_v(u)$. We next show that the definition is independent of the choice of $v$. Let $v,v'$ be vertices such that $v,v'\prec u$. If $e=\set{v,v'}\in X(1)$ then necessarily $z_v(u) = z_{v'}(u)$, since
    \[0=\Delta_{k} z(e)[u] = z_v(u) + z_{v'}(u)\;.\]
    The equality holds true also for any pair $v,v'\prec u$ that are not neighbors, because the graph $(X(0),X(1))$ induced on $\set{v:v\prec u}$ is connected.

    The factor $2^t n^t$ originates from the difference in the definition $|z'| = \sum_{u \in X(k)} |z'(u)|$ and $|z| = \sum_{v \in X(0)} |z(v)|$.
    Indeed,
    \begin{align*}
      |z'|
      = \sum_{u \in X(k)} |z'(u)|
      &= \sum_{v \in X(0)} \sum_{u \in X(k): u \succ v} |z(v)[u]| \\
      &\le \sum_{v \in X(0)} |X_{\ge v}(k)| |z(v)|\\
      &\le 2^t n^t \sum_{v \in X(0)} |z(v)|
      = 2^t n^t |z|
    \end{align*}
    where the last inequality uses that by Lemma~\ref{lem:geom-2}, $X_{\ge v}(k) = {t \choose k} n^{k} \le 2^t n^t$.
\end{proof}

The next lemma verifies that the sequence of spaces $C^i(X,\sheaf_k)$ for $i=0,\ldots,k$, together with the coboundary map $\Delta_k$, forms a co-chain complex.

\begin{lemma}\label{lem:delta-exact}
The map $\Delta_k$ satisfies $\Delta_k \circ \Delta_k =0$. Moreover, $\Delta_k$ is exact on $i$-chains for $0<i\leq k$, i.e. for any $0<i\leq k$ and $y\in C^i(X,\sheaf_k)$ such that $\Delta_k(y)=0$ there is a $z\in C^{i-1}(X,\sheaf_k)$ such that $\Delta_k(z)=y$.
Furthermore, there exists such a $z$ that satisfies $|z|\leq 2^{2t}n^t |y|$.\footnote{In reality, the bound is much tighter $|z|\leq |y|$. The proof is however more complicated, so we adopt the loose bound.}

\end{lemma}

\begin{proof}
To show the lemma,
we first observe that
$\Delta_k$ is the coboundary map of a local chain complex induced from $X_{\leq u}$.
In particular,
\[ C^i(X,\sheaf_k) = \bigoplus_{f\in X(i)}\bigoplus_{u\in X(k), u\succeq f} V_u = \bigoplus_{u\in X(k)} C^i(X_{\leq u}, V_u) \]
where the set $X_{\leq u} = \{f\in X, f\preceq u\}$ is a graded incidence poset, consisting of all faces that have an opinion about the $k$-face $u$.
The map $\Delta_k$ then decomposes as the direct sum of maps $\Delta_{u}:  C^i(X_{\leq u}, V_u) \to C^{i+1}(X_{\leq u}, V_u)$ for each $u\in X(k)$.
This means that to check $\Delta_k \circ \Delta_k = 0$, it suffices to check $\Delta_u \circ \Delta_u = 0$ for each $u\in X(k)$.

The chain complex $C^{*}(X_{\leq u})$ is isomorphic to the following ``hypercube'' construction. Let $S=\type(u)$. For $i\in S$, define a $1$-dimensional chain complex $C^*(L_i)$ over $\F_q$, where $L_i(0)=\{0,1\}$, $L_i(1)=\{e_i\}$ and $\delta:C^0(L_i)\to C^1(L_i)$ is defined by $\delta x(e_1) = x(0) + x(1)\in \F_q$.
Let $C^*(L_S)$ be obtained by taking the homological product of all $C^*(L_i)$ for $i\in S$. Then $C^*(L_S)$ is an $|S|$-dimensional complex, with a single $|S|$-dimensional face that can be seen as a hypercube, the $(|S|-1)$-dimensional faces are the $(|S|-1)$-dimensional faces of the hypercube, etc. Moreover, because $\dim H^0(L_i)=1$ and $\dim H^1(L_i)=0$ then by the K\"unneth formula, $\dim H^0(L_S)=1$ and $\dim H^j(L_S)=0$ for all $0<j\leq |S|$. In particular, $C^*(L_S)$ is exact at all $i>0$. Since each of the maps $\Delta_u$ is isomorphic to a copy of $\delta_{\type(u)}$, it follows that $\Delta_k$ is exact at all levels $i>0$.

Finally, we have the trivial size bound $|z[u]|\leq 2^t |y[u]|$, which follows from $|C^i(L_S)|\leq 2^t$ for all $i$.
Therefore,
\[|z| \le \sum_{u \in X(k)} |z[u]| \leq \sum_{u \in X(k)} 2^t |y[u]| \le \sum_{f \in X(i)} \sum_{u \in X(k): u \succeq f} 2^t |y(f)[u]| \le\sum_{f \in X(i)} 2^t |X_{\ge f} (k)| |y(f)| \le 2^{2t} n^t |y|\]
where the fifth inequality follows from Lemma~\ref{lem:geom-2}, $|X_{\ge f}(k)| = {t-i \choose k-i} n^{k-i}\le 2^t n^t$.
\end{proof}

We have defined the coboundary maps $\Delta_k$. We now move to define boundary maps $\partial_L$ (where the subscript $L$ stands for ``local''). For $S\subseteq\{1,\ldots,t\}$ recall the local chain complex $C(L_S)$ introduced in Section~\ref{sec:loc-chain}, and its boundary map $\partial_S$.
Fix $f\in X(i)$ and let $S=\type(f)$, so that $|S|=i$ and $|\ol S|=t-i$.
By Lemma~\ref{lem:loc-glob}, the space $\sheaf_k(f)=C_k(X_{\geq f})$ is naturally isomorphic to $C_k(L_{\ol{S}})$.
Hence there is a naturally defined map $\partial_L:\sheaf_k(f)\to\sheaf_{k-1}(f)$ given by
\begin{equation}\label{eq:boundary-L}
   \sheaf_k(f) = C_k(X_{\geq f}) \cong C_{k-i}(L_{\ol{S}}) \xrightarrow{\partial_L} C_{k-i-1}(L_{\ol{S}}) \cong C_{k-1}(X_{\geq f})=\sheaf_{k-1}(f)\;.
\end{equation}
This map extends to
$C^i(X,\sheaf_k) \xrightarrow{\partial_L} C^i(X,\sheaf_{k-1})$, $i$-face by $i$-face.

The next lemma shows that the maps $\Delta_k$ and $\partial_L$ commute.

\begin{lemma}\label{lem:partial-delta}
For any $0\leq i < j \leq t$ and $y\in C^i(X,\sheaf_k)$,
\[ \partial_L\circ \Delta_k (y)\,=\, \Delta_{k-1} \circ\partial_L(y)\,\in C^{i+1}(X,\sheaf_{k-1})\;.\]
In other words, the following diagram is commutative:
		\begin{center}
		\begin{tikzcd}[arrows=rightarrow]
    x\in C^{i+1}(X,\sheaf_{k}) \ar[r,"\partial_L"]  & C^{i+1}(X,\sheaf_{k-1})\ni z,z' \\
		 y\in C^i(X,\sheaf_{k}) \ar[u,"\Delta_k"]\ar[r,"\partial_L"]  & C^{i}(X,\sheaf_{k-1})\ni x'  \ar[u,"\Delta_{k-1}"]
				\end{tikzcd}
				\end{center}
\end{lemma}

\begin{proof}
Let $y\in C^i(X,\sheaf_k)$, and define the following chains as in the figure: $x = \Delta_k y \in C^{i+1}(X,\sheaf_k)$, $x'= \partial_L y\in C^i(X,\sheaf_{k-1})$, and let $z=\partial_L x=\partial_L\Delta_k y$ and $z'=\Delta_{k-1}x'=\Delta_{k-1}\partial_L y$ both in $C^{i+1}(X,\sheaf_{k-1})$.  The goal is to prove that $z=z'$.
For any faces $f\prec v$  of dimension $i+1,k-1$ respectively, we have, (using the definition of $\partial_L$ and  \eqref{eq:def-deltak})
\[
z(f)[v] = \partial_L x(f)[v] = \sum_{u\succdot v} x(f)[u] = \sum_{u\succdot v} \Delta_k y(f)[u] = \sum_{u\succdot v}\sum_{g\precdot f} y(g)[u]\;,
\]
which equals
\[
z'(f)[v] = \Delta_{k-1} x'(f)[v] = \sum_{g\precdot f} x'(f)[v] = \sum_{g\precdot f} \partial_L y(g)[v] = \sum_{g\precdot f} \sum_{u\succdot v} y(g)[u]\;.
\]
\end{proof}

\subsection{Proof of Proposition~\ref{prop:distance}}
\label{sec:distance-proof}

Let $1\leq k \leq t$ and $x\in C_k(X)$ be such that $\partial(x)=0$ and $|x|<(2nt)^{-t} \cdot  \mu_\cosyst(t-k)$. For any vertex $v\in X(0)$, let
\[ \xgb\,=\, x(X_{\geq v})\]
be the local view of $x$ restricted to the $k$-faces above the vertex $v$. The set of local views $\set{\xgb}$ can be seen as a cochain  $x^{(0)}\in C^0(X,\sheaf_k)$ where $\sheaf_k$ is introduced in Definition~\ref{def:newsheaf}. By construction
\begin{equation}\label{eq:bound-on-x0}
  |x^{(0)}| \le 2^k|x| \le 2^t|x|\;,
\end{equation}
because each value of $x$ on a $k$-face is copied $2^k$ times, to the face's $2^k$ vertices, to form their local views. (Note that here and as always, $|\cdot|$ denotes the block norm taken in the corresponding coefficient space.)
Furthermore, from the definition it follows that
\begin{equation}\label{eq:delta-x-0}
  \Delta_k(x^{(0)}) \,=\,0\;,
\end{equation}
where $\Delta_k$ is introduced in Definition~\ref{def:newsheaf}.

Using the isomorphism from Lemma~\ref{lem:loc-glob}, $\xgb$ is naturally seen as an element of $C_k(L_{[t]})$. The condition $\partial(x)=0$ implies that for all $v\in X(0)$,
\begin{equation}\label{eq:partial-x-0}
  \partial_{[t]}(\xgb)=0\;,
\end{equation}
where $\partial_{[t]}$ is the ``local'' boundary map defined in Section~\ref{sec:loc-chain}.

If $k=t$ then the two conditions~\eqref{eq:delta-x-0} and~\eqref{eq:partial-x-0} already suffice to conclude the argument. This is done in Claim~\ref{claim:dist-2} below. If $k<t$ then we reduce to this case by identifying an element $x^{(t-k)}$ that satisfies both conditions and is, in some sense, a ``lift'' of $x^{(0)}$ to $C^{t-k}(X,\sheaf_t)$. The details follow.

Assume $k<t$. Using exactness for $\partial_{[t]}$ (Lemma~\ref{lem:exactness}),~\eqref{eq:partial-x-0} implies that there is a $z^{(0)}= \{\zgb\}_{v}$ where $\zgb \in C_{k+1}(L_{[t]})$ for all vertices $v$, such that
\begin{equation}\label{eq:partial-y-x}
 \partial_{[t]}(\zgb)\,=\,\xgb\;.
\end{equation}
Furthermore, if $\xgb=0$ we take $\zgb=0$. Note that $z^{(0)}$ is an element of $C^0(X,\sheaf_{k+1})$, and~\eqref{eq:partial-y-x} can be rewritten as
\begin{equation}\label{eq:partial-y-x-2}
  \partial_{L}(z^{(0)})\,=\,x^{(0)}\;,
 \end{equation}
 where $\partial_L$ is the map introduced around~\eqref{eq:boundary-L}.
 We observe the size bound
\begin{equation}\label{eq:partial-y-x-1}
  |z^{(0)}| \,=\, \sum_{v\in X(0)} |z^{(0)}_{v}| \,\leq\,\sum_{v\in X(0)} |\xgb| \,=\, |x^{(0)}|\;.
\end{equation}

If all $z_v^{(0)}$ are consistent, meaning that whenever $v,v' \prec u$ then $\zgb(u) = z_{v'}^{(0)}(u)$, then the $\set{\zgb}$ can be ``stitched'' together into one $z\in C_{k+1}(X)$ such that for every $v\in X(0)$,  $z|_{X_{\geq v}}=\zgb$. More precisely, by Claim \ref{claim:cons}, if $\Delta_{k+1}(z^{(0)})=0$ then there is a  $z\in C_{k+1}(X)$ such that for every $v\in X(0)$, $z_{|X_{\geq v}(k+1)}=z^{(0)}_{v}$, and moreover $z$ satisfies $\partial(z)=x$ with $|z|\le 2^t n^t |z^{(0)}|$ which implies the size bound $|z|\le 2^t n^t |z^{(0)}|\leq 2^t n^t |x^{(0)}| \le 2^{2t} n^t |x|$ using~\eqref{eq:bound-on-x0} and \eqref{eq:partial-y-x-1}.

However, the different $\zgb$ may be inconsistent with each other. In such a case we will have to look into the differences between local views along edges of $X$, and then continue to move into higher and higher dimensions. Towards this we make the following definition.

\begin{definition}\label{def:xz-2}
  Let $0\leq r\leq t-k-1$.
  A sequence $z^{(0)},\ldots,z^{(r)}$, where for each $i\in\{1,\ldots,r\}$, $z^{(i)}\in C^{i}(X,\sheaf_{k+i+1})$, is said to \emph{explain} $x^{(0)}$ if the following hold. There are $x^{(1)},\ldots,x^{(r+1)}$ such that $x^{(i)} \in C^{i}(X,\sheaf_{k+i})$, $x^{(r+1)} = 0$,  and for all $i=0,\ldots,r-1$,
  \begin{equation}\label{eq:ind-2}
      \Delta_{k+i+1}(z^{(i)}) = x^{(i+1)}, \qquad
      x^{(i)} = \partial_L(z^{(i)})\;.
  \end{equation}
  Furthermore, we require the size bounds
  \begin{equation}\label{eq:ddz-2}
    |z^{(i)}| \,\leq\, |x^{(i)}|, \qquad
    |x^{(i+1)}| \,\leq\, nt |z^{(i)}|\;.
  \end{equation}
  for all $i=0,\ldots,r$.
  \end{definition}

  Figure~\ref{fig:xz-2} gives an illustration of the conditions given in Definition~\ref{def:xz-2}. We first show that having such a sequence lets us find a small chain $z$ whose boundary $\partial(z)=x$.


    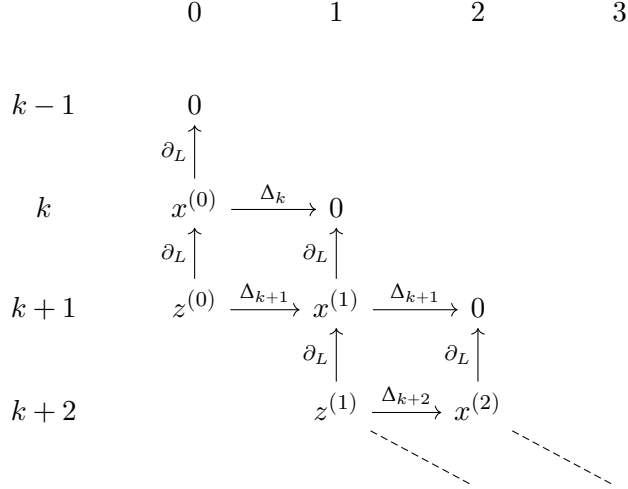
\begin{figure}\begin{center}
\begin{tikzcd}[arrows=rightarrow]
&0&1&2&3&\\
k-1 & 0 &&&&\\
k&     x^{(0)} \ar[r,"\Delta_k"] \ar[u,"\partial_L"] & 0 &  &&\\
k+1&     z^{(0)} \ar[u,"\partial_L"]\ar[r,"\Delta_{k+1}"]& x^{(1)} \ar[u,"\partial_L"]\ar[r,"\Delta_{k+1}"] & 0& &\\
k+2&     & z^{(1)} \ar[rd, dashed, rightarrow, no head] \ar[u,"\partial_L"]\ar[r,"\Delta_{k+2}"]& x^{(2)}\ar[rd, dashed, rightarrow, no head] \ar[u,"\partial_L"]&&\\
 & & & ~& \makebox[\widthof{$x^{(3)}$}]{~}&~
\end{tikzcd}\end{center}
\caption{An illustration of the requirements on the sequences $x^{(i)}$ and $z^{(i)}$.}\label{fig:xz-2}
\end{figure}

\begin{claim}\label{claim:dist-1}
Suppose that $1\le k<t$ and there is some $r\leq t-k-1$ such that $z^{(0)},\ldots,z^{(r)}$ explains $x^{(0)}$. Then there is some $z\in C_{k+1}(X)$ such that $\partial(z)=x$ and $|z|\le (t 2^{2t}n^{t+1})^t |x|$.
\end{claim}

\begin{proof}
The special case where $r=0$ was handled directly above, relying on Claim \ref{claim:cons}.
We reduce the general case to this case. For general $r>0$, $x^{(r+1)}=\Delta_{k+r+1}(z^{(r)})$ by~\eqref{eq:ind-2}, so $x^{(r+1)}=0$ implies (by exactness, Lemma~\ref{lem:delta-exact}) that there is an $u^{(r-1)}\in C^{r-1}(X,\sheaf_{k+r+1})$ such that $\Delta_{k+r+1}(u^{(r-1)})=z^{(r)}$. Furthermore, we can choose $u^{(r-1)}$ in such a way that $|u^{(r-1)}| \leq 2^{2t} n^t |z^{(r)}|$.

Let $w^{(r-1)}=\partial_L(u^{(r-1)})$. Then $w^{(r-1)}\in C^{r-1}(X,\sheaf_{k+r})$, and $|w^{(r-1)}|\leq |u^{(r-1)}|$ (because both vectors are labeled by faces in $X(r-1)$). Furthermore,
\begin{align*}
\Delta_{k+r}(w^{(r-1)}) &= \Delta_{k+r}(\partial_L(u^{(r-1)}))\\
&= \partial_L(\Delta_{k+r+1}(u^{(r-1)}))\\
&= \partial_L(z^{(r)})\\
&=x^{(r)}\;,
\end{align*}
where the second equality is by Lemma~\ref{lem:partial-delta} and the last is by~\eqref{eq:ind-2}. Now let $\tilde{z}^{(r-1)}=z^{(r-1)}+w^{(r-1)}$. We easily verify that this element satisfies $\Delta_{k+r+1}(\tilde{z}^{(r-1)})=0$ and $\partial_L(\tilde{z}^{(r-1)})=x^{(r-1)}$.
This means that $z^{(0)},\ldots,z^{(r-2)},\tilde z^{(r-1)}$ explains $x^{(0)}$. We went from the assumption that $x^{(r+1)}=0$ to the same situation, but now $x^{(r)}=0$, and furthermore
\begin{equation}\label{eq:tildez-b1}
  |\tilde{z}^{(r-1)}| \,\leq\, |z^{(r-1)}| + 2^{2t}n^t |z^{(r)}| \,\leq\, (1+nt2^{2t}n^t)|z^{(r-1)}|\;,
\end{equation}
where the second inequality uses~\eqref{eq:ddz-2}. Iterating, we reduce the case of general $r$ to the case where $r=0$, which as shown above (see Claim~\ref{claim:cons}) completes the proof of Proposition~\ref{prop:distance}.
The weight of $\tilde{z}^{(0)}$ before applying Claim~\ref{claim:cons} satisfies
\[|\tilde{z}^{(0)}| \le \Big(1 + (nt 2^{2t}n^t) + ... + (nt 2^{2t}n^t)^r\Big) |x^{(0)}| \le (nt 2^{2t}n^t)^{r+1} |x^{(0)}| \le (t 2^{2t}n^{t+1})^{t-1} |x^{(0)}|\]
where the second inequality follows from $2 \le nt 2^{2t}n^t$ and the third inequality follows from $r \le t - k - 1 \le t - 2$.
Thus,
\[|z| \le 2^t n^t |\tilde z^{(0)}| \le 2^t n^t (t 2^{2t}n^{t+1})^{t-1} |x^{(0)}| \le 2^t 2^t n^t (t 2^{2t}n^{t+1})^{t-1} |x| \le (t 2^{2t}n^{t+1})^t |x|\;.\]
\end{proof}

Now we show how to construct a sequence that satisfies the conditions of Definition~\ref{def:xz-2}. We do this inductively. Suppose that $z^{(j)}$ and $x^{(j)}$ have been defined for all $0\leq j<r$, such that they satisfy both conditions in the definition. Suppose first that $r<t-k$. We show how to extend the sequence.

Let $x^{(r)}=\Delta_{k+r}(z^{(r-1)})$. We wish to define $z^{(r)}$. By assumption, $\partial_L(z^{(r-1)})=x^{(r-1)}$, which implies
\begin{align}
 \partial_L(x^{(r)})&=\partial_L\circ \Delta_{k+r}(z^{(r-1)})\notag \\
&=\Delta_{k+r-1}\circ \partial_L(z^{(r-1)})\notag \\
&= \Delta_{k+r-1}(x^{(r-1)})\notag\\
&=0\;,\label{eq:prev-1}
\end{align}
where the first equality is because $x^{(r-1)}=\Delta_{k+r}(z^{(r-2)})$, the second equality is by Lemma~\ref{lem:partial-delta} and the last because $\Delta^2=0$ by Lemma~\ref{lem:delta-exact}.
Therefore, using the exactness condition for the local complex we find $z^{(r)}\in C^r(X,\sheaf_{k+r+1})$ such that $\partial_L (z^{(r)})=x^{(r)}$, as desired.

Moreover, because $|z^{(r)}|$ and $|x^{(r)}|$ are labeled by faces in $X(r)$ we have $|z^{(r)}| \leq |x^{(r)}|$.
And because each nonzero face $f \in X(r)$ in $z^{(r)}$
  results in at most $|\up(f)|$ nonzero faces in $x^{(r+1)} = \Delta_k z^{(r)}$,
$|x^{(r+1)}| \leq |\up(f)| |z^{(r)}| \leq nt |z^{(r)}|$
where we apply Lemma~\ref{lem:geom-2}, $|\up(f)|= (t-i) n \le nt$.

It remains to treat the case where we have defined $x^{(0)},\ldots,x^{(t-k)}$, but $x^{(t-k)}\neq 0$. This is handled by our last claim.

\begin{claim}\label{claim:dist-2}
Suppose that $x^{(t-k)}\neq 0$ and $|x^{(t-k)}| < \mu_\cosyst(t-k)$. Then there is some $v\in C_{k+1}(X)$ such that $\partial(v)=x$ and $|v|\leq 2(t^22^{2t}n^{t+1})^t \frac{1}{\eps_\cocyc(t-k-1)} |x|$.
\end{claim}

\begin{proof}
If $k=t$, then we have that $\Delta_t(x^{(0)})=0$ as shown in~\eqref{eq:delta-x-0}, and $\partial_L(x^{(0)})=0$ which follows from~\eqref{eq:partial-x-0}.
If $k<t$, then both equalities hold as well, as we now show. Using $x^{(t-k)} = \Delta_t(z^{(t-k-1)})$ and $\Delta_t\circ \Delta_t=0$ by Lemma~\ref{lem:delta-exact}, we get that
\begin{equation}\label{eq:dist-2-1}
\Delta_t(x^{(t-k)})\,=\,0\;.
\end{equation}
Moreover, using a similar chain of equalities as for~\eqref{eq:prev-1},
\begin{align*}
  \partial_L(x^{(t-k)})&=\partial_L\circ \Delta_{t}(z^{(t-k-1)})\\
 &=\Delta_{t-1}\circ \partial_L(z^{(t-k-1)})\\
 &= \Delta_{t-1}(x^{(t-k-1)})\\
 &=0\;.
 \end{align*}
Note that for each $f\in X(t-k)$, $x^{(t-k)}(f) \in \F_q^{\prod_{j\notin \type(f)} A_j}$. 
So $\partial_L(x^{(t-k)})=0$ means that $x^{(t-k)}(f)\in \otimes_{j\notin\type(f)}\code_j$ by Lemma~\ref{lem:tensor-code}. In particular, by replacing the tensor codeword $x^{(t-k)}(f)$ by a decoding $\tilde{x}^{(t-k)}(f) \in \F_q^{\prod_{j\notin\type(f)} k_j}$ of it, we obtain a $(t-k)$-cochain $\tilde{x}^{(t-k)}\in C^{t-k}(\tilde{X})$ such that, for any ${f}\in \tilde{X}(t-k+1)$,
\begin{align*}
\Big(\bigotimes_{j\notin \type(f)} (h_j^\perp)^T\Big) \tilde{\delta}(\tilde{x}^{(t-k)})({f})
&= \Big(\bigotimes_{j\notin \type(f)} (h_j^\perp)^T\Big) \Big(\sum_{{f}'\precdot {f}} \widetilde{\corest}_{{f}',{f}}(\tilde{x}^{(t-k)}({f}'))\Big)\\
&= \sum_{{f}'\precdot {f}}\Big(\Big(\bigotimes_{j\notin \type(f')} (h_j^\perp)^T\Big)  \tilde{x}^{(t-k)}({f}')\Big)\\
&= \sum_{{f}'\precdot {f}} x^{(t-k)}({f}')\\
&= \Delta_t x^{(t-k)} (f)\\
&=0\;.
\end{align*}
Here, the first equality is by definition of $\tilde{\delta}$, where $\widetilde{\corest}$ is defined as $\corest$ in~\eqref{eq:co-rest-1} but with the map $h_i^T$ replaced by $(h_i^\perp)^T$. The second equality is by definition of $\widetilde{\corest}$. The third equality is by the definition of $(h_i^\perp)^T$, which re-encodes $\tilde{x}^{(t-k)}$ in $x^{(t-k)}$. The fourth equality is by the definition of $\Delta_t$ in~\eqref{eq:def-deltak}. The last equality is by~\eqref{eq:dist-2-1}.
Finally, because $(h_i^\perp)^T$ are injective, this implies $\tilde{\delta}(\tilde{x}^{(t-k)})({f}) = 0$ for all $f$ which means $\tilde{\delta}(\tilde{x}^{(t-k)}) = 0$.

Because $\tilde{\delta}(\tilde{x}^{(t-k)}) = 0$, as long as $|\tilde{x}^{(t-k)}| < \mu_\cosyst(t-k)$, $\tilde{x}^{(t-k)} \in \im\,\tilde{\delta}$, and we can find an $\tilde{u}^{(t-k-1)} \in C^{t-k-1}(\tilde{X})$ such that $\tilde{\delta}(\tilde{u}^{(t-k-1)})=\tilde{x}^{(t-k)}$ and furthermore
\[ |\tilde{u}^{(t-k-1)}|\leq \frac{1}{\eps_\cocyc(t-k-1)}|\tilde{x}^{(t-k)}|=\frac{1}{\eps_\cocyc(t-k-1)}|{x}^{(t-k)}|\]
where the last equality follows from the construction of $\tilde{x}^{(t-k)}$, which is the decoding of ${x}^{(t-k)}$.
Note that $|\tilde{x}^{(t-k)}| <  \mu_\cosyst(t-k)$ is satisfied when $|x| < (2nt)^{-t} \mu_\cosyst(t-k)$, because $|\tilde{x}^{(t-k)}| = |x^{(t-k)}| \le (nt)^{t-k}|x^{(0)}| \le (nt)^{t-k}2^t |x|$ where the last inequality follows from~\eqref{eq:bound-on-x0}.

By re-encoding $\tilde{u}$, i.e.\ applying $\otimes_{j\notin\type(f)} \tilde{h}_j^T$ to each $\tilde{u}^{(t-k-1)}(f)$, we obtain an element $u^{(t-k-1)}\in C^{t-k-1}(X,\mF_t)$ such that $\Delta_{t}(u^{(t-k-1)})=x^{(t-k)}$ (because $\tilde{\delta}(\tilde{u}^{(t-k-1)})=\tilde{x}^{(t-k)}$) and $\partial_L(u^{(t-k-1)})=0$ (because $u^{(t-k-1)}$ is obtained by applying $\otimes_{j\notin\type(f)} \tilde{h}_j^T$). From there we can reduce the problem to Claim~\ref{claim:dist-1}, by setting $\tilde{z}^{(t-k-1)}=z^{(t-k-1)}+u^{(t-k-1)}$ and observing that it satisfies the conditions of the claim, where $\tilde{x}^{(t-k)}=\Delta_t (\tilde{z}^{(t-k-1)})=0$ and $\partial_L (\tilde{z}^{(t-k-1)})=x^{(t-k-1)}$.

Similar to Claim~\ref{claim:dist-1}, for $r = t-k-1$ we have
\begin{align*}
|\tilde{z}^{(r-1)}|
&\le |z^{(r-1)}| + 2^{2t}n^t|\tilde{z}^{(r)}| \\
&\le |z^{(r-1)}| + 2^{2t}n^t (|z^{(r)}|+|u^{(r)}|) \\
&\le (1+nt 2^{2t}n^t)|z^{(r-1)}| + \frac{2^{2t}n^t}{\eps_\cocyc(t-k)(r)}|x^{(r+1)}| \\
&\le (1+nt 2^{2t}n^t)|z^{(r-1)}| + \frac{2^{2t}n^t}{\eps_\cocyc(r)} (nt)^t 2^t |x|\;,
\end{align*}
and
\begin{align*}
|\tilde{z}^{(0)}|
\le (t2^{2t}n^{t+1})^{t-1} |x^{(0)}| + \frac{(t2^{2t}n^{t+1})^{t-2}}{\eps_\cocyc(r)} (nt)^t 2^t |x|\;.
\end{align*}
Finally,
\begin{align*}
|z|
\le 2^t n^t |\tilde{z}^{(0)}|
\le (t2^{2t}n^{t+1})^t|x| +
\frac{(t2^{2t}n^{t+1})^{t-2}}{\eps_\cocyc(r)} t^t n^{2t} 2^{2t} |x|
& \le (t^22^{2t}n^{t+1})^t \left(1+\frac{1}{\eps_\cocyc(r)}\right) |x| \\
&\le \frac{2(t^22^{2t}n^{t+1})^t}{\eps_\cocyc(r)} |x|\;.
\end{align*}
\end{proof}

We now conclude. First, we show the lower bound on $\mu_\syst(k)$. It suffices to show that whenever $x\in C_k(X)$ satisfies $\partial(x)=0$ and $|x|\leq (2nt)^{-t}\mu_\cosyst(t-k)$, there is $z\in C_{k+1}(X)$ such that $\partial(z)=x$. If $k=t$ this follows immediately from Claim~\ref{claim:dist-2} and the bound~\eqref{eq:bound-on-x0}. If $k<t$, we construct a sequence $z^{(0)},\ldots,z^{(r)}$ that explains $x^{(0)}$. If such a sequence is found with $r\leq t-k-1$ then Claim~\ref{claim:dist-1} gives the desired conclusion. If not, then our assumption on $|x|$ together with~\eqref{eq:ddz-2} shows that the assumption of Claim~\ref{claim:dist-2} is satisfied, allowing us to conclude the argument.

Finally, we show the ``Furthermore'' part of Proposition~\ref{prop:distance}. Let $x\in C_k(X)$ be the vector with the smallest cycle expansion where
$x':=\partial(x)\neq 0$,
\begin{equation}\label{eq:min-eq}
  |x|\,=\,\min_{y\in \ker \partial_{k}}|x-y|
\end{equation}
 and $|\partial(x)| = \eps_\cyc(k)|x|$. If $|x'|\geq (2nt)^{-t} \mu_\cosyst(t-(k-1))$ then because $|x|\leq |X(k)|$ we deduce $\eps_\cyc(k)\geq (2nt)^{-t}\mu_\cosyst(t-(k-1))/|X(k)|$. If $|x'|< (2nt)^{-t} \mu_\cosyst(t-(k-1))$, then following the arguments above, applied to $x'$, which satisfies $\partial(x')=0$, we find a $z'\in C_{k}(X)$ such that $\partial(z)=x'$ and $|z'|\leq  2(t^22^{2t}n^{t+1})^t \frac{1}{\eps_\cocyc(t-k)} |x'|$. By~\eqref{eq:min-eq}, $|z'|\geq |x|$, which concludes the proof.

\newcommand{\etalchar}[1]{$^{#1}$}


\begin{thebibliography}{GTC{\etalchar{+}}23}

\bibitem[AAV13]{aharonov2013guest}
Dorit Aharonov, Itai Arad, and Thomas Vidick.
\newblock Guest column: the quantum {PCP} conjecture.
\newblock {\em Acm sigact news}, 44(2):47--79, 2013.

\bibitem[ABN23]{anshu2023nlts}
Anurag Anshu, Nikolas Breuckmann, and Chinmay Nirkhe.
\newblock {NLTS} hamiltonians from good quantum codes.
\newblock In {\em Proceedings of the 55th Annual ACM Symposium on Theory of
  Computing}, pages 1090--1096, 2023.

\bibitem[AE15]{aharonov2015quantum}
Dorit Aharonov and Lior Eldar.
\newblock Quantum locally testable codes.
\newblock {\em SIAM Journal on Computing}, 44(5):1230--1262, 2015.

\bibitem[AGHP92]{AlGoHaPe}
N.~Alon, O.~Goldreich, J.~H\aa{}stad, and R.~Peralta.
\newblock Simple constructions of almost k-wise independent random variables.
\newblock {\em Random Structures and Algorithms}, 3, 1992.
\newblock See also Addendum in {\em Random Structures and Algorithms} 4 pages
  ?? 1993.

\bibitem[AR94]{alon1994random}
Noga Alon and Yuval Roichman.
\newblock Random {C}ayley graphs and expanders.
\newblock {\em Random Structures \& Algorithms}, 5(2):271--284, 1994.

\bibitem[BE21]{BE}
Nikolas~P. Breuckmann and Jens~N. Eberhardt.
\newblock Balanced product quantum codes.
\newblock {\em IEEE Transactions on Information Theory}, 67(10):6653--6674,
  2021.

\bibitem[BH14]{bravyi2014homological}
Sergey Bravyi and Matthew~B Hastings.
\newblock Homological product codes.
\newblock In {\em Proceedings of the forty-sixth annual ACM symposium on Theory
  of computing}, pages 273--282, 2014.

\bibitem[CS96]{calderbank1996good}
A~Robert Calderbank and Peter~W Shor.
\newblock Good quantum error-correcting codes exist.
\newblock {\em Physical Review A}, 54(2):1098, 1996.

\bibitem[DEL{\etalchar{+}}22]{DELLM}
Irit Dinur, Shai Evra, Ron Livne, Alexander Lubotzky, and Shahar Mozes.
\newblock Good locally testable codes.
\newblock {\em arXiv preprint arXiv:2207.11929}, 2022.

\bibitem[DHLV23]{DHLV}
Irit Dinur, Min-Hsiu Hsieh, Ting-Chun Lin, and Thomas Vidick.
\newblock Good quantum {LDPC} codes with linear time decoders.
\newblock In {\em Proceedings of the 55th Annual ACM Symposium on Theory of
  Computing}, pages 905--918, 2023.

\bibitem[EH17]{eldar2017local}
Lior Eldar and Aram~W Harrow.
\newblock Local {H}amiltonians whose ground states are hard to approximate.
\newblock In {\em 2017 IEEE 58th annual symposium on foundations of computer
  science (FOCS)}, pages 427--438. IEEE, 2017.

\bibitem[EK16]{EvraK16}
Shai Evra and Tali Kaufman.
\newblock Bounded degree cosystolic expanders of every dimension.
\newblock In {\em Proceedings of the 48th Annual {ACM} {SIGACT} Symposium on
  Theory of Computing, {STOC} 2016, Cambridge, MA, USA, June 18-21, 2016},
  pages 36--48, 2016.

\bibitem[EKZ22]{evra2022decodable}
Shai Evra, Tali Kaufman, and Gilles Z{\'e}mor.
\newblock Decodable quantum {LDPC} codes beyond the n distance barrier using
  high-dimensional expanders.
\newblock {\em SIAM Journal on Computing}, (0):FOCS20--276, 2022.

\bibitem[FH21]{freedman2021building}
Michael Freedman and Matthew Hastings.
\newblock Building manifolds from quantum codes.
\newblock {\em Geometric and Functional Analysis}, 31(4):855--894, 2021.

\bibitem[FK22]{first2022good}
Uriya~A First and Tali Kaufman.
\newblock On good $2 $-query locally testable codes from sheaves on high
  dimensional expanders.
\newblock {\em arXiv preprint arXiv:2208.01778}, 2022.

\bibitem[Got14]{gottesman2014fault}
Daniel Gottesman.
\newblock Fault-tolerant quantum computation with constant overhead.
\newblock {\em Quantum Information \& Computation}, 14(15-16):1338--1372, 2014.

\bibitem[Gro10]{Gromov2010}
Mikhail Gromov.
\newblock Singularities, expanders and topology of maps. part 2: from
  combinatorics to topology via algebraic isoperimetry.
\newblock {\em Geometric and Functional Analysis}, 20(2):416--526, 2010.

\bibitem[GTC{\etalchar{+}}23]{gu2023single}
Shouzhen Gu, Eugene Tang, Libor Caha, Shin~Ho Choe, Zhiyang He, and Aleksander
  Kubica.
\newblock Single-shot decoding of good quantum {LDPC} codes.
\newblock {\em arXiv preprint arXiv:2306.12470}, 2023.

\bibitem[Has16]{hastings2016quantum}
Matthew~B Hastings.
\newblock Quantum codes from high-dimensional manifolds.
\newblock {\em arXiv preprint arXiv:1608.05089}, 2016.

\bibitem[Has17]{hastings2017quantum}
Matthew~B Hastings.
\newblock Quantum codes from high-dimensional manifolds.
\newblock In {\em 8th Innovations in Theoretical Computer Science Conference
  (ITCS 2017)}. Schloss Dagstuhl-Leibniz-Zentrum fuer Informatik, 2017.

\bibitem[HHO21]{hastings2021fiber}
Matthew~B Hastings, Jeongwan Haah, and Ryan O'Donnell.
\newblock Fiber bundle codes: breaking the $n^{1/2}\poly\log(n)$ barrier for
  quantum {LDPC} codes.
\newblock In {\em Proceedings of the 53rd Annual ACM SIGACT Symposium on Theory
  of Computing}, pages 1276--1288, 2021.

\bibitem[JL99]{JL}
Bruce Jordan and Ron Livne.
\newblock The {R}amanujan property for regular cubical complexes.
\newblock {\em Duke Mathematical Journal}, 105:85--103, 1999.

\bibitem[JMO{\etalchar{+}}22]{JeronimoMO0T22}
Fernando~Granha Jeronimo, Tushant Mittal, Ryan O'Donnell, Pedro Paredes, and
  Madhur Tulsiani.
\newblock Explicit abelian lifts and quantum {LDPC} codes.
\newblock In Mark Braverman, editor, {\em 13th Innovations in Theoretical
  Computer Science Conference, {ITCS} 2022, January 31 - February 3, 2022,
  Berkeley, CA, {USA}}, volume 215 of {\em LIPIcs}, pages 88:1--88:21. Schloss
  Dagstuhl - Leibniz-Zentrum f{\"{u}}r Informatik, 2022.

\bibitem[Kal23]{kalachev2023prexp-vs-agreement}
Gleb Kalachev.
\newblock High-dimensional expansion of product codes is stronger than robust
  and agreement testability, 2023.

\bibitem[KKL14]{KKL14}
Tali Kaufman, David Kazhdan, and Alexander Lubotzky.
\newblock Ramanujan complexes and bounded degree topological expanders.
\newblock In {\em 55th {IEEE} Annual Symposium on Foundations of Computer
  Science, {FOCS} 2014, Philadelphia, PA, USA, October 18-21, 2014}, pages
  484--493, 2014.

\bibitem[KKL16]{kaufman2016isoperimetric}
Tali Kaufman, David Kazhdan, and Alexander Lubotzky.
\newblock Isoperimetric inequalities for ramanujan complexes and topological
  expanders.
\newblock {\em Geometric and Functional Analysis}, 26(1):250--287, 2016.

\bibitem[KP22]{kalachev2022two}
Gleb Kalachev and Pavel Panteleev.
\newblock Two-sided robustly testable codes.
\newblock {\em arXiv preprint arXiv:2206.09973}, 2022.

\bibitem[KP24]{PK24}
G.~Kalachev and P.~Panteleev.
\newblock Maximally extendable product codes are good coboundary expanders,
  2024.

\bibitem[KT21]{kaufman2021new}
Tali Kaufman and Ran~J Tessler.
\newblock New cosystolic expanders from tensors imply explicit quantum ldpc
  codes with {$\Omega(\sqrt n \log^k n)$} distance.
\newblock In {\em Proceedings of the 53rd Annual ACM SIGACT Symposium on Theory
  of Computing}, pages 1317--1329, 2021.

\bibitem[LH22]{LinH22}
Ting{-}Chun Lin and Min{-}Hsiu Hsieh.
\newblock c\({}^{\mbox{3}}\)-locally testable codes from lossless expanders.
\newblock In {\em {IEEE} International Symposium on Information Theory, {ISIT}
  2022, Espoo, Finland, June 26 - July 1, 2022}, pages 1175--1180. {IEEE},
  2022.

\bibitem[Lin24]{lin2024transversal}
Ting-Chun Lin.
\newblock Transversal non-Clifford gates for quantum LDPC codes on sheaves.
\newblock {\em arXiv preprint arXiv:2410.14631}, 2024.

\bibitem[LLZ22]{leverrier2022towards}
Anthony Leverrier, Vivien Londe, and Gilles Z{\'e}mor.
\newblock Towards local testability for quantum coding.
\newblock {\em Quantum}, 6:661, 2022.

\bibitem[LM06]{LinialM06}
Nathan Linial and Roy Meshulam.
\newblock Homological connectivity of random 2-complexes.
\newblock {\em Combinatorica}, 26(4):475--487, 2006.

\bibitem[Lub18]{lubotzky2018high}
Alexander Lubotzky.
\newblock High dimensional expanders.
\newblock In {\em Proceedings of the International Congress of Mathematicians:
  Rio de Janeiro 2018}, pages 705--730. World Scientific, 2018.

\bibitem[LZ22]{leverrier2022quantum}
Anthony Leverrier and Gilles Z{\'e}mor.
\newblock Quantum tanner codes.
\newblock In {\em 2022 IEEE 63rd Annual Symposium on Foundations of Computer
  Science (FOCS)}, pages 872--883. IEEE, 2022.

\bibitem[PK21]{PK1}
Pavel Panteleev and Gleb Kalachev.
\newblock Quantum {LDPC} codes with almost linear minimum distance.
\newblock {\em IEEE Transactions on Information Theory}, pages 1--1, 2021.

\bibitem[PK22]{PK2}
Pavel Panteleev and Gleb Kalachev.
\newblock Asymptotically good quantum and locally testable classical {LDPC}
  codes.
\newblock In {\em Proceedings of the 54th Annual ACM SIGACT Symposium on Theory
  of Computing}, pages 375--388, 2022.

\bibitem[PK24]{PK-future}
Pavel Panteleev and Gleb Kalachev.
\newblock Personal communication, 2024.

\bibitem[Por23]{portnoy2023local}
Elia Portnoy.
\newblock Local quantum codes from subdivided manifolds.
\newblock {\em arXiv preprint arXiv:2303.06755}, 2023.

\bibitem[SS96]{SipserSp96}
Michael Sipser and Daniel~A. Spielman.
\newblock Expander codes.
\newblock {\em IEEE Trans. Inform. Theory}, 42(6, part 1):1710--1722, 1996.
\newblock Codes and complexity.

\bibitem[Ste96]{steane1996error}
Andrew~M Steane.
\newblock Error correcting codes in quantum theory.
\newblock {\em Physical Review Letters}, 77(5):793, 1996.

\bibitem[WLH23a]{wills2023general}
Adam Wills, Ting-Chun Lin, and Min-Hsiu Hsieh.
\newblock General distance balancing for quantum locally testable codes.
\newblock {\em arXiv preprint arXiv:2305.00689}, 2023.

\bibitem[WLH23b]{wills2023tradeoff}
Adam Wills, Ting-Chun Lin, and Min-Hsiu Hsieh.
\newblock Tradeoff constructions for quantum locally testable codes.
\newblock {\em arXiv preprint arXiv:2309.05541}, 2023.

\end{thebibliography}
\end{document}